\documentclass[11pt,letterpaper]{article}
\usepackage[utf8]{inputenc}
\usepackage[T1]{fontenc}
\usepackage{geometry}
\usepackage{amsmath,amssymb,amsthm,amsfonts,mathtools,mathrsfs}
\usepackage{graphicx}
\usepackage{fullpage}
\usepackage{setspace}
\usepackage{verbatim}
\usepackage{indentfirst}
\usepackage{caption,subcaption}
\usepackage{threeparttable}
\usepackage{booktabs}
\usepackage{multirow}
\usepackage{longtable}
\usepackage[dvipsnames]{xcolor}
\usepackage[linesnumbered,ruled,vlined]{algorithm2e}
\usepackage{xspace,makecell,tikz}
\usepackage{siunitx}
\usepackage{float}
\usepackage[sort,numbers]{natbib}
\usepackage{hyperref}
\usepackage[capitalize]{cleveref}
\usepackage{aliascnt}
\usepackage{bbm}
\usepackage{enumerate}
\usepackage[shortlabels]{enumitem}
\usepackage{tocloft}
\usepackage{comment}
\usepackage[toc,page]{appendix}
\usepackage{subfiles}
\numberwithin{equation}{section}
\hypersetup{
hypertexnames=false,
linktocpage=true,
colorlinks=true,
linkcolor=red,
filecolor=magenta,
urlcolor=cyan,
citecolor=blue,
}

\usetikzlibrary{arrows.meta, positioning,calc, shapes.geometric}
\usetikzlibrary{decorations.pathreplacing}
\tikzset{thickline/.style={line width=2pt}}

\newcommand{\E}{\mathbb{E}}

\newcommand{\Z}{\mathbb{Z}}

\newcommand{\Id}{\mathrm{Id}}

\newcommand{\calA}{\mathcal{A}}

\newcommand{\calC}{\mathcal{C}}
\newcommand{\calD}{\mathcal{D}}

\newcommand{\calL}{\mathcal{L}}

\newcommand{\calN}{\mathcal{N}}

\newcommand{\poly}{\mathrm{poly}}

\newcommand{\dTV}{d_{\mathrm{TV}}}

\newtheorem*{theorem*}{Theorem}
\newtheorem*{lemma*}{Lemma}
\newenvironment{restatement}[1]
  {\par\medskip\noindent\textbf{#1 (restated).}\ \itshape}
  {\par\medskip}

\newtheorem{theorem}{Theorem}[subsection] 

\makeatletter
\newcommand{\sharetheorem}[2]{%
  \newaliascnt{#1}{theorem}%      
  \newtheorem{#1}[#1]{#2}%         
  \aliascntresetthe{#1}%           
}
\makeatother

\sharetheorem{lemma}{Lemma}
\sharetheorem{observation}{Observation}
\sharetheorem{claim}{Claim}
\sharetheorem{proposition}{Proposition}
\sharetheorem{corollary}{Corollary}
\sharetheorem{remark}{Remark}
\sharetheorem{example}{Example}
\sharetheorem{conjecture}{Conjecture}
\sharetheorem{fact}{Fact}
\sharetheorem{property}{Property}

\theoremstyle{definition}
\sharetheorem{definition}{Definition}

\crefname{theorem}{Theorem}{Theorems}
\crefname{lemma}{Lemma}{Lemmas}
\crefname{observation}{Observation}{Observations}
\crefname{claim}{Claim}{Claims}
\crefname{proposition}{Proposition}{Propositions}
\crefname{corollary}{Corollary}{Corollaries}
\crefname{remark}{Remark}{Remarks}
\crefname{example}{Example}{Examples}
\crefname{conjecture}{Conjecture}{Conjectures}
\crefname{fact}{Fact}{Facts}
\crefname{property}{Property}{Properties}
\crefname{definition}{Definition}{Definitions}



\begin{document}

\title{Turnstile Streaming Algorithms Might (Still) as Well Be Linear Sketches, for Polynomial-Length Streams}
\author{Cheng Jiang\thanks{Department of Electrical Engineering and Computer Science, MIT. chengj36@mit.edu} \and Yinchen Liu\thanks{Institute for Interdisciplinary Information Sciences, Tsinghua University. liuyinch23@mails.tsinghua.edu.cn} \and Huacheng Yu\thanks{Department of Computer Science, Princeton University. yuhch123@gmail.com}}
\date{}
\maketitle
\thispagestyle{empty}
\setcounter{page}{0}

\begin{abstract}
A fundamental question in streaming complexity is whether every space-efficient turnstile algorithm is implicitly a linear sketch. 
The work of Li, Nguyen, and Woodruff (STOC~2014) established an equivalence between the two, but their reduction requires a stream length that is at least \emph{doubly exponential} in the dimension $n$.
In the opposite direction, Kallaugher and Price (STOC~2020) showed a separation for streams of \emph{linear} length, so the equivalence does not hold in general.
The most natural and practically relevant regime, that of \emph{polynomial}-length streams, has therefore remained open.

We show that polynomial-length turnstile algorithms admit linear-sketch simulations. 
More precisely, if a turnstile algorithm uses $S$ bits of space and succeeds on all streams of length $\poly(D, n)$, then on final vectors $x$ with $\|x\|_2 \le D$, its output can be recovered from $O(S)$ linear measurements of $x$, using $O(S)$ bits overall.
For smooth problems under appropriate input distributions, a mollified version of the reduction yields a bounded-entry sketch with $O(S / \log D)$ measurements and optimal $O(S)$ total space. Our results extend to strict turnstile streams and non-uniform \emph{Read-Once Branching Programs (ROBPs)}.

Instead of the transition-graph arguments of prior work, our proof uses Fourier analysis and tools from additive combinatorics to extract discrete linear measurements. We show that an $S$-bit algorithm can only be sensitive to a small set of ``heavy'' Fourier frequencies, which lie close to a finite group of order $2^{O(S)}$, and we use generators of this group as the rows of the sketching matrix. Combined with existing lower bounds for real sketches and for communication, this gives new lower bounds for polynomial-length turnstile streams.

\end{abstract}
\pagestyle{empty}
\addtocontents{toc}{\protect\thispagestyle{empty}}

\newpage
\pagestyle{empty}

{\small\setlength{\cftbeforesecskip}{0.4em}\tableofcontents}
\addtocontents{toc}{\protect\thispagestyle{empty}}

\newpage
\setcounter{page}{1}
\pagestyle{plain}

%!TEX root = main.tex
{
\numberwithin{theorem}{section}
\section{Introduction}
The turnstile streaming model is a standard framework for processing high-dimensional data that evolves over time. 
In this setting, we aim to track a vector $x \in \mathbb{Z}^n$, initialized to zero, through a stream of updates $(i, \Delta)$ for $\Delta \in \{-1, 1\}$. 
Each update modifies the $i$-th coordinate: $x_i \leftarrow x_i + \Delta$. 
Because $\Delta$ can be negative, this model allows for both the addition and deletion of data, making it more flexible than the simpler ``insertion-only'' model.

The central measure of a streaming algorithm is space complexity: we seek to compute properties of $x$ (such as its norm or heavy hitters) using far less space than would be required to store the vector explicitly. 
The canonical approach to achieving such space efficiency is through linear sketching. 
A linear-sketching algorithm pre-selects a matrix $\mathcal{A} \in \mathbb{Z}^{m \times n}$ and maintains the summary $\mathcal{A}x$ as the vector evolves. 
This approach is algorithmically advantageous because it admits trivial additive updates: when an update $(i, \Delta)$ arrives, the algorithm simply updates the summary by adding $\Delta \cdot \mathcal{A}_i$, where $\mathcal{A}_i$ is the $i$-th column of the matrix.

The turnstile model has proved useful for a broad collection of sublinear algorithms. Canonical examples include frequency moments and $\ell_p$-norm estimation~\cite{AlonMS96,Indyk06,CormodeDIM03}, $L_0$-sampling and related support-sensitive primitives~\cite{FrahlingIS08}, heavy hitters and frequency estimation~\cite{CharikarCF04,CormodeM05}, and dynamic geometric problems such as clustering, diameter-type questions, and coresets~\cite{JohnsonL84,Indyk04,FrahlingS05}. Graph streams provide another major source of applications: by maintaining linear measurements of the edge-incidence vector, one can recover spanning forests, graph sketches, sparsifiers, and spectral sparsifiers in dynamic (turnstile) streams~\cite{AhnGM12,AhnGM12b,KapralovLMMS14}. The same area also studies approximate matching~\cite{assadi2016maximum}, triangle counting~\cite{JhaSP15,BulteauFKP16}, and many related estimation problems; see also McGregor's survey~\cite{McGregor14}.
Across many of the examples above, the best known turnstile algorithms are linear sketches, and in many cases the strongest lower bounds are also currently known only for linear sketches.

This raises a basic question: is linearity inherent to space-efficient computation in the turnstile model? More precisely, if a problem admits a turnstile streaming algorithm using space $S$, must it also admit a linear sketch of comparable size?

The known answers depend on the length of the stream.
Building on Ganguly~\cite{ganguly2008lower}, Li, Nguyen, and Woodruff~\cite{li2014turnstile} and the follow-up by Ai, Hu, Li, and Woodruff~\cite{AHLW16} established an equivalence, showing that any $S$-space turnstile algorithm can be simulated by a linear sketch, provided the stream can be arbitrarily long (at least \emph{doubly exponential} in~$n$).

These arbitrarily-long-stream equivalences became a standard reduction in lower-bound arguments: one first converts a streaming algorithm into an integer or modular sketch via~\cite{li2014turnstile,AHLW16}. Then the proof proceeds either by using communication complexity in the easier SMP model~\cite{Konrad15,assadi2016maximum, AssadiKL17,KallaugherKP18, NelsonY19}, by proving lower bounds against sketching algorithms, as for approximating Max-CSPs~\cite{CGSV24}, by analyzing the resulting discrete sketch via rank/minrank arguments~\cite{KhannaPSW25}, or by importing lower bounds from real-valued sketches through the lifting line of work~\cite{PriceW11,LiW16,GangulyW18,NeedellSW22,SwartworthW23,gribelyuk2025lifting}.

On the other hand, Kallaugher and Price~\cite{KP20} proved that for streams of \emph{linear} length, there is a separation: they demonstrated problems for which non-linear algorithms are exponentially more space-efficient than any linear sketch.
These results leave a large gap: what happens in the most common setting of polynomial-length streams?
\subsection{Our results}
In this paper, we answer this question by showing that turnstile algorithms for polynomial-length streams can be simulated by linear sketches. Our approach avoids the combinatorial transition-graph argument of~\cite{li2014turnstile,AHLW16} and instead uses Fourier analysis, which works directly for streams of polynomial length. For readability, we state our main result for exact computation of a function $f$. We also provide full formulations for approximation, promise, and relation problems in \Cref{sec:randomized-transfer}. The result is as follows.

\begin{theorem}[Informal exact reduction, see \cref{cor:yao-minmax-exact-transfer}] Fix $D\geq 1$ and $R \geq \poly(D,n)$. Suppose there is a turnstile streaming algorithm that computes $f(x)$ in any turnstile stream of length $\poly(R)$ using $S$ bits of space with high probability. Then there is a linear sketching algorithm\footnote{As in~\cite{li2014turnstile,AHLW16}, the sketching matrix is not efficiently computable from the code of the streaming algorithm. The same is true for the theorem below.} that computes $f(x)$ for any $x$ with $\|x\|_2\leq D$ with high probability. Moreover, the linear sketching algorithm has sketching dimension $O(S)$ and uses $O(S)$ bits of space.
\end{theorem}

Our result also holds for strict turnstile streams as in~\cite{AHLW16}, and
even for \emph{Read-Once Branching Programs (ROBPs)}.
This implies that even if the streaming algorithm is allowed a non-uniform update rule (where the memory update logic changes at each step of the stream), the equivalence to linear sketching still holds.\footnote{ROBPs and streaming algorithms are equivalent when the stream length is less than $2^{O(S)}$, as one could also remember the current position in the stream, which the update algorithm can depend on.
But the equivalence does not hold for longer streams.}
Unlike~\cite{li2014turnstile,AHLW16}, our argument does not rely on a uniform transition graph or a mixing-time bound, which is why it also extends to non-uniform models such as ROBPs.

For a broad class of ``smooth'' functions, we prove a stronger result in terms of sketching dimension. In this case, we can require the linear sketch given by the reduction to have sketching dimension $O(S/\log R)$ and entries bounded by $\poly(R)$.

\begin{theorem}[Informal mollified reduction, see \cref{thm:randomized-mollified-transfer}]\label{thm:smooth-informal}
Let $\gamma_R$ be the $n$-dimensional discrete Gaussian distribution with radius $R\geq \poly(n)$, and let $\mathcal I$ be a distribution on $\mathbb Z^n$ supported in the Euclidean $\ell_2$-ball of radius $R^{5/4}$. Suppose $f: \mathbb{Z}^n \to \mathbb{R}$ satisfies a ``smoothness'' condition with respect to $\mathcal{I}$: $$\Pr_{Y\sim\mathcal I,\, Z\sim \gamma_R}\!\left[\left|f(Y+Z)-f(Y)\right|\leq \epsilon\right]\ge 1-\delta. $$
If there is a streaming algorithm using $S$ bits of space and computing $f$ in any turnstile stream of length $\poly(R)$ with error $\epsilon$ and success probability $1-\delta$, then there is a linear sketching algorithm for $f$ with error $O(\epsilon)$ that succeeds with probability $1-O(\delta)$ on $\mathcal{I}$, with sketching dimension $O(S/\log R)$, entries bounded by $\poly(R)$, and total space $O(S)$ bits.
\end{theorem}

We call this a ``mollified'' reduction because it exploits additional conditions on $f$ and the input distribution $\mathcal{I}$, in contrast to the ``exact'' reduction where we do not rely on any additional property of $f$.

The mollified reduction theorem can be used in two steps. First, choose a target distribution $\mathcal I$ that already yields a lower bound for bounded-entry linear sketches (for instance, in $L_p$ approximation one takes $\mathcal I$ to be a (planted) discrete Gaussian). Second, verify that the problem is smooth on average under the perturbation $Y\mapsto Y+Z$ with $Z\sim\gamma_R$. The theorem then converts any hypothetical $S$-bit turnstile streaming algorithm into a bounded-entry integer sketch of dimension $O(S/\log R)$.
\subsection{New turnstile streaming lower bounds}

We use our two reductions to prove new turnstile streaming lower bounds. In all applications, the main new ingredient is the reduction from turnstile algorithms on polynomial-length streams to linear sketches, and the lower bounds then follow from existing ones (for $L_0$ estimation, we give a new direct proof of the lower bound for modular sketches). We combine the mollified reduction with the lifting theorem of~\cite{gribelyuk2025lifting} and lower bounds for real sketches, and the exact reduction with lower bounds in the SMP and sketching models~\cite{AssadiKL17,DuMWY19,KallaugherKP18,CGSV24}. We discuss the two cases in turn.

\paragraph{Applications via lifting real sketches.}
Li and Woodruff~\cite{LiW16} pointed out an obstacle to deriving streaming lower bounds from real-valued sketching lower bounds.  They observed that the two
complexity measures, real sketching \emph{dimension} and streaming \emph{space}, are incomparable.  In one direction, sketching dimension can be far smaller than streaming space: a single real-valued inner product
$\langle u, v\rangle$ with $u = (1, W+1, (W+1)^2, \ldots, (W+1)^{n-1})$
recovers an arbitrary vector $v \in [W]^n$ exactly, so sketching dimension~$1$ suffices to encode all information that would require $n\log W$ streaming bits. In the other direction, real-valued sketching lower bounds do not imply streaming lower bounds: those lower bounds (e.g., \cite{GangulyW18,LiW16,NeedellSW22,SwartworthW23,PriceW11}) are proved against continuous (Gaussian) hard distributions, and after discretizing to integers
``it is no longer clear if the lower bounds hold''~\cite{LiW16}.
Thus, despite the equivalence of~\cite{li2014turnstile,AHLW16}, which gives a reduction from streaming to discrete sketching of the same space complexity, the gap between real-valued sketching complexity and discrete turnstile streaming complexity remained open.

Our results, combined with the recent lifting framework of Gribelyuk, Lin,
Woodruff, Yu, and Zhou~\cite{gribelyuk2025lifting}, close this gap for
smooth approximation problems. The lifting framework of \cite{gribelyuk2025lifting} converts a real-valued sketching dimension lower bound into a lower bound on the dimension of discrete (integer) sketches with bounded entries, losing only a constant factor.

Our \cref{thm:smooth-informal} then fits in the framework by converting any $S$-bit turnstile streaming algorithm into an integer linear sketch of dimension $O(S/\log R)$, with entries bounded by $\poly(R)$.  Contrapositively, an entry-bounded integer linear sketch dimension lower bound of $\Omega(d)$
propagates to a turnstile streaming lower bound of $\Omega(d\log R)$ bits. This bound is tight (up to a constant factor) for polynomial-length streams, since any integer sketching algorithm of dimension $d$ with entries bounded by $\poly(R)$ itself uses $O(d\log R)$ bits of space.

Thus for smooth approximation problems we obtain a reduction from polynomial-length turnstile streaming to real sketching that loses only a constant factor:
{\small
\[
  \underbrace{\text{Real sketch: dim}\;\Omega(d)}_{\text{\cite{GangulyW18,LiW16,NeedellSW22,SwartworthW23,PriceW11}}}
  \;\xrightarrow{\text{\cite{gribelyuk2025lifting}}}\;
  \underbrace{\text{Integer sketch: dim}\;\Omega(d)}_{\text{constant factor loss}}
  \;\xrightarrow{\text{\cref{thm:smooth-informal}}}\;
  \underbrace{\text{Streaming: }\Omega(d\log R)\text{ bits}}_{\text{poly-length streams, tight}}
\]
}We instantiate this reduction for $L_p$ estimation, operator and Ky Fan norms, eigenvalue approximation, PSD testing, and compressed sensing \cite{GangulyW18,LiW16,NeedellSW22,SwartworthW23,PriceW11,gribelyuk2025lifting}. As a representative example we state the $L_p$ results here; the full list of applications is given in \cref{thm:lifting-applications-full}. Item~(1) of the theorem below recovers the bound of~\cite{JayramW13} for $W=\poly(n)$ and $\epsilon^{-2}\log(1/\delta)\le n^{1-\gamma}$ for a fixed $\gamma>0$, while extending to general $W$ and to the regime where the minimum reaches~$n$. For constant $\delta$, the dependence on the stream length was established in~\cite{KNW10small}. For $p>2$, the constant-error term $\Omega(n^{1-2/p}\epsilon^{-2}\log W)$ is already known~\cite{LiW13}; our bound adds the $\log n$ term and the dependence on subconstant~$\delta$.
\setcounter{theorem}{2}
\par\penalty0
\begin{samepage}
\begin{theorem}[see \cref{thm:lifting-applications-full} for a full list]
For fixed $p\ge1$, sufficiently small $\epsilon>0$, and $0<\delta\le1/3$, any randomized turnstile streaming algorithm that outputs a $(1+\epsilon)$-approximation to the $L_p$ norm problem on streams of length at most $W\geq \poly(n / \delta) $ with success probability at least $1-\delta$
must use the following number of bits of space:
\begin{enumerate}
\item ($1\le p\le 2$): $\Omega\!\left(\min\{\epsilon^{-2}\log(1/\delta),n\}\cdot \log W\right)$.
\item ($p>2$, $\delta>2^{-n^{c_p}}$ for a sufficiently small constant $c_p>0$):
\[
\begin{aligned}
\Omega\Bigl(\min\bigl\{&n^{1-2/p}\epsilon^{-2/p}\log n\,\log^{2/p}(1/\delta)+n^{1-2/p}\epsilon^{-2}\log(1/\delta),\,n\bigr\}\cdot\log W\Bigr).
\end{aligned}
\]
\end{enumerate}
\end{theorem}
\end{samepage}

\paragraph{Applications via modular sketches and the SMP model.}
A second group of applications comes from communication complexity. In the public-coin simultaneous-message-passing (SMP) model \cite[Section~2.3]{AssadiKL17}, each player $i$ holds a vector $x^{(i)}\in\mathbb Z^n$ (equivalently, a local portion of the stream), the players share public randomness, each sends one message simultaneously to a referee, and the referee must compute a function $f$ of the sum $x=\sum_i x^{(i)}$. A linear sketching algorithm for $f$ naturally gives a protocol in this model: once a common sketch matrix $\mathcal{A}$ is fixed, player $i$ sends $\mathcal{A}x^{(i)}$, and the referee adds the messages $\mathcal{A}(\sum_i x^{(i)})$, then applies the decoding algorithm to recover $f(x)$.
This additivity removes any dependence on the order of players or intermediate states, so lower bounds are often cleaner to prove in the SMP model than directly against arbitrary streaming algorithms.

Given the equivalence between turnstile streaming and linear sketching \cite{li2014turnstile,AHLW16}, any lower bound in the SMP model for computing $f$ immediately implies a lower bound for turnstile streaming algorithms that compute $f$. This linear-sketch/SMP route has yielded lower bounds for estimating maximum matching size~\cite{AssadiKL17}, graph and hypergraph subgraph counting~\cite{KallaugherKP18}, and spanning forest computation~\cite{NelsonY19}. In the same way, the lower bounds of~\cite{CGSV24} against sketching algorithms for approximating Max-CSPs extend to turnstile streams through the reductions of~\cite{li2014turnstile,AHLW16}.

However, as mentioned above, the prior equivalences of~\cite{li2014turnstile,AHLW16} only hold for arbitrarily long streams (at least doubly exponential in $n$), and do not extend to non-uniform models such as ROBPs. Thus, it was not clear whether the SMP lower bounds for these problems would imply turnstile streaming lower bounds in the polynomial-length regime, or whether the non-uniformity of the streaming model would allow for more efficient algorithms than linear sketches.

Our results show that these SMP lower bounds also imply lower bounds for the more restricted but practical polynomial-length strict turnstile streams, and even for non-uniform models such as ROBPs. We instantiate this reduction for $L_0$ estimation, dynamic graph streaming problems such as matching-size estimation and subgraph counting~\cite{AssadiKL17,DuMWY19,KallaugherKP18}, and approximation of constraint satisfaction problems~\cite{CGSV24}. In each case, we use \cref{cor:yao-minmax-exact-transfer} to obtain a public-coin linear sketch with image size $2^{O(S)}$ that is correct on every input of bounded norm, and then invoke the corresponding public-coin simultaneous or sketching lower bound. For $L_0$ estimation, we instead provide a lower bound proof for such sketches directly, which also removes the restriction $\epsilon\ge n^{-0.49}$ of~\cite{DuMWY19}, with $\log n$ replaced by $\log(1+n\epsilon^2)$. As examples we state $L_0$ estimation and matching-size estimation, whose full list is given in \cref{thm:smp-applications-full}, and then the Max-CSP dichotomy. The $L_0$ lower bound below matches the upper bound of~\cite{KNW10},
$O(\epsilon^{-2}\log n\,(\log(1/\epsilon)+\log\log W))$ bits,
whenever $\epsilon\ge n^{-0.49}$ and $\log(1/\epsilon)=O(\log\log W)$.
Compared with~\cite{DuMWY19}, which relies on~\cite{AHLW16}, our bound only
requires correctness on streams of polynomial length, and it also holds for non-uniform ROBPs.

\begin{theorem}[see \cref{thm:smp-applications-full} for a full list]
For a sufficiently small absolute constant $c_0>0$ and a sufficiently large absolute constant $C_0>0$, the following holds for randomized turnstile streaming algorithms on streams of length at most $W\ge n^{C_0}$ with success probability at least $2/3$.
\begin{enumerate}
\item ($L_0$ estimation, $0<\epsilon\le c_0$) algorithms that output a $(1+\epsilon)$-approximation to the $L_0$ norm of the final frequency vector must use
$
\Omega\!\left(\epsilon^{-2}\log(1+n\epsilon^2)\,\log\log W\right)
$
bits.
\item (Estimating maximum matching size, fixed sufficiently small $\epsilon>0$) algorithms that output a $(1+\epsilon)$-approximation to the maximum matching size in strict turnstile streams whose final frequency vector encodes a multigraph on $n$ vertices must use $\Omega(n^{2-O(\epsilon)})$ bits.
\end{enumerate}
\end{theorem}

Our exact reduction also applies to constraint satisfaction problems on
turnstile streams, in which constraints are inserted and deleted. For a family
$\mathcal F$ of predicates $f:[q]^k\to\{0,1\}$ and $0\le\beta<\gamma\le1$, an
instance of $(\gamma,\beta)$-$\mathsf{Max\text{-}CSP}(\mathcal F)$ on $m$
variables is a multiset of constraints, each applying a predicate in
$\mathcal F$ to $k$ of the variables. The task is to distinguish instances in
which some assignment satisfies at least a $\gamma$ fraction of the
constraints from instances in which every assignment satisfies at most a
$\beta$ fraction. Chou, Golovnev, Sudan, and Velusamy~\cite{CGSV24} proved a
dichotomy for these problems in the sketching model: either a linear sketch
solves the problem in polylogarithmic space, or every sketching algorithm
needs $\Omega(\sqrt m)$ space. As noted above, their lower bound reaches
turnstile streams only through~\cite{li2014turnstile,AHLW16}, which need
superpolynomially long streams~\cite[Remark~2.4]{CGSV24}. Combined with our
exact reduction, their
dichotomy becomes a dichotomy for turnstile streaming algorithms on streams
of polynomial length.

\begin{theorem}[Turnstile streaming dichotomy for polynomial-length streams, see \cref{cor:turnstile-max-csp}]
For every $q,k\in\mathbb N$, every family $\mathcal F\subseteq\{f:[q]^k\to\{0,1\}\}$, and every $0\le\beta<\gamma\le1$, one of the following two conditions holds.
\begin{enumerate}
\item $(\gamma,\beta)$-$\mathsf{Max\text{-}CSP}(\mathcal F)$ on turnstile streams of any length polynomial in $m$ can be solved with success probability $2/3$ by a randomized linear sketch using $O(\log^3 m)$ space.
\item For every $0<\epsilon<(\gamma-\beta)/2$, there is a constant $C$ such that every randomized algorithm that solves $(\gamma-\epsilon,\beta+\epsilon)$-$\mathsf{Max\text{-}CSP}(\mathcal F)$ with success probability $2/3$ on all strict turnstile streams of length at most $m^C$ requires $\Omega(\sqrt m)$ space.
\end{enumerate}
Furthermore, if $\gamma$ and $\beta$ are rational, there is a polynomial-space algorithm that decides which of the two conditions holds, given $\gamma$ and $\beta$ as ratios of two integers in binary and $\mathcal F$ as the truth tables of its functions.
\end{theorem}

\subsection{Related work}
\paragraph{Equivalences between streaming and linear sketching.}
Li, Nguyen, and Woodruff~\cite{li2014turnstile}, building on Ganguly~\cite{ganguly2008lower}, showed that any one-pass turnstile streaming algorithm using $S$ bits can be simulated by a modular linear sketch, which stores $\langle a_i,x\rangle\bmod q_i$ for integer vectors $a_1,\ldots,a_r$ and moduli $q_1,\ldots,q_r$. Ai, Hu, Li, and Woodruff~\cite{AHLW16} extended this equivalence to strict turnstile streams, to multi-pass algorithms, and to streams with bounded intermediate states.

However, all these equivalences require the algorithm to tolerate extremely long streams. The approach of~\cite{li2014turnstile,AHLW16} is combinatorial: it models the streaming algorithm as a finite automaton, studies the transition graph and inserts long random walks to drive the automaton toward stationarity on strongly connected components. The resulting mixing requirement forces the stream length to be doubly exponential in~$S$ and $n$.
Furthermore, because the reduction operates on the transition graph, it
requires a \emph{uniform} transition function (the same update rule
at every time step) and does not extend to non-uniform models such as
Read-Once Branching Programs (ROBPs).

A related issue is how well the reduction preserves the \emph{dimension} of the resulting sketch.
The reduction produces a modular sketch as above with $r$ rows, where each
modulus~$q_i$ may be as small as~$2$. The total bit
complexity~$\sum_i \log q_i$ is essentially~$S$, but the sketch dimension~$r$
can be as large as~$S$ itself. For \emph{approximation problems}, which cover
most streaming applications (norm estimation, frequency
moments, etc.), dimension lower bounds are the natural measure
of hardness. Several works prove dimension lower bounds for real-valued
linear sketches~\cite{GangulyW18,LiW16,PriceW11,NeedellSW22, SwartworthW23}, and the recent lifting
framework of~\cite{gribelyuk2025lifting} converts these to discrete dimension
lower bounds. Since these lower bounds are phrased in terms of the number of
rows rather than total bits, a reduction that inflates dimension by up to a
factor of~$\log W$ for stream length $W$ (as can happen when all moduli are small) yields weaker
implications than one that preserves dimension.

Thus the reductions of~\cite{li2014turnstile,AHLW16} are equivalence results that give lower bounds for unrestricted streams, but they say nothing about the more realistic regime of polynomial-length streams studied here.

\paragraph{Separations for restricted streams.}
Kallaugher and Price~\cite{KP20} showed that the long-stream requirement above cannot be dropped entirely: they obtained \emph{exponential separations} between turnstile algorithms and linear sketches, including problems with $O(\log n)$-space streaming algorithms but linear sketches requiring $\Omega(n^{1/3}/\log^{1/3} n)$ space. Their constructions use one of two kinds of restrictions:
\begin{itemize}
\item For binary or box-constrained streams, the intermediate states are
  promised to lie in $\{0,1\}^n$ (or more generally in
  $\{-(2M-1),\ldots,2M-1\}^n$ for a parameter~$M$). Then an efficient
  non-linear algorithm can use the fact that each update reveals
  information about the new coordinate value.
\item For triangle counting on (bounded-degree) graphs, if the stream length
  is linear in the number of vertices, efficient non-linear algorithms exist,
  obtained by adapting insertion-only sampling algorithms. This short-stream
  result allows arbitrary intermediate multigraphs.
\end{itemize}
These separations exploit restrictions on the \emph{intermediate states} of the stream or on its length. Our result shows that the
intermediate-state or short-stream constraints they impose are \emph{necessary} for the separation: removing them restores the equivalence at polynomial stream length. Indeed, our proof requires only that for every reachable frequency vector~$v$, the streaming algorithm (possibly
randomized) is correct on \emph{all} streams of polynomial length that reach~$v$. This condition holds automatically for standard turnstile algorithms without special intermediate-state promises.

Kallaugher and Price~\cite{KP20} also gave an explicit deterministic reduction for \emph{total functions}, converting an $S$-space streaming algorithm into a linear sketch using $O(S \log n)$ space for streams of length $n + O(S)$. This argument does not extend to approximation, promise problems or binary relations, since it uses that two inputs reaching the same state must have the same unique correct answer.

\paragraph{XOR streams and modular updates.} Another related line of work considers XOR streams or other modular
updates~\cite{KMSY18,hosseini2018optimality}. In these models, addition
is taken over a finite group: over
$\mathbb F_2$, insertions and deletions coincide, and more generally updates
are taken modulo $p$. Kannan, Mossel, Sanyal, and Yaroslavtsev~\cite{KMSY18}
initiated the study of this setting by relating $\mathbb F_2$-sketching to
one-way communication for XOR functions, and showed that streaming algorithms
under uniformly random XOR updates can be converted to $\mathbb F_2$-sketches
for the uniform distribution with only a small loss, already for streams of
length $\widetilde O(n)$. Hosseini, Lovett, and
Yaroslavtsev~\cite{hosseini2018optimality} extended this perspective to
modular updates and approximation, and in particular showed that for total
functions (or approximate functions) the streaming/sketching equivalence already holds under the much
milder assumption that the algorithm works on streams of length
$\widetilde O(n^2)$. Their Fourier-analytic techniques are closely related to ours.
We extend this approach from XOR/modular updates over finite groups to the
standard turnstile model over $\mathbb Z$.

\numberwithin{theorem}{subsection}
}
%!TEX root = main.tex

\section{Technical overview}
\label{sec:technical-overview}
The core of our result is a reduction showing that any space-efficient turnstile algorithm for polynomial-length streams can be transformed into an equivalent linear sketch.
To convey the intuition, let us focus on a decision problem $f: \mathbb{Z}^n \to \{0, 1\}$.
Given a turnstile streaming algorithm for $f$ using $S$ bits of space, we aim to ``extract'' a (random) sketching matrix $\mathcal{A}$ with $O(S)$ rows and bounded integer entries such that the sketch $\mathcal{A}x$ is sufficient to recover $f(x)$ on the target inputs of interest.
 
\paragraph{Noisy landing experiment.}
To extract a linear sketch from the streaming algorithm, we run it on a carefully chosen distribution of ``noisy'' streams.
Suppose the algorithm processes a sequence of $M = \poly(n)$ independent increments $X_1, X_2, \dots, X_M$, where each $X_i$ is sampled from an $n$-dimensional discrete Gaussian distribution $\gamma_R$ with a large radius $R = \poly(n)$.
After these $M$ updates, the current frequency vector is $X = \sum_{i=1}^M X_i$.
Finally, to test the algorithm on a target vector $y$, we append the increment $(y - X)$, effectively ``landing'' the stream at $y$.
If each increment is realized by unit updates, the total length of this stream is still polynomial.
The intuition is that adding many large random updates should erase any nonlinear dependence on the past, so that only linear constraints remain.

\paragraph{Conditioning on a state sequence.}
If the algorithm uses only $S$ bits of space, then it compresses the information in $X_1,\dots,X_M$ into a short state sequence.
Let $\sigma_i \in \{0,1\}^S$ be the memory state after processing $X_i$, and fix a state sequence $\sigma=(\sigma_1,\dots,\sigma_M)$.
Conditioned on these states, each increment acquires a posterior law $\mu_i=(X_i\mid \sigma)$.
Because the algorithm is read-once, conditioning on the \emph{whole state sequence} still leaves the increments independent.
Hence the conditioned total vector $X$ has law $\nu=\mu_1*\mu_2*\cdots *\mu_M$.
Once $\sigma$ is fixed, the whole problem reduces to understanding this convolution measure $\nu$.

We will describe later how to build $\mathcal{A}$ from the posterior laws $\mu_i$.
Before that, it is useful to identify the property we need from $\mathcal{A}$.
Throughout the reduction, we only need the sketch to work for target vectors $y$ of polynomially bounded norm. In particular, we care about the regime $\|y\|_2\ll R^{0.1}$, where $R$ is the radius of each discrete Gaussian.
Thus, for any two such targets $y_1,y_2$ with the same sketch value $\mathcal{A}y_1=\mathcal{A}y_2$, we want the streaming algorithm to be essentially forced to give the same answer, and hence to conclude $f(y_1)=f(y_2)$.

\paragraph{Reduction to translation invariance.} From the algorithm's perspective, after processing the random prefix and reaching the state sequence $\sigma$, the final update needed to land at $y_1$ is $y_1-X$, while for $y_2$ it is $y_2-X$.
Conditioned on $\sigma$, these two final updates are exactly the translates of $-\nu$ by $y_1$ and $y_2$.
Therefore, if $v\coloneqq y_1-y_2$ lies in the kernel of $\mathcal{A}$ and the distributions $\nu$ and $\tau_v\nu\coloneqq\nu((\cdot)-v)$ are close in total variation distance, then the algorithm can barely distinguish the two landing updates.
Consequently it must behave nearly identically on $y_1$ and $y_2$.\footnote{The extracted matrix $\mathcal{A}$ may depend on the chosen state sequence $\sigma$, so the resulting sketching algorithm is randomized. One may also imagine first applying Yao's minimax principle~\cite{yao1977probabilistic} to fix a hard input distribution for linear sketches, then it would suffice to find a single good state sequence and extract a deterministic sketch from it.} Hence, the correctness of the turnstile streaming algorithm gives $f(y_1) = f(y_2)$.

Therefore, the central technical task is to construct a matrix $\mathcal{A}$ with $O(S)$ rows that satisfies this sufficiency condition: it must ``capture'' all directions $v$ where the distribution $\nu$ is sensitive to shifts, leaving only translation-invariant directions in its kernel.

\subsection{\texorpdfstring{Translation invariance of $\nu$}{Translation invariance of nu}}
\label{subsec:overview-translation}
To extract $\mathcal{A}$, we must understand the directions in which the aggregate distribution $\nu=\mu_1*\cdots *\mu_M$ is sensitive to shifts.
Recall that each posterior law $\mu_i=(X_i\mid \sigma)$ comes from conditioning a discrete Gaussian update $X_i$ on an $S$-bit memory state. This suggests that, conditioning on a typical state $\sigma_i$ that appears with probability $2^{-O(S)}$, $\mu_i$ should stay widely spread, with a pointwise upper bound of the form $\mu_i \le 2^{O(S)}\gamma_R$.

\paragraph{Warm-up: the completely diffuse case.}
As a first example, suppose the algorithm behaves in a purely nonlinear fashion, so that each $\mu_i$ is just the Gaussian $\gamma_R$ restricted to a large, random support.
Then a high-dimensional Central Limit Theorem (CLT) suggests that, when $M$ is large enough, after convolving $M$ such independent pieces, the distribution $\nu$ should resemble a broad Gaussian of radius about $R\sqrt M$.
In that regime, the distribution is so ``blurred'' that $\nu$ and its translate $\tau_v\nu$ have total variation distance roughly $O(\|v\|_2/(R\sqrt M))=o(1)$ for every $v$ with $\|v\|_2\ll R^{0.1}$.

Such an algorithm cannot detect small shifts, so it can only compute trivial functions of the input when the frequency vector is polynomially smaller than the radius.
In this case, no sketch is required, and $\mathcal{A}$ can be empty.

\paragraph{Two model obstructions.}
In general, however, the convolution need not diffuse across all of $\mathbb Z^n$.
The following two examples illustrate the obstructions that prevent $\nu$ from becoming translation-invariant:
\begin{itemize}
	\item \emph{Geometric subspace obstruction.} Suppose every $\mu_i$ is the Gaussian $\gamma_R$ restricted to a specific high-dimensional subspace $W \subset \mathbb{R}^n$.
	Then $\nu$ is also supported on $W$.
	Any shift $v$ with a component orthogonal to $W$ moves the distribution off its support, so the TV distance becomes $1$.
	To handle this, the matrix $\mathcal{A}$ must discover the orthogonal complement $W^\perp$ and include a basis for it as rows. 
	This ensures that $\mathcal{A}v=0$ implies $v \in W$.
	\item \emph{Arithmetic lattice obstruction.} Suppose every $\mu_i$ is supported only on ``even'' points, i.e., on vectors satisfying $\sum_j x_j \equiv 0 \pmod 2$.
	Then $\nu$ remains trapped in the same lattice.
	Shifting by an ``odd'' vector produces disjoint support, again giving TV distance $1$.
	In this case, the sketch must record the relevant parity constraint, for instance by including the all-ones vector as a row,  ensuring that $\mathcal{A}v=0$ forces $v$ to be an even vector.
\end{itemize}

These examples suggest that $\mathcal{A}$ must capture the ``non-diffuse'' directions of the posterior distributions. We formalize this using Fourier analysis.

\subsection{\texorpdfstring{Fourier analysis of $\mu_i$}{Fourier analysis of mu_i}}
\label{subsec:overview-fourier}
We now explain how Fourier analysis gives translation invariance.

\paragraph{Warm-up: the pure Gaussian case.}
First consider the ideal situation in which the streaming algorithm reveals no information at all, so $\mu_i=\gamma_R$ for every $i\in[M]$. Then $\nu=\gamma_R*\cdots *\gamma_R$ is simply the convolution of $M$ independent discrete Gaussians. In this setting, we expect $\nu$ to be insensitive to the small shifts relevant to our reduction, so no nontrivial sketching matrix should be needed.

For a distribution $\mu$ on $\mathbb Z^n$, we define its Fourier transform $\widehat{\mu}: \mathbb{T}^n \to \mathbb{C}$ by
$$
\widehat{\mu}(\zeta)=\sum_{x\in\mathbb Z^n}\mu(x)e^{-2\pi i\langle \zeta,x\rangle},
\qquad \zeta\in\mathbb T^n,
$$
where $\mathbb T^n=[-1/2,1/2)^n$ is the $n$-dimensional torus. 

For the discrete Gaussian $\gamma_R$, the Fourier transform $\widehat{\gamma_R}(\zeta)$ is sharply concentrated near the origin:
$
|\widehat{\gamma_R}(\zeta)|\le \exp(-\Omega(R^2\|\zeta\|_{\mathbb T^n}^2))
$ (\cref{lem:dg-fourier-decay-general}), where $\|\zeta\|_{\mathbb T^n}$ is the Euclidean distance from $\zeta$ to $\mathbb Z^n$.
Since $\nu$ is the convolution of $M$ independent copies of $\gamma_R$, its Fourier transform is the pointwise product of the individual transforms:
$$\widehat{\nu}(\zeta) = \prod_{i=1}^M \widehat{\mu}_i(\zeta) = \left( \widehat{\gamma_R}(\zeta) \right)^M \approx \exp(-\Theta(M R^2 \|\zeta\|_{\mathbb{T}^n}^2)).$$
This $M$-fold product makes the concentration much sharper. In particular, for frequencies with
$
\|\zeta\|_{\mathbb T^n} \gg 1/(R\sqrt{M}),
$
the value of $\widehat{\nu}(\zeta)$ is negligibly small.

Now consider the shifted distribution $\tau_v \nu$, where $\| v\|_2 \ll R^{0.1}$. 
Its Fourier transform is simply $\widehat{\nu}(\zeta) \cdot e^{-2\pi i \langle \zeta, v \rangle}$. 
We can bound the difference between the two distributions in the frequency domain:
$$\widehat{\nu - \tau_v \nu}(\zeta) = \widehat{\nu}(\zeta) \left( 1 - e^{-2\pi i \langle \zeta, v \rangle} \right).$$
There are two regimes. For ``large'' frequencies, namely $\|\zeta\|_{\mathbb T^n} \gg 1/(R\sqrt{M})$, the term $\widehat{\nu}(\zeta)$ is already tiny. For ``small'' frequencies, namely $\|\zeta\|_{\mathbb T^n} \lesssim 1/(R\sqrt{M})$, the phase factor $(1 - e^{-2\pi i \langle \zeta, v \rangle})$ is bounded by $O(|\langle \zeta, v \rangle|) \leq O(\|\zeta\|_{\mathbb T^n} \cdot \|v\|_2)$.
Substituting the bound on $\|\zeta\|_{\mathbb T^n}$, we obtain a pointwise bound of order $O(\|v\|_2/(R\sqrt{M}))$ in the small-frequency regime as well. By Parseval's identity, this gives a small $\ell_2$ distance between $\nu$ and $\tau_v \nu$.

\paragraph{From $\ell_2$ to total variation.}
In high dimension, this $\ell_2$ bound does not give a useful bound on the total variation ($\ell_1$) distance, since Cauchy--Schwarz loses the square root of the number of points in the support of $\nu$. We therefore work with $\nu$ one line at a time, in the direction $v$.

Fix $x \in \mathbb{Z}^n$, and consider the one-dimensional restriction $\nu_{x,v}(t)\coloneqq\nu(x+tv)$ for $t\in\mathbb Z$. Its Fourier transform at $t\in\mathbb R/\mathbb Z$ is the integral of $\widehat{\nu}(\zeta)e^{2\pi i\langle\zeta,x\rangle}$ over the frequencies $\zeta$ with $\langle\zeta,v\rangle\equiv t$. Since $\widehat{\nu}$ is concentrated near the origin, $\widehat{\nu_{x,v}}$ is concentrated near $0$, and the Fourier argument above, in one dimension, bounds $\sum_t (\nu_{x,v}(t)-\nu_{x,v}(t+1))^2$. On a single line, Cauchy--Schwarz loses only the square root of the number of points that carry most of the mass of $\nu_{x,v}$, which is about $R\sqrt M$ rather than $(R\sqrt M)^n$. This bounds $\sum_t |\nu_{x,v}(t)-\nu_{x,v}(t+1)|$, and summing over all lines parallel to $v$ bounds $d_{\mathrm{TV}}(\nu,\tau_v\nu)$.

The Gaussian warm-up shows that translation invariance follows once the Fourier mass of $\nu$ is concentrated near the origin. Therefore, for general posterior laws $\mu_i$, the main task is to identify the frequencies $\zeta$ (away from the origin) with a large $\widehat{\nu}(\zeta)$. By incorporating these ``heavy'' frequencies as rows of our sketching matrix $\mathcal{A}$, we ensure that every shift $v$ in the kernel of $\mathcal{A}$ is nearly orthogonal to all heavy $\zeta$. Then $|\widehat{\nu}|$ is small outside a thin strip around the hyperplanes $\{\zeta:\langle\zeta,v\rangle\in\mathbb Z\}$, so every $\widehat{\nu_{x,v}}$ is small outside a short interval around $0$, and the argument above, line by line, again shows that $\nu$ is nearly invariant under translation by $v$.

Figure~\ref{fig:proof-pipeline} summarizes the four steps of the proof.
\begin{figure}[t]
\centering
\small
\begin{tikzpicture}[
font=\footnotesize,
box/.style={draw, rounded corners, align=center, fill=Blue!5, text width=6cm, inner sep=4pt, minimum height=2.15cm},
mainbox/.style={box, fill=yellow!25, draw=red!70!black, line width=1.1pt},
tvbox/.style={box, fill=teal!18, draw=teal!70!black, line width=1.05pt},
arr/.style={-{Latex[length=5mm,width=2.6mm]}, very thick, shorten >=2pt, shorten <=2pt},
mainarr/.style={-{Latex[length=5.4mm,width=2.9mm]}, ultra thick, draw=red!70!black, shorten >=2pt, shorten <=2pt},
tvarr/.style={-{Latex[length=5.4mm,width=2.9mm]}, ultra thick, draw=teal!70!black, shorten >=2pt, shorten <=2pt},
otherr/.style={-{Latex[length=5.4mm,width=2.9mm]}, ultra thick, draw=blue!5!black, shorten >=2pt, shorten <=2pt},
node distance=8mm and 9mm
]
\node[box] (cond) {\textbf{Prefix conditioning} \\Fix a typical state sequence $\sigma$; conditioned blocks $\mu_i \coloneqq (X_i\mid\sigma)$ remain independent dense Gaussian pieces.};
\node[mainbox, right=of cond] (spect) {\textbf{Large-spectrum structure} \\Analyze the conditioned convolution $\nu\coloneqq\mu_1*\cdots*\mu_M$ and show that its heavy Fourier frequencies lie near a finite group of order $2^{O(S)}$.};
\node[tvbox, below=of spect] (tv) {\textbf{Translation invariance} \\If a direction $v$ lies in the kernel of the extracted sketch, then $\nu$ and $\tau_v\nu$ are close in total variation.};
\node[box, left=of tv] (transfer) {\textbf{Decoding} \\Two targets with the same sketch value induce nearly identical last-block laws, so the decoder can assign one answer to each sketch value.};

\draw[mainarr] (cond) -- (spect);
\draw[tvarr] (spect) -- (tv);
\draw[otherr] (tv) -- (transfer);
\end{tikzpicture}
\caption{The four steps of the proof. Within this overview, \cref{subsec:overview-translation} explains the obstructions to translation invariance, and \cref{subsec:overview-fourier} develops the Fourier-analytic argument for translation invariance. The \emph{exact reduction} places the entire large spectrum near a finite group of order $2^{O(S)}$ and yields $O(S)$ space (\cref{subsec:overview-heavy-zeta}). The \emph{mollified reduction} strengthens it for smooth problems: smoothness restricts attention to near-origin frequencies, where the finite group is refined to a low-dimensional subspace, and this yields a sketch with $O(S/\log R)$ rows and bounded integer entries (\cref{subsec:overview-mollified-counting}).}
\label{fig:proof-pipeline}
\end{figure}

\subsection{\texorpdfstring{Identification of heavy $\zeta$: exact reduction}{Identification of heavy zeta: exact reduction}}
\label{subsec:overview-heavy-zeta}
To generalize from the Gaussian case, we first focus on finding the ``heavy'' frequencies of the individual posterior distributions $\mu_i$.
The key technical tool is the following lemma, which uses the fact that $\mu_i$ is ``widely spread'' (relative to a Gaussian) to bound the number of independent frequencies it can support.
\begin{lemma}[Informal, see \cref{cor:kappa-dissociated-upper-general}]\label{lem:dissociated-intro}
	Let $\mu \leq 2^{O(S)} \cdot \gamma_R$ pointwise. 
	Suppose $A = \{\zeta_1, \zeta_2, \ldots\} \subseteq \mathbb{T}^n$ is a set of frequencies such that $|\widehat{\mu}(\zeta_i)| \geq 1/2$ for all $i$. 
	If $A$ is $(5\sqrt{S}/R)$-dissociated, that is, every nontrivial $\{-1,0,1\}$-linear combination of its elements has $\|\cdot\|_{\mathbb T^n}$-norm at least $5\sqrt{S}/R$, then $|A| \leq O(S)$.
\end{lemma}
\noindent We now explain the geometric intuition behind this lemma. If $|\widehat{\mu}(\zeta)|$ is large, then $\mu$ is strongly aligned with the character $x\mapsto e^{-2\pi i\langle \zeta,x\rangle}$.
Geometrically, this means that $\mu$ concentrates on a union of slices in $\mathbb{Z}^n$ (hyperplanes orthogonal to $\zeta$).
If we have a dissociated set $A$ of such frequencies, these ``unions of slices'' behave as if they are approximately independent. 
Each additional frequency $\zeta_i \in A$ cuts down the available support of the distribution. 
Specifically, the mass of the intersection of $|A|$ such independent slices relative to the underlying Gaussian $\gamma_R$ scales as $2^{-\Omega(|A|)}$. 
Since our pointwise bound $\mu \leq 2^{O(S)} \cdot \gamma_R$ implies that $\mu$ cannot be concentrated on a set of measure smaller than $2^{-O(S)}$, it follows that the number of such independent directions $|A|$ must be at most $O(S)$. 

Despite its geometric intuition, the lemma is proved using tools from additive combinatorics, specifically an adaptation of Chang's large-spectrum argument \cite{Chang02Freiman} to our setting of distributions on $\mathbb{Z}^n$: Chang works primarily in finite Abelian groups, while here the fast decay of the discrete Gaussian $\gamma_R$ plays the role of controlling the Gaussian mass of intersections of the corresponding slices.

\paragraph{Greedy algorithm for extracting $\mathcal{A}$.}
To see the basic idea, first consider the simplified case where all posterior laws are identical (to $\mu_1$). We then construct $\mathcal{A}$ iteratively:
\begin{itemize}
	\item Find a frequency $\zeta_i$ such that $|\widehat{\mu}_1(\zeta_i)| \geq 1/2$ and $\zeta_i$ has distance at least $5\sqrt{S}/R$ from the $\{-1, 0 ,1\}$ span of the previously chosen $\{\zeta_1, \dots, \zeta_{i-1}\}$.
	\item Round the coordinates of $\zeta_i$ to the nearest multiple of $1/Q$ (where $Q = \poly(n)$) to ensure the sketch has bounded integer entries.
	\item Add $Q\zeta_i$ as a row in $\mathcal{A}$.
\end{itemize}

The dissociation lemma guarantees that this process must terminate after $O(S)$ steps. Once it stops, every frequency $\zeta$ with large Fourier coefficient $|\widehat{\mu_1}(\zeta)|$ must lie very close to the integer span of the rows of $\mathcal{A}$. If $v$ lies in the kernel of $\mathcal{A}$ (i.e., $v$ is nearly orthogonal to the rows of $\mathcal{A}$), then all ``heavy'' frequencies of $\mu_1$ satisfy $\langle \zeta, v \rangle \approx 0$. Frequencies away from this span have smaller Fourier coefficients, so after taking the product
$
\widehat{\nu}(\zeta) = \prod_{i=1}^M \widehat{\mu}_i(\zeta),
$
their contribution is damped by the $M$-fold convolution, just as in the Gaussian warm-up. In other words, $\mathcal{A}$ records precisely the directions in which the posterior laws fail to diffuse.

\paragraph{General case via symmetrization.}
When the posterior laws $\mu_i$ are no longer identical, we do not have a single distribution whose large spectrum directly controls $\widehat{\nu}$. To recover such an object, we use the standard symmetrization trick from probability theory. For each $i$, let $\widetilde\mu_i(x)\coloneqq\mu_i(-x)$ denote the reflected distribution, and consider the symmetrized average $\mu_{\mathrm{sym}}\coloneqq\frac{1}{M} \sum_{i=1}^M \mu_i * \widetilde\mu_i$. Then $\widehat{\mu_{\mathrm{sym}}}(\zeta)=\frac{1}{M}\sum_{i=1}^M \big|\widehat{\mu}_i(\zeta)\big|^2$. If a frequency $\zeta$ is still ``sharp'' for the full convolution $\nu$, meaning that $|\widehat{\nu}(\zeta)|=\prod_{i=1}^M |\widehat{\mu}_i(\zeta)|$ is large, then AM--GM implies that the symmetrized average $\widehat{\mu_{\mathrm{sym}}}(\zeta)$ must be close to $1$. Thus the heavy frequencies of $\nu$ are also visible in $\mu_{\mathrm{sym}}$. We can therefore apply the same structural result to $\mu_{\mathrm{sym}}$, extract $O(S)$ controlling directions, and conclude that any $v$ in the kernel of $\mathcal A$ is nearly orthogonal to every heavy frequency of $\nu$. 

\paragraph{Improving the space bound.}
The iterative Fourier procedure identifies $O(S)$ heavy frequencies $\{\zeta_1,\dots,\zeta_\ell\}\subset \mathbb{T}^n$.
If we simply round each coordinate to precision $1/Q$ and store the sketch as an ordinary integer matrix $\mathcal A x\in \mathbb Z^\ell$, then each row costs $O(\log n)$ bits, leading to a total space of $O(S\log n)$.
This $O(\log n)$ overhead is often a byproduct of forcing a modular sketch (e.g., over $\mathbb{F}_2$) into an integer representation. 
To achieve tighter space, we must allow the sketch to capture linear constraints modulo various integers $q$.

Consider the case where $\mu_i$ is supported only on ``even'' points (where $\sum x_j \equiv 0 \pmod 2$). 
In this scenario, the Fourier argument will eventually find a frequency $\zeta_i$ that is very close to $(1/2, \dots, 1/2)$. 
Rather than rounding this $\zeta_i$ to a high-precision rational (which would require $O(\log n)$ bits to store), we should round it directly to the vector $(1/2, \dots, 1/2)$.\footnote{If we round it to a multiple of $1/Q$ in every coordinate, and place $Q\zeta_i$ as a row in $\mathcal{A}$, the row has each coordinate very close to $Q/2$. When a (small) vector $x$ is sketched by this row, the information $\sum x_j\pmod 2$ can still be recovered. Hence, the distinction is only the space usage.}
By storing the value $\langle \zeta_i, x \rangle \pmod{\mathbb{Z}}$, we capture the ``sum mod 2'' as a single bit of information.

This example suggests that the correct object is not always the real span of the chosen frequencies, but rather the subgroup of $\mathbb T^n$ that they generate. We address the discretization through four points:
\begin{itemize}
	\item \textbf{Translation invariance in the kernel:} If the sketch stores values $\langle \zeta_i, x \rangle \pmod{\mathbb{Z}}$, the ``kernel'' of $\mathcal{A}$ consists of all vectors $v$ such that $\langle \zeta_i, v \rangle \in \mathbb{Z}$ for all $i$. 
	The proof of translation invariance still works as long as any frequency $\zeta$ with a large Fourier coefficient is close to the subgroup generated by the $\zeta_i$. 
	In this case, $v$ is nearly orthogonal, modulo integers, to the non-diffuse directions, and hence $\nu$ should be close to $\tau_v\nu$.

	\item \textbf{Greedy chains:} when a heavy frequency $\zeta_i$ is far from combinations of previous frequencies, we greedily add the chain $\{\zeta_i, 3\zeta_i,\dots, 3^{r-1}\zeta_i\}$ as long as the union remains dissociated. When this process stops, every heavy frequency is already close to a signed combination of the chosen generators, and each chain also gives an approximate relation of the form $k_j \zeta_j\approx v_j$, where $v_j$ is a combination of previously chosen frequencies.

	\item \textbf{Systematic rounding:} We build a finite subgroup $H\le\mathbb T^n$ chain by chain. For a chain with relation $k_j\zeta_j\approx v_j$, we add to $H$ a point $t_j$ close to $\zeta_j$ such that $k_jt_j$ lies exactly in $H$. If no point of $H$ is close to $k_j\zeta_j$, we first round $k_j\zeta_j$ to a fixed fine grid of width $1/\poly(R)$ and add the rounded point to $H$ (a \emph{reset}). The sketch then records the residues $\langle t,x\rangle\pmod{\mathbb Z}$ for the generators $t$ of $H$.

	\item \textbf{Bounding the space:} The dissociation \cref{lem:dissociated-intro} applies to the union of all chains, giving a bound of $O(S)$ on the total number of frequencies in all chains, which is approximately $\Theta(\sum_j \log k_j)$. Moreover, every chain that resets contributes a block of $\Omega(\log R)$ dissociated frequencies, so there are only $O(S/\log R)$ resets. Hence $|H|\le\prod_j k_j\cdot\poly(R)^{O(S/\log R)}=2^{O(S)}$.
\end{itemize}
Since the sketch takes at most $|H|=2^{O(S)}$ values, the exact reduction gives a sketch with $O(S)$ bits, within a constant factor of the space of the original streaming algorithm.

\subsection{Strengthening near the origin: mollified reduction}
\label{subsec:overview-mollified-counting}

For smooth problems, whose value barely changes when a small Gaussian perturbation is added, we strengthen the exact reduction to give a sketch with $O(S/\log R)$ rows and bounded integer entries. The idea is to modify the stream construction: instead of landing directly at the target~$y$, we append an extra independent Gaussian block $Z\sim\gamma_R$ and land at $y+Z$. Smoothness guarantees that the algorithm's output on~$y+Z$ is still a valid answer for~$y$ with high probability, so this modification is harmless.

On the Fourier side, the conditioned convolution measure $\nu=\mu_1*\cdots *\mu_M$ is now convolved with one more copy of $\gamma_R$, so its Fourier transform gains the extra factor
$$
|\widehat{\gamma_R}(\zeta)|\le \exp(-\Omega(R^2\|\zeta\|_{\mathbb T^n}^2)).
$$
This damps all frequencies away from the origin by an exponential factor. Concretely, any frequency $\zeta$ with $\|\zeta\|_{\mathbb T^n}\gg 1/R$ has $|\widehat{\gamma_R}(\zeta)|$ exponentially small in $R^2\|\zeta\|_{\mathbb T^n}^2$, so after multiplication it cannot remain ``heavy'' for the mollified convolution. Thus only the \emph{near-origin} portion of the large spectrum matters, and we can extract a sketching matrix of dimension $O(S/\log R)$.

\paragraph{Volume argument.}
Intuitively, the near-origin large spectrum is low-dimensional for the following reason. Consider the constraint imposed by each heavy frequency. A frequency $\zeta_j$ being heavy (for $\mu$) means that for most $x\sim \mu$, the value $\langle \zeta_j, x\rangle \bmod 1$ is close to a fixed number. In other words, the distribution $\mu$ is concentrated on periodic thin affine subspaces orthogonal to~$\zeta_j$. If there are $d$ independent heavy frequencies near the origin (each not too close to the span of the others), the distribution is forced onto the intersection of $d$ such slabs.

The main point is that near-origin frequencies have very few periods within the Gaussian scale. A frequency $\zeta$ with $\|\zeta\|_{\mathbb T^n}\lesssim 1/R$ has period $\approx 1/\|\zeta\|_{\mathbb T^n} \gtrsim R$, so the Gaussian $\gamma_R$ sees only $O(1)$ periods of the linear form $\langle\zeta,\cdot\rangle$. Consequently, each such direction constrains the Gaussian mass by a multiplicative factor of $R^{-\Theta(1)}$. With $d$ such directions, the Gaussian mass of the intersection of the $d$ slabs drops to $R^{-\Theta(d)}$. The density bound from the $S$-bit streaming algorithm forces this mass to be at least $2^{-O(S)}$, giving $d = O(S/\log R)$.

The formal argument obtains this by strengthening the finite group $H$ of the exact reduction near the origin. The exact reduction places every heavy frequency within distance slightly larger than $1/R$ of $H$. For a heavy frequency $\zeta$ near the origin, a large integer multiple of $\zeta$ is still heavy, so this multiple also lies that close to $H$, and dividing by the multiplier places $\zeta$ much closer to the real subspace spanned by the short lifts of the points of $H$ to $\mathbb R^n$. A volume argument using $|H|=2^{O(S)}$ shows that this subspace has dimension $O(S/\log R)$. The finer distance is needed because the mollified reduction allows targets of norm up to $R^{5/4}$, larger than $R$, so a shift $v$ in the kernel can have norm up to $2R^{5/4}$.

This gives the integer sketch needed for the mollified translation-invariance argument, with entries bounded by $R$ and dimension $O(S/\log R)$, for a total space of $O(S)$ bits.

\subsection{Organization} Section~\ref{sec:preliminaries} defines the Fourier and Gaussian notations and the randomized streaming and sketching models used throughout, and presents a list of the recurring parameters. 
The technical core appears in Sections~\ref{sec:coarse-large-spectrum-general}--\ref{sec:randomized-transfer}.
Section~\ref{sec:coarse-large-spectrum-general} shows that the heavy Fourier frequencies of the conditioned convolution lie near a finite subgroup of $\mathbb T^n$ of order $2^{O(S)}$, which is the structure used by the exact reduction. Near the origin, it strengthens this to a subspace of dimension $O(S/\log R)$, which is the structure used by the mollified reduction.
Section~\ref{sec:translation-invariance-general} converts both structures into translation-invariance statements, and Section~\ref{sec:randomized-transfer} converts this invariance into the two reduction theorems. In both sections, the mollified reduction reuses the argument of the exact reduction, with one more Gaussian block that removes the heavy frequencies far from the origin. \cref{sec:application} uses the exact reduction for the SMP-based applications and the Max-CSP application, and the mollified reduction for the lifting-based applications.

%!TEX root = main.tex

\section{Preliminaries and notations}
\label{sec:preliminaries}

This section fixes the notation and the model of computation used in the rest
of the paper. \Cref{subsec:randomized-conventions-prelim} defines turnstile
streams, randomized streaming algorithms, linear sketches, and the streaming
problems we consider. \Cref{subsec:fourier-prelim} recalls Fourier analysis
on $\mathbb Z^n$, and \cref{subsec:gaussian-prelim} introduces discrete
Gaussians and their dense pieces. Finally, \cref{subsec:global-parameters-prelim} lists the parameters that
appear in several sections.

Unless explicitly stated otherwise, all norms, balls, and diameters in
$\mathbb R^n$ are Euclidean. For a radius $D\ge 0$, we write
\[
B_D\coloneqq\{x\in\mathbb R^n:\|x\|_2\le D\}
\]
for the ball of radius $D$ around the origin; its lattice points and its
nonnegative lattice points are $B_D\cap\mathbb Z^n$ and
$B_D\cap\mathbb Z_{\ge 0}^n$. We write $\log$ for the base-two logarithm and
$\ln$ for the natural logarithm, unless a base is displayed as a subscript.

\subsection{Randomized streaming algorithms and linear sketches}
\label{subsec:randomized-conventions-prelim}

In this paper, all randomized algorithms and sketches are public-coin. The random seed is sampled independently of the input and does not count toward the space of a streaming algorithm or the size of a sketch. Equivalently, a randomized streaming algorithm is a distribution over deterministic streaming algorithms with the same bound on the number of memory states, and a randomized linear sketch is a distribution over deterministic linear sketches with the same bounds on dimension and image size. This convention differs from the standard one in the literature \cite{li2014turnstile,AHLW16}, where the random seed counts toward the space usage. We use the public-coin convention because lower-bound arguments typically fix a hard input distribution and then analyze deterministic algorithms against that distribution.

\begin{definition}[Turnstile streams, unit updates, and strictness]\label{def:turnstile-streams}
A \emph{unit-update turnstile stream} on $\Z^n$ is a finite sequence of updates
\[
(i_1,\Delta_1),\ldots,(i_T,\Delta_T),
\qquad
i_t\in[n],\ \Delta_t\in\{+1,-1\},
\]
followed by a final query. Starting from the zero vector, the frequency vector after the first $t$ updates is
\[
x^{(t)}\coloneqq\sum_{s=1}^t \Delta_s e_{i_s}\in\Z^n.
\]
The final frequency vector is $x^{(T)}$. The stream is called \emph{strict} if every prefix vector $x^{(t)}$ lies in $\Z_{\ge 0}^n$.

A \emph{distribution of streams} is a probability distribution on unit-update turnstile streams. Its induced distribution on $\Z^n$ is the distribution of the final frequency vector at the time of the query.
\end{definition}

Each update increases or decreases one coordinate of the frequency
vector by one, and an answer is required only at the final query.
Throughout the paper, all turnstile streams are unit-update streams, and
the length of a stream is its number of updates $T$.

\begin{definition}[Randomized streaming algorithms]\label{def:randomized-streaming-algorithms}
A \emph{randomized turnstile streaming algorithm using $S$ bits of space} is a public-coin family $\{A_\rho\}_\rho$ of deterministic streaming algorithms such that, for every seed $\rho$, the algorithm $A_\rho$ has at most $2^S$ memory states. Its success probability is computed over the public random seed $\rho$ in the worst case of any input stream.
\end{definition}

Since the seed is free, fixing a seed turns a randomized algorithm using $S$
bits of space into a deterministic algorithm with at most $2^S$ memory states.

\begin{definition}[Deterministic and randomized linear sketches]\label{def:linear-sketches}
A \emph{deterministic linear sketch} on $\Z^n$ is a pair $(\mathcal A,g)$ consisting of a group homomorphism
$
\mathcal A:\Z^n\to \Gamma
$
into an Abelian group $\Gamma$ and a \emph{decoder} $g$ defined on $\mathcal A(\Z^n)$. On input $x\in\Z^n$, the sketch stores the \emph{sketch value} $\mathcal A(x)$ and outputs $g(\mathcal A(x))$. When $\Gamma=\mathbb Z^m$ or $\Gamma=\mathbb T^m$, where $\mathbb T\coloneqq\mathbb R/\mathbb Z$, we call $m$ the \emph{dimension} of the sketch; we call $|\mathcal A(\Z^n)|$ its \emph{image size}.

A \emph{randomized linear sketch} is a public-coin family $\{(\mathcal A_\rho,g_\rho)\}_\rho$ of deterministic linear sketches. When we say that a randomized linear sketch has dimension at most $m$ and image size at most $2^S$, we mean that every deterministic sketch in its support has dimension at most $m$ and image size at most $2^S$. 
\end{definition}

Since $\mathcal A(x\pm e_i)=\mathcal A(x)\pm\mathcal A(e_i)$, the sketch value
of the final frequency vector can be maintained under turnstile updates. When the image
$\mathcal A(\Z^n)$ is finite, storing the sketch value takes
$\lceil\log|\mathcal A(\Z^n)|\rceil$ bits. We also allow the image group
$\Gamma$ to be infinite, such as $\mathbb{Z}^k$. In this case, we may further
require the entries of the integer matrix representing $\mathcal A$ to be
bounded by an entry bound, usually denoted $Q$.

In both models we only  measure space. We do not bound the time that a
streaming algorithm takes to process an update, or the time needed to compute
$\mathcal A$, the sketch value, or the decoder $g$.

We will also use the following general notion of a streaming problem, which includes promise problems and metric approximation problems as special cases.

\begin{definition}[Streaming problems]\label{def:streaming-problems}
A turnstile streaming problem on $\Z^n$ is specified by a set $\mathcal O$ of valid outputs and a requirement relation
$
\mathcal R\subseteq \Z^n\times \mathcal O,
$
where $(x,o)\in \mathcal R$ means that $o$ is an acceptable output on the final frequency vector $x$. A randomized streaming algorithm solves this problem with success probability at least $1-\delta$ if, on any input stream, with probability at least $1-\delta$ over the public randomness, its output $o$ satisfies $(x,o)\in\mathcal R$, where $x$ is the final frequency vector.
\end{definition}

\begin{definition}[Metric approximation problems]\label{def:metric-approximation}
A metric approximation problem on $\Z^n$ is the special case of a streaming problem obtained from a metric space $(\mathcal M,d)$, a target map $f:\Z^n\to \mathcal M\cup\{*\}$, a set of outputs $\mathcal O\subseteq \mathcal M$, and an accuracy parameter $\epsilon\ge 0$. On an input $x$ with $f(x)\in\mathcal M$, the valid outputs are those $z\in\mathcal O$ satisfying
$
d(z,f(x))\le \epsilon,
$
and on an input with $f(x)=*$ every $z\in\mathcal O$ is valid. An algorithm that solves this problem is an $\epsilon$-approximation algorithm for $(f,\mathcal O)$.
\end{definition}

The value $*$ marks inputs outside the promise, on which there is no correctness requirement. A binary promise problem, given by a map $f:\Z^n\to\{0,1,*\}$, is the special case $\mathcal M=\mathcal O=\{0,1\}$ with the discrete metric and $\epsilon=0$. In statements about metric approximation problems, we use the convention that $d(z,f(x))\le\epsilon$ holds whenever $f(x)=*$.

\subsection{Fourier transform on the integer lattice}
\label{subsec:fourier-prelim}

The dual group of $\mathbb Z^n$ is the torus
$\mathbb T^n\cong(\mathbb R/\mathbb Z)^n$. For $t\in\mathbb R/\mathbb Z$, we
write $e(t)\coloneqq e^{2\pi i t}$. For an absolutely summable function
$\nu\in\ell^1(\mathbb Z^n)$, its Fourier transform
$\widehat{\nu}:\mathbb T^n\to\mathbb C$ is
\[
\widehat{\nu}(\zeta)\coloneqq\sum_{x\in \mathbb Z^n}\nu(x)e(-\langle \zeta,x\rangle),
\qquad \zeta\in \mathbb T^n.
\]
When $\nu$ is a probability measure, $\widehat\nu(\zeta)$ is the expectation
of $e(-\langle\zeta,x\rangle)$ for $x\sim\nu$. 

For $u\in\mathbb R/\mathbb Z$, we write $\|u\|_{\mathbb R/\mathbb Z}$ for the
distance from $u$ to the nearest integer. A \emph{lift} of a frequency
$\zeta\in\mathbb T^n$ is a vector $w\in\mathbb R^n$ with
$w\bmod\mathbb Z^n=\zeta$. We write $\widetilde\zeta$ for the shortest lift,
that is, the unique representative of $\zeta$ in $[-1/2,1/2)^n$, and
\[
\|\zeta\|_{\mathbb T^n}\coloneqq\operatorname{dist}(\zeta,\mathbb Z^n)=\|\widetilde\zeta\|_2.
\]
The torus $\mathbb T^n$ carries the quotient metric induced by the Euclidean
metric on $\mathbb R^n$: for $\zeta,\xi\in\mathbb T^n$ and a set
$A\subseteq\mathbb T^n$, we write
\[
\operatorname{dist}_{\mathbb T^n}(\zeta,\xi)\coloneqq\|\zeta-\xi\|_{\mathbb T^n},
\qquad
\operatorname{dist}_{\mathbb T^n}(\zeta,A)\coloneqq\inf_{a\in A}\operatorname{dist}_{\mathbb T^n}(\zeta,a).
\]
We write $d\zeta$ for the normalized Haar measure on $\mathbb T^n$ and use the
usual $\ell^p$ and $L^p$ norms.

\subsection{Discrete Gaussians and dense pieces}
\label{subsec:gaussian-prelim}

For a width $s>0$ and a center $c\in\mathbb R^n$, we write
\[
\rho_{s,c}(x)\coloneqq e^{-\pi\|x-c\|_2^2/s^2},\qquad
\rho_s\coloneqq\rho_{s,0},
\]
and, for a discrete set $A\subseteq\mathbb R^n$,
$\rho_{s,c}(A)\coloneqq\sum_{x\in A}\rho_{s,c}(x)$.

\begin{definition}[Discrete Gaussian]
\label{def:discrete-gaussian-notation-general}
The centered discrete Gaussian of radius $R$ on a lattice $\mathcal L\subseteq\mathbb R^n$ is
\[
\calD_{\mathcal L,R}(x)\coloneqq\frac{\rho_R(x)}{\rho_R(\mathcal L)}\qquad(x\in\mathcal L),
\]
and we write $\gamma_R\coloneqq\calD_{\mathbb Z^n,R}$.
\end{definition}

Thus $\gamma_R$ is the integer analogue of a centered Gaussian with variance
$R^2/(2\pi)$ in each coordinate. Most of the mass of $\gamma_R$ lies within distance
$O(R\sqrt n)$ of the origin. We use the following form.

\begin{fact}[Gaussian moment and tail bounds {\cite[Lemma~1.4]{Banaszczyk93}}]
\label{lem:gaussian-moment-general}
For a lattice $\mathcal L\subseteq\mathbb R^n$, $Y\sim\calD_{\mathcal L,R}$, and $t\ge0$,
\[
\mathbb E\exp\!\left(\frac{\pi\|Y\|_2^2}{2R^2}\right)\le2^{n/2}, \qquad  \Pr[\|Y\|_2>t]\le2^{n/2}e^{-\pi t^2/(2R^2)}.
\]
\end{fact}

We will also use the notion of dense pieces.
\begin{definition}[Dense piece]
\label{def:dense-piece}
Let $\mu,\nu$ be probability measures on a countable set, and let
$\alpha\in(0,1]$. We call $\mu$ an \emph{$\alpha$-dense piece} of $\nu$ if
$\mu(x)\le\alpha^{-1}\nu(x)$ for every $x$. A \emph{Gaussian piece} is a dense
piece of a discrete Gaussian $\gamma_R$.
\end{definition}

For example, conditioning $\nu$
on an event of probability $\alpha$ gives an $\alpha$-dense piece of $\nu$.

\subsection{Global parameters}
\label{subsec:global-parameters-prelim}

Table~\ref{tab:global-parameters} lists the parameters that appear in several
sections. When a symbol is used in a different local role, we restate its
meaning at the point of use.

\begin{table}[htbp]
\centering
\small
\renewcommand{\arraystretch}{1.08}
\begin{tabular}{p{0.16\linewidth}p{0.74\linewidth}}
\toprule
Symbol & Typical role in the paper \\
\midrule
$R$ & Discrete-Gaussian radius. \\
$D$ & Euclidean radius of the target distribution $\mathcal I$ in the reduction theorems. \\
$W$ & Stream-length bound in the applications. \\
$M$ & Number of random prefix blocks, also the number of Gaussian pieces in a convolution. Usually we set $M \coloneqq R^2$. \\
$Q$ & Denominator or entry bound for integer sketches. \\
$S$ & Streaming space in bits, also the logarithmic density $S=\log(2/\alpha)$ of a Gaussian piece. \\
$K,\Lambda_K$ & Large-spectrum threshold and $\Lambda_K(\mu)=\{\zeta\in\mathbb T^n:|\widehat\mu(\zeta)|\ge 1-1/K\}$. \\
$\sigma$ & Sequence of memory states after the random prefix blocks. \\
$\alpha,\beta$ & Densities of Gaussian pieces and conditioned prefix laws. \\
$\kappa,\rho,\Delta$ & Lengths on $\mathbb T^n$: $\kappa$ is the dissociation threshold; $\rho$ is the cutoff for short elements of the finite group $H$; $\Delta$ is the distance from the large spectrum to $H$ or to the near-origin subspace. \\
\bottomrule
\end{tabular}
\caption{Global parameters.}
\label{tab:global-parameters}
\end{table}

\section{Large spectrum of convolutions of Gaussian pieces}
\label{sec:coarse-large-spectrum-general}

Recall from \cref{sec:technical-overview} that, to prove the reduction
theorems, we condition the Gaussian prefix blocks of a stream on the memory
states of a streaming algorithm. The sum of the prefix blocks then follows a convolution of Gaussian
pieces, and the linear sketch is built from the frequencies at which this
convolution has a large Fourier coefficient. The goal of this section is to show that these
frequencies lie close to a simple structure that a linear sketch with $O(S)$
bits can record.

Throughout this section, let $M\ge1$ be the number of pieces and let
$\alpha\in(0,1]$ be a density parameter. Let $\mu_1,\ldots,\mu_M$ be
probability measures on $\mathbb Z^n$ such that each $\mu_i$ is an
$\alpha_i$-dense piece of the discrete Gaussian
$\gamma_R$, where
$\prod_{i=1}^M\alpha_i\ge\alpha^M$.
We write
\begin{equation}
\label{eq:convolution-setup-general}
\nu\coloneqq\mu_1*\cdots*\mu_M,\qquad S\coloneqq\log(2/\alpha)
\end{equation}
for the convolution and its logarithmic density. We reuse the letter $S$ of
the space bound because conditioning on a
likely sequence of memory states of an $S$-bit algorithm gives pieces of
geometric-mean density $2^{-O(S)}$.

We first show that every frequency $\zeta$ with
$|\widehat\nu(\zeta)|$ large lies close to a finite subgroup of
$\mathbb T^n$ of order $2^{O(S)}$. We then strengthen this result near the
origin. We
show that every such frequency close to the origin lies close to a real
subspace of dimension $O(S/\log R)$ spanned by points of $R^{-1}\mathbb Z^n$. The
exact reduction, which applies to every streaming problem, uses the first
result. The mollified reduction, which applies to problems that are
insensitive to a small Gaussian perturbation of the input, uses the strengthened version.

\paragraph{Roadmap.}
We first study a single Gaussian piece $\mu$. In
\cref{subsec:coarse-dissociated-rudin}, we show that the large spectrum of
$\mu$ cannot contain a large dissociated set, that is, a large set of
frequencies without short signed relations. In
\cref{subsec:coarse-phase-dissociation}, we use this bound to place the large
spectrum of $\mu$ close to a finite subgroup of order $2^{O(S)}$, and then
use this subgroup to place the part near the origin close to a real subspace
of dimension $O(S/\log R)$. In
\cref{subsec:coarse-different-pieces}, we pass from a single piece to the
convolution $\nu$ by averaging the symmetrizations of its factors.

\subsection{Dissociated sets in the Gaussian spectrum}
\label{subsec:coarse-dissociated-rudin}
\label{subsec:coarse-fourier-gaussian}

In this subsection, we show that the large spectrum of a Gaussian piece cannot
contain a large dissociated set. We use two analytic results. First,
the Fourier transform of the discrete Gaussian $\gamma_R$ decays like a
Gaussian away from the origin. Second, measures with such decay satisfy a
coarse form of Rudin's inequality. The decay estimate is folklore; see
\cref{app:dg-fourier-decay-general} for a proof.
\begin{lemma}[Fourier decay of the discrete Gaussian]
\label{lem:dg-fourier-decay-general}
For all $R\ge 2$ and $\zeta\in \mathbb T^n$, one has
\begin{equation*}
|\widehat{\gamma_R}(\zeta)|\le \exp\left(-R^2\|\zeta\|_{\mathbb T^n}^2 / 5\right).
\end{equation*}
\end{lemma}

In particular, the Fourier coefficient $|\widehat{\gamma_R}(\zeta)|$ is tiny unless
the frequency $\zeta$ lies within distance $O(1/R)$ of $\mathbb Z^n$. We now
define the large spectrum, which collects the frequencies with large Fourier
coefficients.

\begin{definition}[Large Fourier spectrum]\label{def:large-spectrum-general}
For a probability measure $\mu$ on $\mathbb Z^n$ and a threshold $K\ge 2$, the
\emph{large spectrum} of $\mu$ is
\[
\Lambda_K(\mu)\coloneqq\big\{\zeta\in \mathbb T^n: |\widehat{\mu}(\zeta)|\ge 1-1/K\big\}.
\]
\end{definition}

Intuitively, the set $\Lambda_K(\mu)$ consists of the frequencies $\zeta$ at which the phase
$\langle\zeta,x\rangle$, $x\sim\mu$, is concentrated near a constant. The next fact shows that integer multiples of a
large-spectrum frequency stay in a weaker large spectrum.

\begin{fact}[Integer multiples of a frequency]
\label{fact:integer-multiples-general}
For every probability measure $\mu$ on $\mathbb Z^n$, $\xi\in\mathbb T^n$,
and integer $m$,
\begin{equation*}
1-|\widehat\mu(m\xi)|
\le m^2\bigl(1-|\widehat\mu(\xi)|\bigr).
\end{equation*}
Consequently, if $\xi\in\Lambda_K(\mu)$ and $m\ne0$ is an integer with
$m^2\le K/2$, then $m\xi\in\Lambda_{K/m^2}(\mu)$.
\end{fact}

\begin{proof}
Let $X\sim\mu$, and choose a phase $\theta$ so that
$U\coloneqq e(\langle\xi,X\rangle-\theta)$ satisfies
$\mathbb EU=|\widehat\mu(\xi)|$. Since $|U|=1$, we have
$\frac12|U-1|^2=1-\Re U$, and therefore
$\frac12\mathbb E|U-1|^2=1-|\widehat\mu(\xi)|$. In the same way,
$|\mathbb EU^m|=|\widehat\mu(m\xi)|$ gives
\[
1-|\widehat\mu(m\xi)|\le1-\Re\mathbb EU^m=\tfrac12\mathbb E|U^m-1|^2.
\]
For $m\ge1$, the factorization $U^m-1=(U-1)(U^{m-1}+\cdots+1)$ gives
$|U^m-1|\le m|U-1|$; the case $m<0$ follows by conjugation, and $m=0$ is
trivial. Hence
\[
1-|\widehat\mu(m\xi)|
\le\tfrac12\mathbb E|U^m-1|^2
\le\tfrac{m^2}{2}\mathbb E|U-1|^2
=m^2\bigl(1-|\widehat\mu(\xi)|\bigr).
\]
If $\xi\in\Lambda_K(\mu)$, the right-hand side is at most $m^2/K$, so
$|\widehat\mu(m\xi)|\ge1-m^2/K$, which is the consequence.
\end{proof}

\begin{definition}[$\kappa$-dissociated sets]\label{def:kappa-dissociated-general}
Let $\kappa>0$. A finite set $T=\{\xi_1,\dots,\xi_m\}\subseteq \mathbb T^n$ is called \emph{$\kappa$-dissociated} if for every nonzero vector $\varepsilon=(\varepsilon_1,\dots,\varepsilon_m)\in\{-1,0,1\}^m$ one has
\[
\|\varepsilon_1\xi_1+\cdots+\varepsilon_m\xi_m\|_{\mathbb T^n}\ge \kappa.
\]
\end{definition}

Equivalently, no nonempty signed sum of elements of a $\kappa$-dissociated set
$T$ comes within distance $\kappa$ of $\mathbb Z^n$. Rudin's inequality states
that a sum of characters at dissociated frequencies has sub-Gaussian tails,
as if the characters were independent. We use the following coarse version,
adapted from~\cite[Lemma~4.33]{tao2006additive}. It applies to a probability
measure $\nu$ on $\mathbb Z^n$ with $|\widehat\nu(\zeta)|\le\eta$ at every
frequency $\zeta$ with $\|\zeta\|_{\mathbb T^n}\ge\kappa$; see
\cref{app:coarse-rudin-general} for the proof.

\begin{lemma}[Coarse Rudin inequality]
\label{lem:coarse-rudin-general}
Let $\nu$ be a probability measure on $\mathbb Z^n$, let $\kappa>0$, and assume that
\begin{equation*}
|\widehat{\nu}(\zeta)|\le \eta
\qquad \text{whenever } \|\zeta\|_{\mathbb T^n}\ge \kappa
\end{equation*}
for some $\eta\in[0,1]$. Let $T \subseteq \mathbb T^n$ be finite and $\kappa$-dissociated, and let $c:T\to\mathbb C$. Define
\[
F(x)\coloneqq\Re\sum_{\xi\in T} c(\xi)e(\langle \xi,x\rangle),
\qquad x\in \mathbb Z^n.
\]
Then for every $\sigma\ge 0$,
\begin{equation*}
\mathbb E_{x\sim\nu} e^{\sigma F(x)}
\le \exp\!\left(\frac{\sigma^2}{2}\sum_{\xi\in T}|c(\xi)|^2\right)
+ \eta\exp\!\left(\sigma\sum_{\xi\in T}|c(\xi)|\right).
\end{equation*}
\end{lemma}

We now combine the two lemmas to bound the size of a dissociated subset of the
large spectrum of a Gaussian piece.

\begin{corollary}[$\kappa$-dissociated size bound]
\label{cor:kappa-dissociated-upper-general}
Let $R\ge 2$, $\alpha\in(0,1]$, and $\mu$ be an $\alpha$-dense piece of $\gamma_R$.
Set
$
S\coloneqq\log(2/\alpha)
$
and
$
\kappa\coloneqq 5\sqrt{S}/R.
$
If $A \subseteq \Lambda_{2}(\mu)$ is $\kappa$-dissociated, then
$
|A|\le 14S.
$
\end{corollary}

\begin{proof}
Put $m\coloneqq|A|$ and suppose, for contradiction, that $m\ge14S$. For each
frequency $\xi\in A$, choose a phase $\theta_\xi$ so that
$\mathbb E_\mu e(\langle\xi,x\rangle-\theta_\xi)=|\widehat\mu(\xi)|$, and set
\[
F(x)\coloneqq\sum_{\xi\in A}\cos(2\pi(\langle\xi,x\rangle-\theta_\xi)).
\]
Each summand has $\mu$-mean $|\widehat\mu(\xi)|$, which is at least $1/2$
because $A\subseteq\Lambda_2(\mu)$. Hence $\mathbb E_\mu F\ge m/2$, and for
every $\lambda>0$ Jensen's inequality gives the lower bound
\[
\mathbb E_\mu e^{\lambda F}\ge e^{\lambda\mathbb E_\mu F}\ge e^{\lambda m/2}.
\]

For the upper bound, we pass from the Gaussian piece $\mu$ to $\gamma_R$.
Since $\mu$ is an $\alpha$-dense piece of $\gamma_R$ and $e^{\lambda F}\ge0$,
\[
\mathbb E_\mu e^{\lambda F}\le\alpha^{-1}\,\mathbb E_{\gamma_R}e^{\lambda F}.
\]
By \cref{lem:dg-fourier-decay-general}, every frequency $\zeta$ with
$\|\zeta\|_{\mathbb T^n}\ge\kappa$ satisfies
$|\widehat{\gamma_R}(\zeta)|\le e^{-R^2\kappa^2/5}=e^{-5S}$. We may therefore
apply \cref{lem:coarse-rudin-general} with $\nu=\gamma_R$, $\eta=e^{-5S}$,
the $\kappa$-dissociated set $T=A$, the coefficients
$c(\xi)=e(-\theta_\xi)$, and $\sigma=\lambda$. Since $|c(\xi)|=1$, this gives
\[
\mathbb E_{\gamma_R}e^{\lambda F}\le e^{\lambda^2m/2}+e^{\lambda m-5S}.
\]

Now take $\lambda=4S/m$. Then $\lambda m/2=2S$, $\lambda m-5S=-S$, and
$\lambda^2m/2=8S^2/m\le4S/7$ because $m\ge14S$. Moreover,
$\alpha^{-1}=2^{S-1}<e^{S\ln2}$. Combining the three displays gives
\[
e^{2S}<e^{(\ln2+4/7)S}+e^{(\ln2-1)S}\le e^{1.27S}+1.
\]
This is false for every $S\ge1$, and indeed $S=\log(2/\alpha)\ge1$ since
$\alpha\le1$. This contradiction proves $m<14S$.
\end{proof}

\subsection{Structure of the large spectrum of a Gaussian piece}
\label{subsec:coarse-phase-dissociation}

% !TEX root = main.tex

We now show that the large spectrum of a Gaussian piece lies close to a finite
subgroup of $\mathbb T^n$ of order $2^{O(S)}$, where $S$ is the logarithmic
density of the piece. 

\begin{theorem}[Finite-group structure of the large spectrum]
\label{thm:finite-group-large-spectrum-general}
Fix $0<\theta\le1/3$.  Let $R\ge n^{2/\theta}$ be sufficiently large depending only on
$\theta$, and let $K\ge R^\theta$.
Let $\mu$ be an $\alpha$-dense piece of $\gamma_R$, and put
$S\coloneqq\log(2/\alpha)$.
Then there is a finite subgroup $H\le\mathbb T^n$ with
\[
|H|=2^{O_\theta(S)},
\qquad
\sup_{\xi\in\Lambda_K(\mu)}
\operatorname{dist}_{\mathbb T^n}(\xi,H)\le R^{-1+\theta}.
\]
\end{theorem}

We first sketch the proof, which has two steps. First, a greedy procedure
builds a $\kappa$-dissociated set out of geometric chains
$a_j,3a_j,\ldots,3^{r_j-1}a_j$ whose starting points $a_j$ lie in
$\Lambda_K(\mu)$. By \cref{fact:integer-multiples-general}, if
$a_j\in\Lambda_K(\mu)$, then $ma_j\in\Lambda_{K/m^2}(\mu)$, so the chain
points stay in $\Lambda_2(\mu)$ as long as $3^{2\ell}\le K/2$. Hence the size
bound of \cref{cor:kappa-dissociated-upper-general} applies to the union of
the chains, which therefore has at most $14S$ points. We use powers of $3$
because the signed sums of $1,3,\ldots,3^{r-1}$ are exactly the integers of
absolute value at most $(3^r-1)/2$. Thus, when a chain stops before reaching
the maximal length, the relation that stops it collapses to an approximate
integer relation $k_ja_j\approx v_j$, where $v_j$ is a signed sum of the
earlier chains.  The second step turns the chains into a finite subgroup $H$
of order $2^{O_\theta(S)}$. We describe the procedure after
\cref{prop:unrounded-chain-general}. The only properties of the Gaussian piece $\mu$ that
the proof uses are \cref{fact:integer-multiples-general} and the size bound
of \cref{cor:kappa-dissociated-upper-general}.

For a finite set $A\subseteq\mathbb T^n$, we write
\[
\Sigma(A)\coloneqq\left\{\sum_{a\in A}\varepsilon_a a:
                    \varepsilon_a\in\{-1,0,1\}\right\}
\]
for the set of signed sums of elements of $A$.  The greedy procedure keeps
adding frequencies that are far from the signed sums of the frequencies
chosen so far.

\begin{proposition}[Greedy chains]
\label{prop:unrounded-chain-general}
Let $R\ge2$, $K\ge16$, and let $\mu$ be an $\alpha$-dense piece of
$\gamma_R$.  Put
\[
S\coloneqq\log(2/\alpha),\qquad
\kappa\coloneqq5\sqrt S/R,\qquad
r_*\coloneqq\left\lfloor\tfrac12\log_3 K\right\rfloor.
\]
There are chains $A_1,\ldots,A_J$ of the form
\[
A_j=\{a_j,3a_j,\ldots,3^{r_j-1}a_j\},
\qquad a_j\in\Lambda_K(\mu),\quad 1\le r_j\le r_*,
\]
whose union $A$ is a $\kappa$-dissociated subset of $\Lambda_2(\mu)$
with
$
|A|=\sum_j r_j\le14S$, such that every $\xi\in\Lambda_K(\mu)$ satisfies
$$\operatorname{dist}_{\mathbb T^n}(\xi,\Sigma(A))\le\kappa.$$
Moreover, if $r_j<r_*$, then there are an integer $k_j \in [\tfrac12\,3^{r_j}, \tfrac32\,3^{r_j}]$ and a signed sum
$v_j\in\Sigma(A_1\cup\cdots\cup A_{j-1})$ such that
\[
\|k_ja_j-v_j\|_{\mathbb T^n}<\kappa.
\]
\end{proposition}

\begin{proof}
We construct the chains by the following greedy procedure and index them
$A_1,A_2,\ldots$ in the order in which they are started. We start with
$A=\varnothing$. As long as some frequency $\xi\in\Lambda_K(\mu)$ satisfies
$\operatorname{dist}_{\mathbb T^n}(\xi,\Sigma(A))>\kappa$, start a new chain
$A_j$ at $a_j\coloneqq\xi$: add $a_j$ to $A$, and then, for
$\ell=1,2,\ldots$, add $3^\ell a_j$ to $A$ as long as $\ell<r_*$ and
\[
\operatorname{dist}_{\mathbb T^n}\bigl(3^\ell a_j,\Sigma(A)\bigr)\ge\kappa.
\]
Here $A$ always denotes the current set, which consists of the earlier
chains and the points of $A_j$ added so far.

We now show that $|A|$ never exceeds $14S$, so that the procedure
terminates. First, the set $A$ stays $\kappa$-dissociated; we prove this by
induction on the number of added points. Suppose that $A$ is $\kappa$-dissociated and that we are adding a new point $b$ with $\operatorname{dist}_{\mathbb T^n}(b,\Sigma(A))\ge\kappa$. Then a nonzero signed sum of
elements of $A\cup\{b\}$ either does not involve $b$, and then has norm at
least $\kappa$ by assumption, or has the form $\pm b+s$ with $s\in\Sigma(A)$.
In the latter case, since $\Sigma(A)=-\Sigma(A)$, its norm is at least
$\operatorname{dist}_{\mathbb T^n}(b,\Sigma(A))\ge\kappa$. This completes the induction.

Second, let $r_j\coloneqq|A_j|$, so $1\le r_j\le r_*$. Every point added to $A$ has the form $3^\ell a_j$ with
$a_j\in\Lambda_K(\mu)$ and $0\le\ell<r_*$. Since
$9^\ell\le9^{r_*-1}\le K/9$ by the choice of $r_*$,
\cref{fact:integer-multiples-general} gives
$3^\ell a_j\in\Lambda_{K/9^\ell}(\mu)\subseteq\Lambda_2(\mu)$. Therefore, $A$ is a $\kappa$-dissociated subset of $\Lambda_2(\mu)$, so
\cref{cor:kappa-dissociated-upper-general} gives $|A|\le14S$.

It remains to prove the covering and relation assertions. When the procedure stops, no new chain can be started, that is,
every $\xi\in\Lambda_K(\mu)$ satisfies
$\operatorname{dist}_{\mathbb T^n}(\xi,\Sigma(A))\le\kappa$. For the last assertion, suppose that $r_j<r_*$. Then the
procedure stopped extending $A_j$ because the next multiple $3^{r_j}a_j$
violated the distance condition. At that moment, the current set consisted
of the chains $A_1,\ldots,A_{j-1}$ and the points $a_j,\ldots,3^{r_j-1}a_j$,
so
\[
\operatorname{dist}_{\mathbb T^n}\bigl(3^{r_j}a_j,
\Sigma(A_1\cup\cdots\cup A_{j-1}\cup A_j)\bigr)<\kappa.
\]
Every element of $\Sigma(A_1\cup\cdots\cup A_{j-1}\cup A_j)$ has the form
$\sum_{\ell<r_j}\varepsilon_\ell3^\ell a_j+v_j$ with signs
$\varepsilon_\ell\in\{-1,0,1\}$ and $v_j\in\Sigma(A_1\cup\cdots\cup A_{j-1})$.
Choose such an element within distance $\kappa$ of $3^{r_j}a_j$, and put
\[
k_j\coloneqq3^{r_j}-\sum_{\ell<r_j}\varepsilon_\ell3^\ell\in\mathbb Z.
\]
Then $k_ja_j-v_j=3^{r_j}a_j-\bigl(\sum_{\ell<r_j}\varepsilon_\ell3^\ell a_j+v_j\bigr)$,
so $\|k_ja_j-v_j\|_{\mathbb T^n}<\kappa$. Finally,
$\bigl|\sum_{\ell<r_j}\varepsilon_\ell3^\ell\bigr|\le\sum_{\ell<r_j}3^\ell
=(3^{r_j}-1)/2$, and therefore
\[
\tfrac12\,3^{r_j}<\tfrac12\bigl(3^{r_j}+1\bigr)
\le k_j\le\tfrac12\bigl(3^{r_j+1}-1\bigr)<\tfrac32\,3^{r_j}.
\]
\end{proof}

Given the chains, the proof of \cref{thm:finite-group-large-spectrum-general}
builds the subgroup $H$ chain by chain, replacing each starting point $a_j$
by a nearby point $b_j\in H$ of finite order. For a full chain ($r_j=r_*$),
we round $a_j$ to the grid $Q^{-1}\mathbb Z^n/\mathbb Z^n$ with $Q=\poly(R)$
and insert the rounded point into $H$; this
multiplies $|H|$ by at most $Q$. Since we may take $K\geq R^\theta$, we have
$r_*=\Theta_\theta(\log R)$, and the chains have at most $14S$ points in
total, so there are only $O_\theta(S/\log R)$ full chains, and they
contribute $O_\theta(S)$ to $\log|H|$. For a chain that stops
early ($r_j<r_*$), the relation $k_ja_j\approx v_j$ suggests taking $b_j$ to
be a $k_j$-th root of a point of $H$ close to $k_ja_j$, that is, a point
$b_j$ with $k_jb_j\in H$. This multiplies $|H|$ by at most $k_j$, and
$\sum_j\log k_j=O(\sum_jr_j)=O(S)$.

The difficulty is the propagation of rounding errors. Write
$v_j=\sum_{i<j}c_ia_i$ with integer coefficients
$|c_i|\le\sum_{\ell<r_i}3^\ell<k_i$, and let $v_j'\coloneqq\sum_{i<j}c_ib_i$
be the corresponding point of $H$. If $b_j$ is the $k_j$-th root of $v_j'$
nearest to $a_j$, then $\|a_j-b_j\|_{\mathbb T^n}\le\delta_j/k_j$, where
$\delta_j\coloneqq\operatorname{dist}_{\mathbb T^n}(k_ja_j,v_j')$. Since
$\|k_ja_j-v_j\|_{\mathbb T^n}<\kappa$, we get
\[
\delta_j<\kappa+\sum_{i<j}|c_i|\,\|a_i-b_i\|_{\mathbb T^n}
\le\kappa+\sum_{i<j}\delta_i.
\]
This recursion allows $\delta_j$ to grow like $2^j\kappa$, exponentially in
the number of chains. Hence, the final error would then
be $2^{O(S)}\kappa$ instead of the required $R^{-1+\theta}$. 

To prevent this, we reset the error whenever it becomes large. Concretely,
if $k_ja_j$ is within $3\kappa$ of $H$, we take $b_j$ to be a $k_j$-th root
of a nearest point of $H$; otherwise, we directly round $k_ja_j$ itself to
the grid $Q^{-1}\mathbb Z^n/\mathbb Z^n$, insert the rounded point into $H$
(a \emph{reset}), and take $b_j$ to be a $k_j$-th root of the rounded point. Then every chain has error at most $3\kappa$, and the
accumulated error stays $O(S\kappa)$.

Each reset multiplies $|H|$ by up to $Q=\poly(R)$, so we must show that
there are only $O_\theta(S/\log R)$ resets. For a long chain with at least
$t = \Theta(\log K)$ points, the number of resets is bounded by $O_\theta(S/\log R)$ by counting the size of $|A|$ (similar as the full chain case). For a shorter chain, a reset happens at a
point $p=k_ja_j$ with $3\kappa<\operatorname{dist}_{\mathbb T^n}(p,H)\le O(S\kappa)$,
which is far from $H$ at the scale $\kappa = 5\sqrt{S}/R$ but much closer than
$R^{-1+\theta}$ (assume $S\leq n \log R \ll R^\theta$ by taking $R \gg n^{2 / \theta}$), and the multiples $p,3p,\ldots,3^{t-1}p$ of such a point
still lie in $\Lambda_2(\mu)$. We show that these blocks, over all such
resets with short chains, form a $\kappa$-dissociated set, so that
\cref{cor:kappa-dissociated-upper-general} bounds their number by
$O(S/\log K)$ as well. Since the earlier blocks all lie close to $H$,
dissociation of a new block requires its nonzero signed sums, which are the
multiples $mp$ with $0<|m|<3^t$, to stay far from $H$. Thus we need to show that an
integer multiple $mp$ of a point at distance $\delta$ from $H$ has distance
about $|m|\delta$ from $H$ (as long as $|m|\delta \ll R^{-1+\theta}$).

For an arbitrary finite group this might fail. If $H$ contains an element with a
lift $w$ of length at most $\rho$, where $\rho=R^{-1+\theta}$, then a point $p$ close to $w/m$ can be far from $H$ while $mp$
is close to $w\in H$. 

To remove this obstruction, we add a final step that promotes moderately short lifts to lines. Whenever an element of $H$ has a lift $w$
of length at most $\rho$, we adjoin to $H$ the $N$-th root
$w/N\bmod\mathbb Z^n$. The multiples of $w/N$ form a grid of spacing $\|w\|_2/N\le\rho/N$ on
the line $\mathbb Rw$, and $N$ is chosen so large that this spacing is
negligible (unlike the grid size $Q$, however, $N \ll 1/\rho$). Near the origin, we
may therefore think of $H$ as containing the whole line $\mathbb Rw$. Let $L$
be the real span of the lifts $w$ promoted so far. We promote $w$ only when
$w\notin L$, because otherwise the line $\mathbb Rw$ lies in $L$, which $H$
already covers up to the grid spacing. After the promotion steps, near the
origin the distance to $H$ is essentially the distance to $L$, which grows
linearly under integer multiplication (as formalized in \eqref{eq:finite-grid-multiple-distance}). We use a volume argument to show that
promotion raises $|H|$ at most to the fourth power of the contribution of the
other insertions.

\begin{lemma}[Promoting short lifts to lines]
\label{lem:finite-grid-promotion-general}
Fix $0<\rho<1/100$ and an integer $2\le N\le\rho^{-3/4}$.
Start with $H=\{0\}\le\mathbb T^n$ and $L=\{0\}\le\mathbb R^n$.
After each insertion into $H$ of a point of finite order, repeat the following
promotion step while possible: choose a shortest lift $w$ of an element of
$H$ with $\|w\|_2\le\rho$ and $w\notin L$, and set
\[
L\leftarrow L+\mathbb Rw,
\qquad
H\leftarrow H+\langle w/N\bmod\mathbb Z^n\rangle.
\]
Let $B$ be the product, over
all insertions other than promotion steps, of the index by which the
insertion enlarges $H$. Then $|H|\le B^4$. Moreover, after each promotion
phase, every $p\in\mathbb T^n$ and every integer $m\ne0$ with
$|m|\operatorname{dist}_{\mathbb T^n}(p,H)\le\rho/2$ satisfy
\begin{equation}
\label{eq:finite-grid-multiple-distance}
\operatorname{dist}_{\mathbb T^n}(mp,H)
\ge|m|\Bigl(\operatorname{dist}_{\mathbb T^n}(p,H)-\frac{n\rho}{2N}\Bigr).
\end{equation}
\end{lemma}

\begin{proof}
Let $w_1,\ldots,w_d$ be the chosen lifts. Since such vectors are linearly independent, $d\leq n$. Each promotion step adjoins an $N$-th root of an element
already in $H$, and therefore
$
|H|\le BN^d.
$
For the reverse bound on $|H|$, let $s_i$ be the order of $w_i$ modulo the subgroup
generated by $w_1,\ldots,w_{i-1}$ in $\mathbb T^n$.  Thus
$\prod_i s_i\le|H|$, and there are integers $c_{i\ell}$ such that
\[
z_i\coloneqq s_iw_i-\sum_{\ell<i}c_{i\ell}w_\ell\in\mathbb Z^n.
\]
These integer vectors are linearly independent, so their Gram determinant is
a positive integer and their $d$-dimensional volume is at least one. The
triangular change of basis from $w_1,\ldots,w_d$ to $z_1,\ldots,z_d$ then gives
\begin{equation}
\label{eq:finite-group-volume-general}
1\le\operatorname{vol}_d(z_1,\ldots,z_d)
=\left(\prod_i s_i\right)\operatorname{vol}_d(w_1,\ldots,w_d)
\le |H|\operatorname{vol}_d(w_1,\ldots,w_d)
\le |H|\rho^d.
\end{equation}
By \eqref{eq:finite-group-volume-general}, $\rho^{-d}\le|H|$. Since
$N\le\rho^{-3/4}$, we get $N^d\le\rho^{-3d/4}\le|H|^{3/4}$. Combined with
$|H|\le BN^d$, this gives $|H|^{1/4}\le B$, that is, $|H|\le B^4$.

For the distance bound, we first compare the distance to $H$ with the
distance to $L$ near the origin. We claim that after each promotion phase, every
$x\in\mathbb T^n$ whose shortest lift $\widetilde x$ has norm at most $\rho/2$
satisfies
\begin{equation}
\label{eq:finite-grid-local-distance}
\operatorname{dist}(\widetilde x,L)
\le\operatorname{dist}_{\mathbb T^n}(x,H)
\le\operatorname{dist}(\widetilde x,L)+\frac{n\rho}{2N}.
\end{equation}
Indeed, for the upper bound, project $\widetilde x$ onto $L$ and write the
projection as $\sum_i t_iw_i$.  Rounding each coefficient $t_i$ to
$N^{-1}\mathbb Z$ gives a lift of an element of $H$ within $d\rho/(2N)$ of
this projection. For the lower bound, choose a nearest $h\in H$ and a lift
$w$ of $h$ realizing its distance to $\widetilde x$. Since $0\in H$,
\[
\|\widetilde x-w\|_2=\operatorname{dist}_{\mathbb T^n}(x,H)\le\|\widetilde x\|_2\le\rho/2 \Longrightarrow \|w\|_2\le\rho.
\]
Thus $w$ is a shortest lift of $h$ and, since no further promotion step is
possible, it belongs to $L$. It follows that
$\operatorname{dist}(\widetilde x,L)\le\|\widetilde x-w\|_2$.

Now let $p$ and $m$ be as in the lemma. Choose a nearest $h\in H$ to $p$, let
$e$ be the shortest lift of $p-h$, and put
$\delta\coloneqq\|e\|_2=\operatorname{dist}_{\mathbb T^n}(p,H)$. By
assumption, $\|e\|_2\le\|me\|_2=|m|\delta\le\rho/2$, so $e$ and $me$ are the
shortest lifts of $p-h$ and $mp-mh$. Since $h$ and $mh$ lie in $H$, the lower
bound in \eqref{eq:finite-grid-local-distance} at $mp-mh$ gives
\[
\operatorname{dist}_{\mathbb T^n}(mp,H)
=\operatorname{dist}_{\mathbb T^n}(mp-mh,H)
\ge\operatorname{dist}(me,L)=|m|\operatorname{dist}(e,L),
\]
and the upper bound at $p-h$ gives
$\delta=\operatorname{dist}_{\mathbb T^n}(p-h,H)\le\operatorname{dist}(e,L)+n\rho/(2N)$.
Combining the two bounds proves \eqref{eq:finite-grid-multiple-distance}.
\end{proof}

\begin{proof}[Proof of \cref{thm:finite-group-large-spectrum-general}]
If $S\ge n\log R$, take $q\coloneqq\lceil R\rceil$ and
$H\coloneqq(q^{-1}\mathbb Z/\mathbb Z)^n$.
Then $\log|H|=n\log q=O(S)$, and every point of $\mathbb T^n$ is within
$\sqrt n/(2q)\le R^{-1+\theta}$ of $H$.
We may therefore assume $S<n\log R$.

Without loss of generality assume that $K=R^\theta$.
Take the chains from \cref{prop:unrounded-chain-general}, and use the
dissociation scale $\kappa\coloneqq5\sqrt S/R$, the promotion radius $\rho\coloneqq R^{-1+\theta}$, the promotion
denominator $N\coloneqq\lfloor R^{3\theta/2}\rfloor$, and the grid size $Q\coloneqq\lceil R^3\rceil$.
All rounding below is to the grid $Q^{-1}\mathbb Z^n/\mathbb Z^n$, and $N$
enters only through \cref{lem:finite-grid-promotion-general}.  The
assumptions imply
$
n\le R^{\theta/2},
S\le R^{\theta/2+o(1)}.
$
Since $\theta\le1/3$, we have
$N\le R^{3\theta/2}\le R^{3(1-\theta)/4}=\rho^{-3/4}$.

\medskip\noindent\emph{Constructing the group.}
Process the chains in order, starting with $H=\{0\}$ and running the
promotion steps of \cref{lem:finite-grid-promotion-general} after \emph{every}
insertion (thus whenever we use the current group $H$ below, the promotion
phase after the last insertion is complete, and we may assume that
\eqref{eq:finite-grid-multiple-distance} holds for $H$).
For each chain, we will show how to insert a point $b_j$ close to $a_j$ into $H$ and maintain the
invariant
\begin{equation}
\label{eq:finite-grid-chain-error}
\|3^\ell(a_j-b_j)\|_{\mathbb T^n}\le3\kappa
\qquad(0\le\ell<r_j),
\end{equation}
that is, every point $3^\ell a_j$ of the chain lies within $3\kappa$ of the
point $3^\ell b_j\in H$.

For a full chain ($r_j=r_*$), round $a_j$ to the grid
$Q^{-1}\mathbb Z^n/\mathbb Z^n$ and let $b_j$ be the rounded point. Its order
modulo $H$ is at most $Q$, and
$\|3^\ell(a_j-b_j)\|_{\mathbb T^n}\le3^\ell\sqrt n/(2Q)\le\sqrt{Kn}/(2Q)\le R^{-3+3\theta/4}<\kappa$.

For a chain that stops early ($r_j<r_*$), let $k_j$ and
$v_j=\sum_{i<j}c_ia_i$ be given by \cref{prop:unrounded-chain-general}, and put $v_j'\coloneqq\sum_{i<j}c_ib_i\in H$. Since $c_ia_i$ is a
signed sum of the points $3^\ell a_i$ of chain $i$, the invariant for the
earlier chains gives $\|c_i(a_i-b_i)\|_{\mathbb T^n}\le3\kappa r_i$.
Together with $\|k_ja_j-v_j\|_{\mathbb T^n}<\kappa$ and $|A|\le14S$ from
\cref{prop:unrounded-chain-general},
\begin{equation}
\label{eq:finite-grid-return-distance}
\delta\coloneqq\operatorname{dist}_{\mathbb T^n}(k_ja_j,H)
\le\|k_ja_j-v_j'\|_{\mathbb T^n}
\le\kappa+3\kappa|A|\le43S\kappa.
\end{equation}
We now choose the point $p'\in H$ that $k_jb_j$ will equal.

If $\delta\le3\kappa$, let $p'$ be a point of $H$ nearest to $k_ja_j$, so
that $\|k_ja_j-p'\|_{\mathbb T^n}=\delta\le3\kappa$.
Otherwise, we reset the error by rounding $k_ja_j$ to the grid
$Q^{-1}\mathbb Z^n/\mathbb Z^n$ and insert the rounded point $p'$ into $H$, whose order
modulo $H$ is at most $Q$. After the reset, $p'$ approximates $k_ja_j$ up to
the rounding error alone, $\|k_ja_j-p'\|_{\mathbb T^n}\le\sqrt n/(2Q)<\kappa$. In both cases
$\|k_ja_j-p'\|_{\mathbb T^n}\le3\kappa$. Let $e$ be the shortest lift of
$k_ja_j-p'$, and insert $b_j\coloneqq a_j-e/k_j\bmod\mathbb Z^n$ into $H$. Then
$k_jb_j=p'\in H$, so the order of $b_j$ modulo $H$ is at most $k_j$, and
since $3^\ell\le3^{r_j-1}\le\frac23k_j$ for $\ell<r_j$,
\[
\|3^\ell(a_j-b_j)\|_{\mathbb T^n}
\le\frac{3^\ell}{k_j}\|e\|_2\le\frac23\cdot3\kappa\le3\kappa.
\]
This proves \eqref{eq:finite-grid-chain-error} by induction. Below, for a
chain that resets, we write $p\coloneqq k_ja_j$. Thus $p$ lies at distance
$\delta\in(3\kappa,43S\kappa]$ from $H$ before the reset, and within $\kappa$
of the point $p'\in H$ after it.

\medskip\noindent\emph{Bounding the number of resets.}
Put $t\coloneqq\lfloor r_*/4\rfloor$, so that
$t=\Theta(\log K)$ and $3^t\le3^{r_*/4}\le K^{1/8}$. Call a chain long if
$r_j\ge t$, and short otherwise. A long chain contributes at least $t$ points
to $A$, so at most $|A|/t\le14S/t$ long chains reset. For the short chains,
we use an auxiliary set $W$, which starts empty. Just before a short chain
resets, we append the block
\[
p,3p,\ldots,3^{t-1}p
\]
to $W$. The points lie in $\Lambda_2(\mu)$. Indeed, a short chain satisfies $3^{r_j}<3^t\le K^{1/8}$, so $k_j\le\frac32\,3^{r_j}\le\frac32K^{1/8}$, and
every new point $3^\ell p=3^\ell k_j\,a_j$ with $\ell<t$ has multiplier
$3^\ell k_j\le\frac12K^{1/4}\le\sqrt K/3$. Hence
\cref{fact:integer-multiples-general} gives $3^\ell p\in\Lambda_2(\mu)$, since $a_j \in \Lambda_K(\mu)$.

We show below that $W$ is $\kappa$-dissociated, and hence $|W|\le14S$ by
\cref{cor:kappa-dissociated-upper-general} and there are at most
$14S/t + 14S/t =O(S/\log K)$ (long or short) resets in total.

It remains to show that $W$ stays $\kappa$-dissociated. We prove this by
induction on the number of blocks. After the reset, each point $3^\ell p$ of
the new block has the approximation $3^\ell p'\in H$, with
$\|3^\ell p-3^\ell p'\|_{\mathbb T^n}\le3^\ell\sqrt n/(2Q)\le\sqrt{Kn}/(2Q)$.
Since $H$ only grows, these approximations remain in $H$ when later blocks
are added, so every point of $W$ lies within $\sqrt{Kn}/(2Q)$ of a point of
$H$. Now suppose that the old $W$ is $\kappa$-dissociated, so that it has at
most $14S$ points by
\cref{cor:kappa-dissociated-upper-general}. Replacing each old point in a signed sum by its
approximation gives a point of $H$ and changes the sum by at most
$14S\sqrt{Kn}/(2Q)\le R^{-3+5\theta/4+o(1)}<\kappa$. A nonempty signed sum of points of the
new block equals $mp$ with $0<|m|<3^t\le K^{1/8}$. Hence a signed relation of norm
less than $\kappa$ involving the new block would imply
\begin{equation}
\label{eq:finite-grid-short-relation}
\operatorname{dist}_{\mathbb T^n}(mp,H)<2\kappa
\qquad\text{for some integer }0<|m|<3^t.
\end{equation}
To rule this out, recall that a chain resets only when
$\delta=\operatorname{dist}_{\mathbb T^n}(p,H)>3\kappa$, and that
$\delta\le43S\kappa$ by \eqref{eq:finite-grid-return-distance}. Hence
$|m|\delta\le43K^{1/8}S\kappa\le R^{-1+7\theta/8+o(1)}\le\rho/2$. Moreover,
$N\ge R^{3\theta/2}/2$ gives $n\rho/(2N)\le R^{-1}\le\kappa/5$. Therefore
\eqref{eq:finite-grid-multiple-distance} gives
\[
\operatorname{dist}_{\mathbb T^n}(mp,H)
\ge|m|\Bigl(\delta-\frac{n\rho}{2N}\Bigr)
\ge3\kappa-\frac{\kappa}{5}>2\kappa,
\]
contradicting \eqref{eq:finite-grid-short-relation}. This proves the
induction.

\medskip\noindent\emph{Order and approximation.}
Let $B$ be the product of the orders modulo $H$ of the inserted points,
excluding promotion steps.
Since $k_j\le\frac32\,3^{r_j}$ by \cref{prop:unrounded-chain-general}, we
have $\log k_j\le3r_j$, so all root insertions together contribute at most
$\sum_j\log k_j\le3\sum_jr_j\le42S$ to $\log B$.
A full chain contributes $r_*$ points to $A$, so there are at most
$|A|/r_*\le14S/r_*$ full chains, and there are at most $28S/t$ resets. Each
full chain and each reset inserts a point of order at most $Q$ modulo $H$.
Since $r_*,t=\Theta(\log K)$, $\log K=\theta\log R$, and $\log Q=O(\log R)$,
\[
\log B\le42S+\Bigl(\frac{14S}{r_*}+\frac{28S}{t}\Bigr)\log Q
=O_\theta(S).
\]
\Cref{lem:finite-grid-promotion-general} gives $\log|H|\le4\log B=O_\theta(S)$.

Finally, the covering assertion of \cref{prop:unrounded-chain-general} and
\eqref{eq:finite-grid-chain-error} give, for every $\xi\in\Lambda_K(\mu)$,
\[
\operatorname{dist}_{\mathbb T^n}(\xi,H)
\le\kappa+3\kappa|A|
\le43S\kappa=215\frac{S^{3/2}}R
\le R^{-1+\theta}.
\]
The last inequality follows from
$S^{3/2}\le R^{3\theta/4+o(1)}$ and $3\theta/4<\theta$.
\end{proof}

For the mollified reduction, we strengthen this structure near the origin, that is,
for the frequencies $\zeta\in\Lambda_K(\mu)$ with
$\|\zeta\|_{\mathbb T^n}\le2R^{-1+\theta}$: we want them to lie close to a
low-dimensional real subspace, at a distance much smaller than the error
$R^{-1+\theta}$ of \cref{thm:finite-group-large-spectrum-general}. For such
a frequency $\zeta$, a suitable integer multiple $m\zeta$ still lies in the
large spectrum by \cref{fact:integer-multiples-general}, and hence lies close
to the finite subgroup $H$. Dividing by $m$ then reduces the error by the
factor $m$. A volume argument shows that the short elements of $H$ span a
subspace of dimension only $O_\theta(S/\log R)$.

\begin{corollary}[Near-origin structure of the large spectrum]
\label{cor:near-origin-large-spectrum-general}
Fix $0<\theta\le1/4$.  Let $R\ge n^{2/\theta}$ be sufficiently large depending only on
$\theta$, let $R^\theta\le K\le R^{2-2\theta}$, and let
$Q\ge R/2$ be an integer. Let $\mu$ be an $\alpha$-dense piece of $\gamma_R$,
and put $S\coloneqq\log(2/\alpha)$. There are vectors
\[
\eta_1,\ldots,\eta_d\in Q^{-1}\mathbb Z^n\cap[-1/2,1/2)^n,
\qquad d=O_\theta(S/\log R),
\]
whose real span $W\subseteq\mathbb R^n$ satisfies the following: for every
$\zeta\in\Lambda_K(\mu)$ with $\|\zeta\|_{\mathbb T^n}\le2R^{-1+\theta}$,
the Euclidean distance in $\mathbb R^n$ from the shortest lift
$\widetilde\zeta$ to $W$ satisfies
\[
\operatorname{dist}(\widetilde\zeta,W)\le\frac{10R^{3\theta/2}}{R\sqrt K}.
\]
\end{corollary}

\begin{proof}
By \cref{thm:finite-group-large-spectrum-general} with threshold $R^\theta$,
there is a group $H$ of order $2^{O_\theta(S)}$ such that every point of
$\Lambda_{R^\theta}(\mu)$ lies within $R^{-1+\theta}$ of $H$. Let
$m\coloneqq\lfloor\sqrt{K/R^\theta}\rfloor$. Since
$R^\theta\le K\le R^{2-2\theta}$, we have $1\le m\le R^{1-3\theta/2}$, and
hence $m\ge\frac12\sqrt{K/R^\theta}$.

Fix a frequency $\zeta$ as in the statement. Applying
\cref{fact:integer-multiples-general} and using $m^2\le K/R^\theta$, we get
\[
1-|\widehat\mu(m\zeta)|
\le m^2\bigl(1-|\widehat\mu(\zeta)|\bigr)
\le m^2/K\le R^{-\theta},
\]
so $m\zeta\in\Lambda_{R^\theta}(\mu)$ and $m\zeta$ lies within
$R^{-1+\theta}$ of some $h\in H$. The vector $m\widetilde\zeta$ is a lift of
$m\zeta$ of norm at most $2mR^{-1+\theta}\le2R^{-\theta/2}$, so there is a
lift $\widetilde h$ of $h$ with
$\|m\widetilde\zeta-\widetilde h\|_2\le R^{-1+\theta}$ and hence $\|\widetilde h\|_2\le2R^{-\theta/2}+R^{-1+\theta}\le3R^{-\theta/2}<1/2$.
In particular, $\widetilde h$ is the shortest lift of $h$.

Let $V$ be the real span of the short lifts $\widetilde h$ obtained in this
way from all frequencies $\zeta$ in the statement, and let $d\coloneqq\dim V$.
Choose a basis $w_1,\ldots,w_d$ of $V$ among these lifts. The proof of
\eqref{eq:finite-group-volume-general} uses only that $w_1,\ldots,w_d$ are
linearly independent lifts of elements of $H$. Applied to the present
vectors $w_i$, each of norm at most $3R^{-\theta/2}$, it gives
$1\le|H|\,(3R^{-\theta/2})^d$. Since $\log|H|=O_\theta(S)$, this gives
$d=O_\theta(S/\log R)$. Moreover, $\widetilde h\in V$ implies
$\operatorname{dist}(m\widetilde\zeta,V)\le R^{-1+\theta}$, and dividing by
$m\ge\frac12\sqrt{K/R^\theta}$ gives
\begin{equation}
\label{eq:near-origin-amplification-general}
\operatorname{dist}(\widetilde\zeta,V)
\le\frac{R^{-1+\theta}}m\le\frac{2R^{3\theta/2}}{R\sqrt K}.
\end{equation}

It remains to replace the subspace $V$ by a subspace spanned by points of
$Q^{-1}\mathbb Z^n$. Let $u_1,\ldots,u_d$ be an orthonormal basis of $V$,
and round each $u_i/4$ coordinatewise to a nearest point
$\eta_i\in Q^{-1}\mathbb Z^n$. Each coordinate changes by at most
$1/(2Q)\le1/R<1/4$, so $\|u_i/4-\eta_i\|_2\le\sqrt n/R$, and since every
coordinate of $u_i/4$ lies in $[-1/4,1/4]$, the points $\eta_i$ lie in
$[-1/2,1/2)^n$. Write
$w\coloneqq\operatorname{proj}_V\widetilde\zeta=\sum_i c_i u_i$. Then
$\|c\|_2=\|w\|_2\le\|\widetilde\zeta\|_2\le2R^{-1+\theta}$, so
$\sum_i|c_i|\le2R^{-1+\theta}\sqrt d$, and
\begin{equation}
\label{eq:near-origin-rounding-general}
\left\|w-4\sum_i c_i\eta_i\right\|_2
\le4\sum_i|c_i|\,\frac{\sqrt n}{R}
\le\frac{8R^{-1+\theta}\sqrt{nd}}R
\le8R^{-2+3\theta/2}
\le\frac{8R^{3\theta/2}}{R\sqrt K},
\end{equation}
where the last two steps use $d\le n\le R^{\theta/2}$ and
$\sqrt K\le R^{1-\theta}\le R$.
The vector $4\sum_ic_i\eta_i$ lies in the span $W$ of the $\eta_i$, so the
triangle inequality, \eqref{eq:near-origin-amplification-general}, and
\eqref{eq:near-origin-rounding-general} give
$\operatorname{dist}(\widetilde\zeta,W)\le10R^{3\theta/2}/(R\sqrt K)$.
\end{proof}

\subsection{Convolution spectra}
\label{subsec:coarse-different-pieces}

We now return to the convolution $\nu=\mu_1*\cdots*\mu_M$ from
\eqref{eq:convolution-setup-general}. For a convolution of $M$ pieces, the
natural threshold for a large Fourier coefficient is multiplicative. Concretely, we
consider the frequencies $\zeta$ with $|\widehat\nu(\zeta)|\ge e^{-M/K}$ for a
threshold $K$. We reduce to a single dense piece by averaging the
symmetrizations $\mu_i*\widetilde\mu_i$, where
$\widetilde\mu_i(x)\coloneqq\mu_i(-x)$ and using the AM--GM inequality. Since each $\mu_i*\widetilde\mu_i$ is dominated by a
multiple of $\gamma_R*\gamma_R$, we need the following comparison, which says
that $\gamma_R*\gamma_R$ is a $(1/4)$-dense piece of the wider discrete Gaussian
$\gamma_{\sqrt2R}$; see \cref{app:gamma-convolution-vs-gamma-sqrt2R-general}
for the proof.

\begin{fact}[Convolution of two discrete Gaussians]
\label{fact:gamma-convolution-vs-gamma-sqrt2R-general}
If $R\ge\sqrt{(2/\pi)\ln(8n)}$, then
\[
\Gamma_R(x)\coloneqq(\gamma_R*\gamma_R)(x)\le4\gamma_{\sqrt2R}(x)
\qquad(x\in\mathbb Z^n).
\]
\end{fact}

\begin{lemma}[Symmetrization]
\label{lem:heterogeneous-symmetrization-general}
Let $K\ge8$ and $R\ge\sqrt{(2/\pi)\ln(8n)}$. Let $\alpha\in(0,1]$, let
$\mu_1,\ldots,\mu_M$ be probability measures on $\mathbb Z^n$ such that each
$\mu_i$ is an $\alpha_i$-dense piece of $\gamma_R$, where
$\prod_{i=1}^M\alpha_i\ge\alpha^M$, and put $\nu\coloneqq\mu_1*\cdots*\mu_M$.
There is a set $I\subseteq[M]$
of at least $M/2$ indices such that the probability measure
\[
\mu_{\rm sym}\coloneqq\frac1{|I|}\sum_{i\in I}\mu_i*\widetilde\mu_i
\]
is an $(\alpha^4/4)$-dense piece of $\gamma_{\sqrt2R}$ and satisfies
\[
|\widehat\nu(\zeta)|\ge e^{-M/K}
\quad\Longrightarrow\quad
\widehat{\mu_{\rm sym}}(\zeta)\ge1-\frac4K.
\]
\end{lemma}

\begin{proof}
Let $I\coloneqq\{i:\alpha_i\ge\alpha^2\}$. If $\alpha=1$, then every
$\alpha_i=1$ and $I=[M]$. Otherwise, taking logarithms in
$\prod_{i=1}^M\alpha_i\ge\alpha^M$ gives
$\sum_{i=1}^M\ln(1/\alpha_i)\le M\ln(1/\alpha)$. Every
index $i\notin I$ has $\ln(1/\alpha_i)>2\ln(1/\alpha)$, so by Markov's
inequality fewer than $M/2$ indices lie outside $I$. Thus $|I|\ge M/2$.

Since $\widehat{\widetilde\mu_i}=\overline{\widehat\mu_i}$, the Fourier
transform of $\mu_i*\widetilde\mu_i$ is $|\widehat\mu_i|^2$. Suppose that
the frequency $\zeta$ satisfies $|\widehat\nu(\zeta)|\ge e^{-M/K}$. By AM--GM,
and since $|\widehat\mu_i(\zeta)|\le1$ for the indices $i\notin I$,
\[
\widehat{\mu_{\rm sym}}(\zeta)
=\frac1{|I|}\sum_{i\in I}|\widehat\mu_i(\zeta)|^2
\ge\Bigl(\prod_{i\in I}|\widehat\mu_i(\zeta)|^2\Bigr)^{1/|I|}
\ge\Bigl(\prod_{i=1}^M|\widehat\mu_i(\zeta)|^2\Bigr)^{1/|I|}
=|\widehat\nu(\zeta)|^{2/|I|}.
\]
Using $|I|\ge M/2$, the right-hand side is at least
$e^{-2M/(K|I|)}\ge e^{-4/K}\ge1-4/K$.

For the density bound, fix $i\in I$ and $x\in\mathbb Z^n$. Since
$\mu_i\le\alpha_i^{-1}\gamma_R$ pointwise and $\gamma_R$ is symmetric,
\[
(\mu_i*\widetilde\mu_i)(x)=\sum_{y\in\mathbb Z^n}\mu_i(y)\mu_i(y-x)
\le\alpha_i^{-2}\sum_{y\in\mathbb Z^n}\gamma_R(y)\gamma_R(x-y)
=\alpha_i^{-2}\Gamma_R(x)\le\alpha^{-4}\Gamma_R(x),
\]
where the last step uses $\alpha_i\ge\alpha^2$. Averaging over $i\in I$ and
applying \cref{fact:gamma-convolution-vs-gamma-sqrt2R-general} gives
$\mu_{\rm sym}\le4\alpha^{-4}\gamma_{\sqrt2R}$, that is, $\mu_{\rm sym}$ is an
$(\alpha^4/4)$-dense piece of $\gamma_{\sqrt2R}$.
\end{proof}

Applying \cref{thm:finite-group-large-spectrum-general} and
\cref{cor:near-origin-large-spectrum-general} to the symmetrized piece
$\mu_{\rm sym}$ gives the finite-group structure of the convolution spectrum
and its strengthening near the origin.

\begin{corollary}[Structure of a convolution spectrum]
\label{cor:coarse-dissociated-convolution-general}
Fix $0<\theta\le1/4$.  Let $R\ge n^{2/\theta}$ be an integer that is
sufficiently large depending only on $\theta$, and let $K\ge8R^\theta$. Let
$\alpha\in(0,1]$, let $\mu_1,\ldots,\mu_M$ be probability measures on
$\mathbb Z^n$ such that each $\mu_i$ is an $\alpha_i$-dense piece of
$\gamma_R$, where $\prod_{i=1}^M\alpha_i\ge\alpha^M$, and put
$\nu\coloneqq\mu_1*\cdots*\mu_M$ and $S\coloneqq\log(2/\alpha)$.
\begin{enumerate}
\item There is a finite subgroup $H\le\mathbb T^n$ of order $2^{O_\theta(S)}$
such that
\[
|\widehat\nu(\zeta)|\ge e^{-M/K}
\quad\Longrightarrow\quad
\operatorname{dist}_{\mathbb T^n}(\zeta,H)\le R^{-1+\theta}.
\]
\item If moreover $K\le R^{2-2\theta}$, there are vectors
\[
\eta_1,\ldots,\eta_d\in R^{-1}\mathbb Z^n\cap[-1/2,1/2)^n,
\qquad d=O_\theta(S/\log R),
\]
whose real span $W\subseteq\mathbb R^n$ satisfies the following: for every
$\zeta\in\mathbb T^n$ with $|\widehat\nu(\zeta)|\ge e^{-M/K}$ and
$\|\zeta\|_{\mathbb T^n}\le R^{-1+\theta}$, the Euclidean distance in
$\mathbb R^n$ from the shortest lift $\widetilde\zeta$ to $W$ satisfies
\[
\operatorname{dist}(\widetilde\zeta,W)\le\frac{20R^{3\theta/2}}{R\sqrt K}.
\]
\end{enumerate}
\end{corollary}

\begin{proof}
Put $R'\coloneqq\sqrt2R$. By \cref{lem:heterogeneous-symmetrization-general},
every frequency $\zeta$ with $|\widehat\nu(\zeta)|\ge e^{-M/K}$ belongs to
$\Lambda_{K/4}(\mu_{\rm sym})$, where $\mu_{\rm sym}$ is an
$(\alpha^4/4)$-dense piece of $\gamma_{R'}$. Its logarithmic density is
$\log(8/\alpha^4)=4\log(1/\alpha)+3=4S-1=O(S)$, and
\[
R'\ge n^{2/\theta},\qquad K/4\ge2R^\theta\ge(R')^\theta.
\]
Hence \cref{thm:finite-group-large-spectrum-general} applies to
$\mu_{\rm sym}$ with radius $R'$, threshold $K/4$, and the same $\theta$,
giving the group of part~(1) with order $2^{O_\theta(S)}$ and approximation
error $(R')^{-1+\theta}\le R^{-1+\theta}$.

For part~(2), we also have $K/4\le R^{2-2\theta}\le(R')^{2-2\theta}$, and
$Q\coloneqq R\ge R'/2$ is an integer, so
\cref{cor:near-origin-large-spectrum-general} applies to $\mu_{\rm sym}$
with radius $R'$, threshold $K/4$, and grid size $Q=R$. A frequency $\zeta$ as
in part~(2) lies in $\Lambda_{K/4}(\mu_{\rm sym})$ and satisfies
$\|\zeta\|_{\mathbb T^n}\le R^{-1+\theta}\le2(R')^{-1+\theta}$. The corollary
therefore gives vectors
$\eta_1,\ldots,\eta_d\in R^{-1}\mathbb Z^n\cap[-1/2,1/2)^n$ with
$d=O_\theta(S/\log R')=O_\theta(S/\log R)$ and
\[
\operatorname{dist}(\widetilde\zeta,W)
\le\frac{10(R')^{3\theta/2}}{R'\sqrt{K/4}}
=\frac{20(R')^{3\theta/2}}{R'\sqrt K}
\le\frac{20R^{3\theta/2}}{R\sqrt K},
\]
where the last step uses $R'\ge R$ and $3\theta/2<1$.
\end{proof}

Part~(1) of
\cref{cor:coarse-dissociated-convolution-general} places the large spectrum
of the convolution $\nu$ near a finite subgroup of $\mathbb T^n$ of order
$2^{O(S)}$. Part~(2) strengthens this near the origin, where it places the
large spectrum near a real subspace spanned by $O(S/\log R)$ points of
$R^{-1}\mathbb Z^n$.

\section{Translation invariance}
\label{sec:translation-invariance-general}

In this section, we turn the finite-group structure of the large spectrum
from \cref{sec:coarse-large-spectrum-general}, and its strengthening near the
origin, into translation invariance. We first fix notation for translates. For a
probability measure $\nu$ on $\mathbb Z^n$ and a translation vector
$v\in\mathbb Z^n$, we write
\[
\tau_v\nu(x)\coloneqq\nu(x-v)
\]
for the translate of $\nu$ by $v$, so that
\[
\dTV(\nu,\tau_v\nu)=\frac12\sum_{x\in\mathbb Z^n}|\nu(x)-\nu(x-v)|.
\]
We show that the convolution $\nu$ changes little in total variation when it
is translated by a short integer vector $v$ that annihilates the frequencies
defining these structures, that is, $\langle t,v\rangle\in\mathbb Z$ for every
such frequency $t$. The reason is that these structures control the large
spectrum. In other words, every large Fourier coefficient of $\nu$ sits at a frequency $\zeta$
close to some such frequency $t$. Since $v$ is short,
$\langle\zeta,v\rangle$ is then close to $\langle t,v\rangle\in\mathbb Z$.

In \cref{subsec:translation-cutoff-general}, we show that
$\dTV(\nu,\tau_v\nu)$ is small whenever $\langle\zeta,v\rangle$ is close to
$\mathbb Z$ at every large frequency $\zeta$ and most of the mass of $\nu$
lies on few lattice points. In \cref{subsec:sketch-from-tv-general}, we
verify both conditions, first for the convolution $\nu$ and then for its
convolution $\nu*\gamma_R$ with one more discrete Gaussian, where the
strengthening near the origin gives a stronger bound.

\subsection{From Fourier bounds to total variation}
\label{subsec:translation-cutoff-general}

The following lemma bounds the total variation distance
$\dTV(\nu,\tau_v\nu)$ when the distribution $\nu$ on $\mathbb{Z}^n$ and the translation vector $v\in \mathbb{Z}^n$ satisfy that the Fourier transform $\widehat\nu$ is small at
every frequency $\zeta$ for which $\langle\zeta,v\rangle$ is not close to
$\mathbb Z$.

\begin{lemma}[Translation invariance from Fourier bounds]
\label{lem:translation-cutoff-general}
Let $\nu$ be a probability measure on $\mathbb Z^n$, let
$v\in\mathbb Z^n$, and let $u,\eta\ge0$. Suppose that
\[
|\widehat\nu(\zeta)|\le\eta
\qquad\text{whenever}\qquad
\|\langle\zeta,v\rangle\|_{\mathbb R/\mathbb Z}>u.
\]
If a set $E\subseteq\mathbb Z^n$ has $|E|\le N$ and
$\nu(E)\ge1-\varepsilon$, where $0\le\varepsilon\le1$, then
\begin{equation}
\label{eq:fourier-cutoff-tv-general}
\dTV(\nu,\tau_v\nu)
\le20u+\sqrt{2N}\,\eta+\varepsilon.
\end{equation}
\end{lemma}

We explain the idea of the proof. 
On each line $x+\mathbb Zv$, $\nu$ restricts to the measure $\nu_{x,v}:t\mapsto\nu(x+tv)$ on
$\mathbb Z$, whose Fourier transform at $t\in\mathbb R/\mathbb Z$ is an
average of $\widehat\nu$, with phases $e(\langle\zeta,x\rangle)$, over the
frequencies $\zeta$ with $\langle\zeta,v\rangle\equiv t$. The hypothesis
thus says that these one-dimensional transforms are at most $\eta$ outside
the arc $\|t\|_{\mathbb R/\mathbb Z}\le u$. In \cref{subsec:overview-fourier}, we
bounded the $\ell^1$ distance on each line by Parseval's identity and
Cauchy--Schwarz over the effective support of $\nu_{x,v}$. Here we proceed
differently to obtain a better bound, in which the term of order $u$ does
not depend on the size of the support.
On every line, we split the
restriction $\nu_{x,v}$ into a low-frequency part $\nu_{x,v}*k$ and a remainder
$\nu_{x,v}-\nu_{x,v}*k$, where $k$ is a difference of two Fej\'er kernels on
$\mathbb Z$ whose Fourier transform is a trapezoid equal to $1$ on the arc.
The weights of $k$ are spread over about $1/u$ consecutive integers and
change by $O(u)$ in total from one integer to the next, so shifting $\nu_{x,v}*k$
by one changes it by $O(u)\|\nu_{x,v}\|_1$ in $\ell^1$. Summed over all lines,
this gives the first term in \eqref{eq:fourier-cutoff-tv-general}. The
Fourier transform of the remainder vanishes on the arc and is at most $\eta$
elsewhere, so Parseval's identity and Cauchy--Schwarz on the at most $2N$
points of $E\cup(E+v)$ give the second term, and the mass of $\nu$ outside
$E$ gives the third. In the proof below, we do all lines at once by
convolving $\nu$ directly with the weights $k(j)$ placed at the points $jv$.

\begin{proof}
The cases $v=0$ and $u\ge1/4$ are immediate. If $u=0$, the hypothesis also
holds with every $u'>0$ in place of $u$, and letting $u'\to0$ in the bound
for $u'$ gives the claim. For $0<u<1/4$, we build a
trapezoid function $h$ on $\mathbb R/\mathbb Z$ from the triangle functions
\[
\phi_a(t)\coloneqq\max\left\{0,\,1-\frac{\|t\|_{\mathbb R/\mathbb Z}}a\right\}
\quad(0<a\le1/2),
\qquad h\coloneqq2\phi_{2u}-\phi_u.
\]
Then $0\le h\le1$ and $h(t)=1$ for $\|t\|_{\mathbb R/\mathbb Z}\le u$.
The Fourier coefficients of $\phi_a$ are
\[
P_a(j)\coloneqq af(aj),\qquad
f(t)\coloneqq\left(\frac{\sin(\pi t)}{\pi t}\right)^2,\qquad f(0)\coloneqq1,
\]
with $P_a\ge0$ and $\sum_jP_a(j)=1$.
Since $f(t)=\int_{-1}^1(1-|s|)e(st)\,ds$, differentiating under the
integral gives $|f'(t)|\le2\pi\int_{-1}^1|s|(1-|s|)\,ds=2\pi/3$ for every
$t$. For $|t|>1$, differentiating $f$ directly gives the sharper bound
$|f'(t)|\le1/(\pi t^2)+2/(\pi^2|t|^3)$.
Hence $\|f'\|_1\le4\pi/3+2/\pi+2/\pi^2<8$. Writing each difference
$P_a(j)-P_a(j-1)$ as an integral of $af'$ over an interval of length $a$,
we obtain
\[
\sum_j|P_a(j)-P_a(j-1)|\le a\|f'\|_1<8a.
\]
Place the coefficients $k(j)\coloneqq2P_{2u}(j)-P_u(j)$ at the lattice points
$jv$, viewing $k$ as a signed measure on $\mathbb Z^n$. These points are
distinct since $v\ne0$. Then $\|k\|_1\le2+1=3$, because each $P_a$ is a
probability vector. Since moreover $P_a$ is even and $\phi_a$ is continuous,
Fourier inversion gives $\widehat{P_a}=\phi_a$ on $\mathbb R/\mathbb Z$,
where $P_a$ is viewed as a measure on $\mathbb Z$. Hence, with
$t=\langle\zeta,v\rangle$,
$\widehat k(\zeta)=\sum_jk(j)e(-jt)=2\widehat{P_{2u}}(t)-\widehat{P_u}(t)=h(t)$.
Finally,
\[
\|k-\tau_vk\|_1=\sum_j|k(j)-k(j-1)|\le2\cdot8(2u)+8u=40u.
\]

Write $\delta_y$ for the point mass at $y$. Since
$\nu-\tau_v\nu=(\delta_0-\delta_v)*\nu$, we can decompose
\[
\nu-\tau_v\nu=\nu*(k-\tau_vk)+r,
\qquad
r\coloneqq(\delta_0-\delta_v)*(\delta_0-k)*\nu.
\]
The first summand has $\ell^1$ norm at most $\|k-\tau_vk\|_1\le40u$. For the
second, with $t=\langle\zeta,v\rangle$,
\[
\widehat r(\zeta)
=(1-e(-t))(1-h(t))\widehat\nu(\zeta).
\]
This vanishes when $\|t\|_{\mathbb R/\mathbb Z}\le u$, since then $h(t)=1$.
Otherwise $|1-e(-t)|\le2$, $|1-h(t)|\le1$, and $|\widehat\nu(\zeta)|\le\eta$
by hypothesis, so $|\widehat r(\zeta)|\le2\eta$. Parseval's identity
\cite[Section~1.6]{Rudin62} gives $\|r\|_2\le2\eta$.

Set $F\coloneqq E\cup(E+v)$, so $|F|\le2N$. Cauchy--Schwarz on $F$ gives
\[
\sum_{x\in F}|\nu(x)-\nu(x-v)|
\le40u+\sum_{x\in F}|r(x)|\le40u+\sqrt{|F|}\,\|r\|_2
\le40u+2\sqrt{2N}\,\eta.
\]
On $F^c$, the sum is at most $2\varepsilon$, since both $E$ and $E+v$
are contained in $F$. Dividing by two proves the claim.
\end{proof}

The hypothesis of the lemma follows from the structure of the large spectrum.
Suppose that every frequency $\zeta$ with $|\widehat\nu(\zeta)|\ge\eta$ has
distance at most $\Delta$ from a set $W\subseteq\mathbb T^n$, and that the
translation vector $v$ annihilates $W$, that is, $\langle w,v\rangle\in\mathbb Z$
for every $w\in W$. Then every such frequency $\zeta$ satisfies
\[
\|\langle\zeta,v\rangle\|_{\mathbb R/\mathbb Z}
\le\|v\|_2\operatorname{dist}_{\mathbb T^n}(\zeta,W)
\le\|v\|_2\Delta.
\]
Thus \cref{lem:translation-cutoff-general} applies with
$u=\|v\|_2\Delta$.

\subsection{Linear sketches from low total variation}
\label{subsec:sketch-from-tv-general}

We now apply \cref{lem:translation-cutoff-general} to convolutions of
dense pieces of a truncated discrete Gaussian. We begin with the convolution $\nu$ itself, which is
the case of the exact reduction.
\begin{corollary}
\label{cor:translation-invariance-from-convolution-spectrum-general}
Fix $0<\theta\le1/4$. Let $R\ge n^{2/\theta}$ be an integer that is
sufficiently large depending only on $\theta$, put $M\coloneqq R^2$ and
$T_R\coloneqq100R(\sqrt n+\ln R)$, and let $\gamma_R^{\rm tr}$ be $\gamma_R$
conditioned on the ball $B_{T_R}$. Let $\alpha,\alpha_1,\dots,\alpha_M\in(0,1]$
satisfy $\prod_{i=1}^M \alpha_i\ge \alpha^M$, and assume
$S\coloneqq\log(2/\alpha)\le n\log R$. Let each $\mu_i$ be an $\alpha_i$-dense piece of $\gamma_R^{\rm tr}$. Write
$
\nu\coloneqq\mu_1*\cdots *\mu_M.
$
Then there is a finite subgroup $H\le\mathbb T^n$ with
\[
\log|H|=O_\theta(S),
\]
generated by frequencies $t_1,\dots,t_m\in\mathbb T^n$ with $m=O_\theta(S)$,
such that every $v\in\mathbb Z^n$ satisfying
$
\langle v,t_j\rangle\in \mathbb Z$
for every $j\in[m]$ obeys
\[
\dTV(\nu,\tau_v\nu)
\le20\|v\|_2R^{-1+\theta}+R^{-1}.
\]
\end{corollary}

\begin{proof}
We first pass from the truncated Gaussian to $\gamma_R$. Since
$(\sqrt n+\ln R)^2\ge n+\ln^2R$, the tail bound of
\cref{lem:gaussian-moment-general} gives
\[
\gamma_R(\mathbb Z^n\setminus B_{T_R})
\le2^{n/2}e^{-\pi T_R^2/(2R^2)}
\le2^{n/2}e^{-5000\pi(n+\ln^2R)}
\le R^{-100}.
\]
Hence $\gamma_R^{\rm tr}\le(1-R^{-100})^{-1}\gamma_R\le2\gamma_R$, so each
$\mu_i$ is an $(\alpha_i/2)$-dense piece of $\gamma_R$ supported in
$B_{T_R}$. These densities have product at least $(\alpha/2)^M$, and
$\log(2/(\alpha/2))=S+1$.

By part~(1) of \cref{cor:coarse-dissociated-convolution-general} with
density $\alpha/2$ and threshold $K\coloneqq R^{2-2\theta}\ge8R^\theta$,
there is a finite subgroup $H$ of order $2^{O_\theta(S)}$ such that
\[
|\widehat\nu(\zeta)|\ge e^{-M/K}=e^{-R^{2\theta}}
\quad\Longrightarrow\quad
\operatorname{dist}_{\mathbb T^n}(\zeta,H)
\le\Delta\coloneqq R^{-1+\theta}.
\]
Every finite group $G$ has a generating set of size at most $\log_2|G|$ (by greedy selection), so
$H$ is generated by $m=O_\theta(S)$ frequencies $t_1,\dots,t_m$.
Put $\eta\coloneqq e^{-R^{2\theta}/10}$. Since $\eta\ge e^{-R^{2\theta}}$,
every frequency $\zeta$ with $|\widehat\nu(\zeta)|\ge\eta$ also lies within
$\Delta$ of $H$.
A vector $v$ as in the statement annihilates the generators of $H$,
hence all of $H$.

It remains to bound the second term of \eqref{eq:fourier-cutoff-tv-general}.
The convolution $\nu$ is supported in the ball $B_{MT_R}$, which contains at
most $N\coloneqq(2MT_R+1)^n$ lattice points. Since $n\le R^{\theta/2}$, we have
$\log N=O(n\log R)=o(R^{2\theta})$, so
\begin{equation}
\label{eq:translation-cutoff-errors-general}
\sqrt{2N}\,\eta\le\tfrac12R^{-1}
\end{equation}
for sufficiently large $R$. Hence \cref{lem:translation-cutoff-general} with
$u\coloneqq\|v\|_2\Delta$, $E\coloneqq B_{MT_R}$, and $\varepsilon=0$,
together with \eqref{eq:translation-cutoff-errors-general}, gives
\[
\dTV(\nu,\tau_v\nu)
\le20u+\sqrt{2N}\,\eta\le20\|v\|_2R^{-1+\theta}+R^{-1}.
\]
\end{proof}

The mollified reduction convolves $\nu$ with one more discrete Gaussian
$\gamma_R$. The extra factor $\widehat{\gamma_R}$ suppresses the Fourier
coefficients far from the origin, so every large Fourier coefficient of
$\nu*\gamma_R$ lies at a frequency near the origin, where part~(2) of
\cref{cor:coarse-dissociated-convolution-general} applies. This strengthens
\cref{cor:translation-invariance-from-convolution-spectrum-general} in two
ways: the conditions on $v$ become $O(S/\log R)$ orthogonality relations
with points of $R^{-1}\mathbb Z^n$, and the coefficient $20R^{-1+\theta}$ of
$\|v\|_2$ improves to $400R^{-2+5\theta/2}$, which allows much longer shifts
$v$.

\begin{corollary}
\label{cor:translation-invariance-smoothed-near-origin-general}
Fix $0<\theta\le1/4$. Let $R\ge n^{2/\theta}$ be an integer that is
sufficiently large depending only on $\theta$, put $M\coloneqq R^2$ and
$T_R\coloneqq100R(\sqrt n+\ln R)$, and let
$\gamma_R^{\rm tr}$ be $\gamma_R$ conditioned on the ball $B_{T_R}$. Let
$\alpha,\alpha_1,\dots,\alpha_M\in(0,1]$ satisfy
$\prod_{i=1}^M \alpha_i\ge \alpha^M$, and assume
$S\coloneqq\log(2/\alpha)\le n\log R$. Let each $\mu_i$ be an $\alpha_i$-dense piece of $\gamma_R^{\rm tr}$, and write
$
\nu\coloneqq\mu_1*\cdots *\mu_M
$
and
$
\nu_0\coloneqq\nu*\gamma_R.
$
Then there exist frequencies
\[
\zeta_1,\dots,\zeta_\ell\in R^{-1}\mathbb Z^n\cap[-1/2,1/2)^n,
\qquad
\ell=O_\theta(S/\log R),
\]
such that every $v\in\mathbb Z^n$ satisfying
$\langle v,\zeta_j\rangle=0$ for every $j\in[\ell]$ obeys
\[
\dTV(\nu_0,\tau_v\nu_0)
\le400\|v\|_2R^{-2+5\theta/2}+R^{-1}.
\]
\end{corollary}

\begin{proof}
As in the proof of
\cref{cor:translation-invariance-from-convolution-spectrum-general},
$\gamma_R(B_{T_R})\ge1-R^{-100}$ and each $\mu_i$ is an $(\alpha_i/2)$-dense
piece of $\gamma_R$ supported in $B_{T_R}$.
Apply part~(2) of \cref{cor:coarse-dissociated-convolution-general} with
density $\alpha/2$ and threshold $K\coloneqq R^{2-2\theta}$.  It gives frequencies
$\zeta_1,\ldots,\zeta_\ell\in R^{-1}\mathbb Z^n\cap[-1/2,1/2)^n$,
with $\ell=O_\theta(S/\log R)$, whose real span $W$ satisfies
\[
\operatorname{dist}(\widetilde\xi,W)\le\Delta,
\qquad \Delta\coloneqq\frac{20R^{3\theta/2}}{R\sqrt K}=20R^{-2+5\theta/2},
\]
whenever $|\widehat\nu(\xi)|\ge e^{-M/K}=e^{-R^{2\theta}}$ and
$\|\xi\|_{\mathbb T^n}\le R^{-1+\theta}$.

We check that this bound covers every large Fourier coefficient of the
smoothed convolution $\nu_0$. Put $\eta\coloneqq e^{-R^{2\theta}/10}$. If $|\widehat{\nu_0}(\xi)|\ge\eta$,
then $|\widehat\nu(\xi)|\ge\eta\ge e^{-R^{2\theta}}$, because
$|\widehat{\nu_0}|=|\widehat\nu|\,|\widehat{\gamma_R}|\le|\widehat\nu|$.
Moreover, \cref{lem:dg-fourier-decay-general} gives
\[
|\widehat{\nu_0}(\xi)|
\le|\widehat{\gamma_R}(\xi)|
\le e^{-R^{2\theta}/5}<\eta
\qquad\text{if }\|\xi\|_{\mathbb T^n}>R^{-1+\theta}.
\]
Thus every frequency $\xi$ with
$|\widehat{\nu_0}(\xi)|\ge\eta$ lies within $\Delta$ of the image of $W$
in $\mathbb T^n$.

A vector $v$ in the statement is orthogonal to $W$ over $\mathbb R$,
and hence annihilates its image in $\mathbb T^n$.
Since the convolution $\nu$ is supported in $B_{MT_R}$, we have
$\nu_0(B_{(M+1)T_R})\ge\gamma_R(B_{T_R})\ge1-R^{-100}$. The ball
$B_{(M+1)T_R}$ contains at most $N\coloneqq(2(M+1)T_R+1)^n$ lattice points,
and the argument for \eqref{eq:translation-cutoff-errors-general} again gives
$\sqrt{2N}\,\eta\le\frac12R^{-1}$.
Apply \cref{lem:translation-cutoff-general} with $u\coloneqq\|v\|_2\Delta$,
$E\coloneqq B_{(M+1)T_R}$, and $\varepsilon=R^{-100}$ to obtain
\[
\dTV(\nu_0,\tau_v\nu_0)
\le20u+\sqrt{2N}\,\eta+R^{-100}\le400\|v\|_2R^{-2+5\theta/2}+R^{-1}.
\]
\end{proof}

%!TEX root = main.tex

\section{From streaming algorithms to linear sketches}
\label{sec:randomized-transfer}

This section proves the reduction theorems. We first fix a target
distribution $\mathcal I$ and turn a randomized $S$-bit streaming algorithm
into a linear sketch that solves the same problem for final frequency vectors drawn from
$\mathcal I$. \Cref{subsec:exact-transfer} does this for the exact reduction
and then applies Yao's minimax principle~\cite{yao1977probabilistic} to
obtain a sketch that is correct on \emph{every} input of bounded norm.
\Cref{subsec:mollified-transfer} proves the mollified reduction for problems
that are smooth on average on $\mathcal I$.
From \cref{sec:coarse-large-spectrum-general,sec:translation-invariance-general},
we use only the translation-invariance statements
\cref{cor:translation-invariance-from-convolution-spectrum-general,cor:translation-invariance-smoothed-near-origin-general}.

\subsection{Exact reduction}
\label{subsec:exact-transfer}

The streams we construct consist of blocks, each of which adds a prescribed
increment $v\in\mathbb Z^n$ to the frequency vector. To fix the order of the
updates inside a block, we use the following canonical realization.

\begin{definition}[Canonical blocks]
\label{def:canonical-realization}
For $v\in\mathbb Z^n$, let $\mathsf{can}(v)$ update the coordinates in
order, using $|v_i|$ updates of sign $\operatorname{sgn}(v_i)$ on coordinate
$i$.
\end{definition}

The canonical block $\mathsf{can}(v)$ has $\|v\|_1$ unit-updates and changes
the frequency vector by $v$.

\begin{definition}[Exact stream distribution]\label{def:exact-stream-distribution}
Fix an integer radius $R\ge n^{8}$ and a target radius $D\ge1$. Put
$M\coloneqq R^2$ and let $\gamma_R^{\rm tr}$ be the discrete Gaussian $\gamma_R$
conditioned on the
ball $B_{T_R}$, where $T_R\coloneqq100R(\sqrt n+\ln R)$. Draw independent prefix increments
$X_1,\ldots,X_M\sim\gamma_R^{\rm tr}$ and a target $Y\sim\mathcal I$, where the
target distribution $\mathcal I$ is supported in $B_D$. The exact
stream consists of
\[
\mathsf{can}(X_1),\ldots,\mathsf{can}(X_M),
\mathsf{can}\!\left(Y-\sum_iX_i\right),
\]
followed by a query.
\end{definition}
\begin{remark}\label{rmk:exact-stream-distribution}
The final frequency vector is $Y$, and the stream has
$O(R^3\sqrt n(\sqrt n+\ln R)+\|Y\|_1)$ unit-updates. The same computation as
in the proof of
\cref{cor:translation-invariance-from-convolution-spectrum-general} shows
that the truncation discards mass
$\gamma_R(\mathbb Z^n\setminus B_{T_R})\le R^{-100}$, that is,
$\dTV(\gamma_R^{\rm tr},\gamma_R)\le R^{-100}$.
\end{remark}

The next lemma conditions on the states of the algorithm after the prefix
blocks. Some state path $\sigma$ is likely enough and keeps the algorithm
correct for most targets, and conditioned on $\sigma$ the prefix increments
are independent dense pieces of $\gamma_R^{\rm tr}$. Their sum is then a
convolution as in \cref{sec:coarse-large-spectrum-general}, which later
yields the linear sketch. We state the lemma for a general law $\gamma$ and a
general success event, so that the mollified reduction in
\cref{subsec:mollified-transfer} can use it as well. 

\begin{lemma}[Conditioning on a heavy state path]
\label{lem:heavy-state-path-transfer}
Let $A$ be a deterministic streaming algorithm with at most $2^S$ states.
Let $X_1,\ldots,X_M$ be independent samples from a probability law $\gamma$
on $\mathbb Z^n$, and let $Y\sim\mathcal I$ be independent of
$(X_1,\ldots,X_M)$. Run $A$ on
$\mathsf{can}(X_1),\ldots,\mathsf{can}(X_M)$, and write
$\Sigma=(\Sigma_0,\ldots,\Sigma_M)$ for its states, where $\Sigma_0$ is the
initial state and $\Sigma_i$ is the state after $\mathsf{can}(X_i)$. Let
$\mathsf{Succ}$ be an event with $\Pr[\mathsf{Succ}]\ge1-\delta$. Then there
is a state path $\sigma=(\sigma_0,\ldots,\sigma_M)$ with the following
properties, where conditioning on $\sigma$ means conditioning on the event
$\{\Sigma=\sigma\}$.
\begin{enumerate}
\item For $i\in[M]$, let $\beta_i^\sigma$ be the probability that $A$,
started in state $\sigma_{i-1}$, is in state $\sigma_i$ after processing
$\mathsf{can}(X)$ for $X\sim\gamma$, and let $X_i^\sigma$ be distributed as
$X$ conditioned on this event, with $X_1^\sigma,\ldots,X_M^\sigma$
independent. Conditioned on $\sigma$, the tuple $(X_1,\ldots,X_M)$ has the
law of $(X_1^\sigma,\ldots,X_M^\sigma)$ and is independent of
$Y\sim\mathcal I$. The law of each $X_i^\sigma$ is a $\beta_i^\sigma$-dense
piece of $\gamma$, and
\[
\prod_{i=1}^M\beta_i^\sigma=\Pr[\Sigma=\sigma],\qquad
\beta_\sigma\coloneqq\Pr[\Sigma=\sigma]^{1/M}\ge2^{-(S+1)}.
\]
\item The algorithm still succeeds on most $Y \sim  \mathcal{I}$:
\[
\Pr[\mathsf{Succ}\mid\sigma]\ge1-2\delta,
\qquad
\Pr_{y\sim\mathcal I}\Big[\Pr[\mathsf{Succ}\mid\sigma,\,Y=y]\ge\frac34\Big]
\ge1-8\delta.
\]
\end{enumerate}
\end{lemma}

\begin{proof}
We first condition on an arbitrary state path $\sigma$ with
$\Pr[\Sigma=\sigma]>0$. The event $\{\Sigma=\sigma\}$ is the intersection
over $i\in[M]$ of the events that $A$ moves from $\sigma_{i-1}$ to
$\sigma_i$ on $\mathsf{can}(X_i)$. The $i$-th event depends only on $X_i$
and has probability $\beta_i^\sigma$. Since $X_1,\ldots,X_M,Y$ are
independent, conditioning on this intersection keeps them independent,
leaves the law of $Y$ unchanged, and gives $X_i$ the law of $X_i^\sigma$. It
also gives $\Pr[\Sigma=\sigma]=\prod_i\beta_i^\sigma$. Conditioning $\gamma$
on an event of probability $\beta_i^\sigma$ gives a $\beta_i^\sigma$-dense
piece of $\gamma$, as noted after \cref{def:dense-piece}. Thus part~(1) holds for every such $\sigma$. The lower bound on $\beta_\sigma$ will be proved in the next paragraph.

We now choose $\sigma$. Put $q_\sigma\coloneqq\Pr[\mathsf{Succ}\mid\sigma]$.
Call a state path $\sigma$ \emph{good} if $q_\sigma\ge1-2\delta$, and let $p$
be the probability that $\Sigma$ is not good. Since every $q_\sigma\le1$,
\[
1-\delta\le\mathbb E\,q_\Sigma\le(1-p)+p(1-2\delta)=1-2p\delta,
\]
so $p\le1/2$ (if $\delta=0$, then $q_\Sigma=1$ almost surely and $p=0$), and
the good paths have total probability at least $1/2$. The initial state is
fixed, so there are at most $2^{SM}$ paths, and some good path $\sigma$ has
$\Pr[\Sigma=\sigma]\ge2^{-1}2^{-SM}\ge2^{-(S+1)M}$, that is,
$\beta_\sigma\ge2^{-(S+1)}$. This completes part~(1).

Since $\sigma$ is good, $\Pr[\mathsf{Succ}\mid\sigma]=q_\sigma\ge1-2\delta$.
By the first paragraph, $Y\sim\mathcal I$ conditioned on $\sigma$, so
\[
\mathbb E_{y\sim\mathcal I}\Pr[\mathsf{Succ}^c\mid\sigma,\,Y=y]
=1-q_\sigma\le2\delta,
\]
and Markov's inequality gives
\[
\Pr_{y\sim\mathcal I}\Big[\Pr[\mathsf{Succ}^c\mid\sigma,\,Y=y]>\frac14\Big]
\le8\delta.
\]
This proves part~(2).
\end{proof}

For a homomorphism $\mathcal A : \mathbb Z^n \to \Gamma$, the \emph{fiber} of
$\mathcal A$ containing a point $y$ is the set of points $y'$ with
$\mathcal A(y')=\mathcal A(y)$. The next lemma builds a decoder from a family
of laws $P_y$ on $\mathbb Z^n$ (one for each target $y$ in the support of a
target distribution $\mathcal I$), that are close in total variation within
every fiber. In our application, $P_y$ is the law of the last increment
$a\in\mathbb Z^n$ of the stream for the target $y$. Once the seed and the
state path are fixed, the output of the streaming algorithm is determined by
the last block $\mathsf{can}(a)$, and $\Phi(a)$ denotes this output.

\begin{lemma}[Decoding from the sketch value]
\label{lem:stable-final-block-decoding}
Let $\mathcal I$ be a distribution on $\mathbb Z^n$, and let
$\mathcal A:\mathbb Z^n\to\Gamma$ be a homomorphism. For every
$y\in\operatorname{supp}(\mathcal I)$, let $P_y$ be a probability law on
$\mathbb Z^n$, and let $\Phi:\mathbb Z^n\to\mathcal O$ be a function.
Suppose that
\begin{equation}
\label{eq:stable-final-block-tv}
\dTV(P_y,P_{y'})\le\tau
\quad\text{whenever}\quad
\mathcal{A}(y) = \mathcal{A}(y').
\end{equation}
Then the following hold.
\begin{enumerate}
\item Let $(\mathcal M,d)$ be a metric space with $\mathcal O\subseteq\mathcal M$,
let $f:\mathbb Z^n\to\mathcal M\cup\{*\}$, and let $\epsilon\ge0$. There is
a deterministic decoder $g$ such that
$d(g(\mathcal A(y)),f(y))\le3\epsilon$ for every
$y\in\operatorname{supp}(\mathcal I)$ with
\[
\Pr_{a\sim P_y}\big[d(\Phi(a),f(y))\le\epsilon\big]>\frac{1+\tau}2.
\]
\item For a relation $\mathcal R\subseteq\mathbb Z^n\times\mathcal O$, there
is a deterministic decoder $g$ such that
\[
\Pr_{Y\sim\mathcal I}\big[(Y,g(\mathcal A(Y)))\notin\mathcal R\big]
\le\Pr_{Y\sim\mathcal I,\,a\sim P_Y}\big[(Y,\Phi(a))\notin\mathcal R\big]+\tau.
\]
\end{enumerate}
\end{lemma}

\begin{proof}
For the first claim, call $y\in\operatorname{supp}(\mathcal I)$ \emph{good}
if it satisfies the displayed inequality $\Pr_{a\sim P_y}\big[d(\Phi(a),f(y))\le\epsilon\big]>\frac{1+\tau}2$, and put
$C_y\coloneqq\{a\in\mathbb Z^n:d(\Phi(a),f(y))\le\epsilon\}$, so that
$P_y(C_y)>(1+\tau)/2$ for good $y$. Let $y,y'$ be good points in the same
fiber.
Then \eqref{eq:stable-final-block-tv} gives
$P_{y'}(C_y)\ge P_y(C_y)-\tau$, and a union bound gives
\[
P_{y'}(C_y\cap C_{y'})
\ge P_{y'}(C_y)+P_{y'}(C_{y'})-1
\ge P_y(C_y)+P_{y'}(C_{y'})-1-\tau>0.
\]
Thus some $a$ lies in both $C_y$ and $C_{y'}$. If $f(y)$ and $f(y')$ both
lie in $\mathcal M$, the triangle inequality through $\Phi(a)$ gives
$d(f(y),f(y'))\le2\epsilon$.

Now fix a fiber $F$ that contains a good point. Choose a good $y_F\in F$,
taking $f(y_F)\in\mathcal M$ whenever possible, and put
$g(\mathcal A(y_F))\coloneqq\Phi(a)$ for some $a\in C_{y_F}$, so that
$g(\mathcal A(y_F))$ is within $\epsilon$ of $f(y_F)$. Define $g$
arbitrarily on the remaining values. Let $y\in F$ be good. If $f(y)=*$,
there is nothing to prove, by the convention following
\cref{def:metric-approximation}. Otherwise $f(y_F)\in\mathcal M$ by the choice of
$y_F$, and
$d(g(\mathcal A(y)),f(y))\le d(g(\mathcal A(y_F)),f(y_F))+d(f(y_F),f(y))
\le3\epsilon$.

For the second claim, on every fiber $F$ that meets
$\operatorname{supp}(\mathcal I)$, choose $y_F\in F\cap\operatorname{supp}(\mathcal I)$.
For $y\in F\cap\operatorname{supp}(\mathcal I)$,
\eqref{eq:stable-final-block-tv} gives
\[
\Pr_{a\sim P_{y_F}}\big[(y,\Phi(a))\notin\mathcal R\big]
\le\Pr_{a\sim P_y}\big[(y,\Phi(a))\notin\mathcal R\big]+\tau.
\]
Multiplying by $\Pr[Y=y]$ and summing over $y\in F$ shows that
$\mathbb E_{a\sim P_{y_F}}\Pr_{Y\sim\mathcal I}[Y\in F,\,(Y,\Phi(a))\notin\mathcal R]$
is at most the contribution of $F$ to the right-hand side of the claim.
Fix $a_F$ for which the inner probability is at most this mean, and put
$g(\mathcal A(y_F))\coloneqq\Phi(a_F)$. Summing over the fibers proves the
claim. The remaining values of $\mathcal A$ may be decoded arbitrarily.
\end{proof}

We can now state and prove the exact reduction theorem for a fixed target
distribution.

\begin{theorem}
\label{thm:randomized-exact-transfer}
Let $R\ge n^{8}$ be a sufficiently large integer, and let
$1\le D\le R^{1/3}$.
Suppose $\mathcal I$ is supported in $B_D$ and a randomized
$S$-bit algorithm succeeds with probability at least $1-\delta$ on the
exact stream distribution (\cref{def:exact-stream-distribution}).
There is a homomorphism $\mathcal A:\mathbb Z^n\to\mathbb T^m$, with
$m=O(S)$ and $|\mathcal A(\mathbb Z^n)|\le2^{O(S)}$ (hidden constants
are absolute and do not depend on $\mathcal I$), and the following decoders:
\begin{enumerate}
\item For a metric approximation problem $(f,\mathcal O)$ with accuracy
$\epsilon$, a deterministic decoder $g$ satisfies
\[
\Pr_{Y\sim\mathcal I}[d(g(\mathcal A(Y)),f(Y))\le3\epsilon]\ge1-8\delta.
\]
\item For a relation $\mathcal R\subseteq\mathbb Z^n\times\mathcal O$,
if $R\ge\delta^{-3}$, a deterministic decoder $g$ satisfies
\[
\Pr_{Y\sim\mathcal I}[(Y,g(\mathcal A(Y)))\in\mathcal R]\ge1-3\delta.
\]
\end{enumerate}
\end{theorem}

\begin{proof}
First suppose that $S>n\log R/2$, and put $q\coloneqq\lfloor2D\rfloor+1$.
Two distinct points of $\operatorname{supp}(\mathcal I)$ differ in some
coordinate by a nonzero integer of absolute value at most $2D<q$, so reduction
modulo $q$ is injective on $\operatorname{supp}(\mathcal I)$. Its dimension
is $n<2S/\log R$ and its image size is $q^n=2^{O(n\log R)}=2^{O(S)}$. Since it
determines the target, the decoder can return a valid output whenever the
algorithm produces one with positive probability. Hence the result holds trivially in this case. Now assume that
$S\le n\log R/2$.

Using the easy direction of Yao's minimax principle~\cite{yao1977probabilistic}, there exists a seed of the streaming algorithm, so that the streaming algorithm run with such a seed fails with probability at most $\delta$ on the exact
stream distribution. We fix such a seed, and hence from now on the streaming algorithm is deterministic.
Let $\mathsf{Succ}$ be the event that the output of the algorithm on the
exact stream distribution is valid for its final frequency vector $Y$, that is,
$d(o,f(Y))\le\epsilon$ for a metric approximation problem and
$(Y,o)\in\mathcal R$ for a relation, where $o$ is the output. By the choice
of the seed, $\Pr[\mathsf{Succ}]\ge1-\delta$. We apply
\cref{lem:heavy-state-path-transfer} with $\gamma=\gamma_R^{\rm tr}$ and this
event, and take a state path $\sigma$ and the variables
$X_1^\sigma,\ldots,X_M^\sigma$ from it.
Write $X_\sigma\coloneqq X_1^\sigma+\cdots+X_M^\sigma$ for the conditioned
prefix sum and $\nu_\sigma$ for its law. It is the convolution of the laws of
the $X_i^\sigma$.
By part~(1) of \cref{lem:heavy-state-path-transfer}, the law of each
$X_i^\sigma$ is a $\beta_i^\sigma$-dense piece of $\gamma_R^{\rm tr}$,
$\prod_i\beta_i^\sigma=\beta_\sigma^M$, and $\beta_\sigma\ge2^{-(S+1)}$.
Thus the geometric mean density $\alpha\coloneqq\beta_\sigma$ satisfies
\begin{equation}
\label{eq:conditioned-density-transfer}
\log(2/\alpha)=\log(2/\beta_\sigma)\le S+2\le n\log R.
\end{equation}
By \cref{cor:translation-invariance-from-convolution-spectrum-general} with
$\theta=1/4$, there is a group $H\le\mathbb T^n$ of order $2^{O(S)}$ with generators
$t_1,\ldots,t_m$, $m=O(S)$. Set
\[
\mathcal A_\sigma(y)\coloneqq(\langle t_1,y\rangle,\ldots,\langle t_m,y\rangle)
\in\mathbb T^m.
\]
Since $t_1,\ldots,t_m$ generate $H$, the sketch values $\mathcal A_\sigma(y)$ are in one-to-one correspondence with the characters $h\mapsto e(\langle h,y\rangle)$ of $H$. Hence the image $|\mathcal A_\sigma(\mathbb Z^n)|\le|\widehat H|=|H|\le2^{O(S)}$.

Let $P_y\coloneqq\operatorname{Law}(y-X_\sigma)$, and let $\Phi_\sigma(a)$ be
the algorithm's output after processing $\mathsf{can}(a)$ from state
$\sigma_M$. Conditioned on $\sigma$ and $Y=y$, the last increment $y-\sum_iX_i$
has law $P_y$, so $\Phi_\sigma$ applied to a sample of $P_y$ has the
conditional law of the algorithm's output. If
$\mathcal A_\sigma(y)=\mathcal A_\sigma(y')$, then $v\coloneqq y-y'$ satisfies
$\langle v,t_j\rangle\in\mathbb Z$ for every $j$, and $\|v\|_2\le2D\le2R^{1/3}$
because $\mathcal I$ is supported in $B_D$. Since $P_y$ is the
translate of $P_{y'}$ by $v$,
\cref{cor:translation-invariance-from-convolution-spectrum-general} gives
$\dTV(P_y,P_{y'})=\dTV(\nu_\sigma,\tau_v\nu_\sigma)\le40R^{1/3}\cdot R^{-3/4}+R^{-1}\le\tau\coloneqq R^{-1/3}$
for sufficiently large $R$. If
$\Pr[\mathsf{Succ}\mid\sigma,\,Y=y]\ge3/4$, then $a\sim P_y$ satisfies
$d(\Phi_\sigma(a),f(y))\le\epsilon$ with probability at least $3/4$, which
exceeds $(1+\tau)/2$ because $\tau<1/2$.
Part \textup{(1)} of \cref{lem:stable-final-block-decoding} therefore gives a
decoder that is $3\epsilon$-accurate at every such target $y$, and by
part~(2) of \cref{lem:heavy-state-path-transfer} these targets have
probability at least $1-8\delta$ under $\mathcal I$. This gives the first
conclusion. For a relation,
$\Pr_{Y\sim\mathcal I,\,a\sim P_Y}[(Y,\Phi_\sigma(a))\notin\mathcal R]
=\Pr[\mathsf{Succ}^c\mid\sigma]$, so part \textup{(2)} of
\cref{lem:stable-final-block-decoding} bounds the failure probability by
\[
\Pr[\mathsf{Succ}^c\mid\sigma]+\tau
\le2\delta+R^{-1/3}\le3\delta,
\]
using part~(2) of \cref{lem:heavy-state-path-transfer} and
$R\ge\delta^{-3}$.
\end{proof}

Finally, we turn the guarantee for a fixed target distribution into a
guarantee for every input. This uses the hard direction of Yao's minimax
principle~\cite{yao1977probabilistic}:  if the set of inputs is finite and,
for every distribution on inputs, some deterministic sketch succeeds with probability at least $1-\delta'$ under that distribution, then some distribution over these deterministic sketches succeeds with
probability at least $1-\delta'$ on every input.

\begin{corollary}
\label{cor:yao-minmax-exact-transfer}
Fix $D\ge1$ and a set of inputs $\mathcal X\subseteq B_D\cap\mathbb Z^n$, where $B_D\subseteq\mathbb R^n$ is the Euclidean ball of radius $D$. Let $R\ge n^{8}$ be a sufficiently large integer, and assume $D\le R^{1/3}$. Suppose there is a randomized turnstile streaming algorithm using $S$ bits of space which solves a streaming problem with success probability at least $1-\delta$ on every turnstile stream of length at most $R^{4}$ whose final frequency vector lies in $\mathcal X$. Then there exists a public-coin randomized linear sketch of dimension $O(S)$ with image size $2^{O(S)}$,
with the following guarantees on every $x\in\mathcal X$:
\begin{enumerate}
\item For a metric approximation problem $(f,\mathcal O)$ with accuracy
$\epsilon$, the decoder outputs $z\in\mathcal O$ with $d(z,f(x))\le3\epsilon$
with probability at least $1-8\delta$.
\item For a relation $\mathcal R\subseteq\mathcal X\times\mathcal O$,
if $R\ge\delta^{-3}$, it outputs $o$ with $(x,o)\in\mathcal R$
with probability at least $1-3\delta$.
\end{enumerate}
\end{corollary}

\begin{proof}
Let $\mathcal I$ be any target distribution on $\mathcal X$. By
\cref{rmk:exact-stream-distribution}, its exact streams have at most
$O(R^3\sqrt n(\sqrt n+\ln R)+\sqrt nD)\le R^4$ unit-updates, and their final
frequency vectors lie in $\mathcal X$. Hence the algorithm solves the problem with
probability at least $1-\delta$ on the exact stream distribution.
Since $1\le D\le R^{1/3}$, \cref{thm:randomized-exact-transfer} applies to
every such $\mathcal I$
and gives a deterministic sketch of dimension $O(S)$ and image size
$2^{O(S)}$, with absolute constants, together with a deterministic decoder.

Since the input set $\mathcal X$ is finite, Yao's minimax
principle~\cite{yao1977probabilistic}, in the form stated before the
corollary, gives a public-coin randomized linear sketch with the same bounds
that succeeds on every input with probability at least $1-8\delta$ in
case~\textup{(1)} and at least $1-3\delta$ in case~\textup{(2)}.
\end{proof}

\subsection{Mollified reduction}\label{subsec:mollified-transfer}

The exact reduction theorem (\cref{thm:randomized-exact-transfer}) gives a sketch with $O(S)$ rows whose entries are
frequencies in $\mathbb T^n$.  For problems that are insensitive to a small
Gaussian perturbation of the input, we strengthen it.  Compared to the exact
stream, the mollified stream adds an independent Gaussian vector $Z$ to the
final frequency vector. This multiplies the Fourier transform of the
conditioned prefix law by $\widehat{\gamma_R}$, which removes its large
Fourier coefficients far from the origin. Then
\cref{cor:translation-invariance-smoothed-near-origin-general} applies and
gives an integer sketch with $O(S/\log R)$ rows and entries bounded by $R$.
The insensitivity we need is the following average smoothness condition.

\begin{definition}[Average smoothness]
\label{def:smoothness}
Let $\mathcal Z$ be a noise distribution on $\mathbb Z^n$ (in our
application, $\mathcal Z=\gamma_R$). For independent $Y\sim\mathcal I$ and
$Z\sim\mathcal Z$, define average smoothness on $\mathcal I$ with respect to
$\mathcal Z$ as follows.
\begin{enumerate}
\item A metric target $f:\mathbb Z^n\to\mathcal M\cup\{*\}$ is
\emph{$(\epsilon,\delta)$-smooth on average} if
\[
\Pr_{Y\sim\mathcal I, Z\sim\mathcal Z}\bigl[f(Y)\in\mathcal M\text{ and }\bigl(f(Y+Z)=*
      \text{ or }d(f(Y+Z),f(Y))>\epsilon\bigr)\bigr]\le\delta.
\]
\item A relation $\mathcal R\subseteq\mathbb Z^n\times\mathcal O$ is
\emph{$\delta$-smooth into $\mathcal R'$ on average} if
\[
\Pr_{Y\sim\mathcal I, Z\sim\mathcal Z}[\exists o\in\mathcal O:\ (Y+Z,o)\in\mathcal R
                         \text{ and }(Y,o)\notin\mathcal R']\le\delta.
\]
\end{enumerate}
\end{definition}
Intuitively, the smoothness means adding the noise $Z$ to a typical target $Y$ rarely changes the answer by
much. In the metric case it keeps the promise and moves the target value by
at most $\epsilon$, and in the relation case every output that is valid for $Y+Z$ under $\mathcal R$ is
valid for $Y$ under the relaxed relation $\mathcal R'$.

\begin{definition}[Mollified stream distribution]\label{def:mollified-stream-distribution}
Use the Gaussian prefix and truncation of the exact stream distribution
(\cref{def:exact-stream-distribution}).
Let the target distribution $\mathcal I$ be supported in $B_D$, draw
independent $Y\sim\mathcal I$ and $Z\sim\gamma_R^{\mathrm{tr}}$, and set the
last increment to
\[
a\coloneqq Y+Z-(X_1+\cdots+X_M).
\]
The resulting distribution is the \emph{mollified stream distribution of
radius $R$ and target radius $D$ associated with $\mathcal I$}.
\end{definition}
\begin{remark}\label{rmk:mollified-stream-distribution}
The final frequency vector is $Y+Z$, and the stream still has
$O(R^3\sqrt n(\sqrt n+\ln R)+\|Y\|_1)$ unit-updates, as in
\cref{rmk:exact-stream-distribution}.
\end{remark}

\begin{theorem}
\label{thm:randomized-mollified-transfer}
Let $R\ge n^{8}$ be a sufficiently large integer, and let $\mathcal I$ be
supported in $B_D$, where $D\le R^{5/4}$. Suppose a randomized $S$-bit
algorithm succeeds with probability at least $1-\delta$ on the mollified
stream distribution (\cref{def:mollified-stream-distribution}). There is an integer matrix $\mathcal A$ with
$O(S/\log R)$ rows and entries bounded by $R$, and the following decoders:
\begin{enumerate}
\item For a metric approximation problem $(f,\mathcal O)$ with accuracy
$\epsilon$, if $f$ is $(\epsilon,\delta)$-smooth on average on $\mathcal I$
with respect to $\gamma_R$, a deterministic decoder $g$ satisfies
\[
\Pr_{Y\sim\mathcal I}[d(g(\mathcal A(Y)),f(Y))\le6\epsilon]\ge1-16\delta.
\]
\item If the algorithm solves a relation $\mathcal R$ that is
$\delta$-smooth into $\mathcal R'$ on average on $\mathcal I$ with respect
to $\gamma_R$, and $R\ge\delta^{-10}$, a deterministic decoder $g$ satisfies
\[
\Pr_{Y\sim\mathcal I}[(Y,g(\mathcal A(Y)))\in\mathcal R']\ge1-6\delta.
\]
\end{enumerate}
\end{theorem}

\begin{proof}
If $S>n\log R/2$, use the identity matrix, as in the proof of
\cref{thm:randomized-exact-transfer}. It has $n<2S/\log R$ rows and entries
bounded by $1$, and it determines the target $Y$. The smoothness assumption
is stated for $Z\sim\gamma_R$, while the stream uses $Z\sim\gamma_R^{\rm tr}$.
Since $\gamma_R(B_{T_R})\ge1-R^{-100}\ge1/2$ by
\cref{rmk:exact-stream-distribution}, the truncation at most doubles the
probability of the bad event in \cref{def:smoothness}. By the success of the
algorithm and the smoothness assumption, a set of targets of mass at least
$1-3\delta$ has an output that is valid for the target in the sense of the
theorem, and the decoder returns such an output.

Otherwise, fix a seed and a state path $\sigma$ as in the proof of
\cref{thm:randomized-exact-transfer}, where $\mathsf{Succ}$ is now the event
that the output is valid for the final frequency vector $Y+Z$, and use the
same notation $\nu_\sigma,X_\sigma,\Phi_\sigma$.
By \eqref{eq:conditioned-density-transfer} and
\cref{cor:translation-invariance-smoothed-near-origin-general} with
$\theta=1/4$, there are $\ell=O(S/\log R)$ frequencies
$\zeta_j\in R^{-1}\mathbb Z^n\cap[-1/2,1/2)^n$. Let $\mathcal A_\sigma$
be the integer matrix with rows $R\zeta_j$; its entries are bounded by $R/2$.

Conditioned on $\sigma$ and $Y=y$, the last increment has law
$P_y\coloneqq\operatorname{Law}(y-X_\sigma+Z)$ with $Z\sim\gamma_R^{\rm tr}$.
Replacing $Z$ by an untruncated $Z'\sim\gamma_R$ changes $P_y$ by at most
$\dTV(\gamma_R^{\rm tr},\gamma_R)\le R^{-100}$
(\cref{rmk:exact-stream-distribution}), and since $\gamma_R$ is symmetric,
the resulting law is $\operatorname{Law}(y-W)$ with
$W\sim\nu_\sigma*\gamma_R$. If $\mathcal A_\sigma y=\mathcal A_\sigma y'$,
then $v\coloneqq y-y'$ is orthogonal to every $\zeta_j$, and
$\|v\|_2\le2D\le2R^{5/4}$ because $\mathcal I$ is supported in $B_D$. As in
the proof of \cref{thm:randomized-exact-transfer}, $\operatorname{Law}(y-W)$
is the translate of $\operatorname{Law}(y'-W)$ by $v$, so
\cref{cor:translation-invariance-smoothed-near-origin-general} gives
\[
\dTV(P_y,P_{y'})
\le\dTV(\nu_\sigma*\gamma_R,\tau_{v}(\nu_\sigma*\gamma_R))+2R^{-100}
\le400\cdot2R^{5/4}\cdot R^{-11/8}+R^{-1}+2R^{-100}
\le\tau\coloneqq R^{-1/10}
\]
for sufficiently large $R$.

For a metric approximation problem, let $\mathrm{Sm}(y)$ be the probability,
over $Z\sim\gamma_R$, that the smoothness event in \cref{def:smoothness}
fails at $y$. Since $\mathbb E\mathrm{Sm}(Y)\le\delta$, Markov's inequality and part~(2)
of \cref{lem:heavy-state-path-transfer} show that, with probability at least
$1-16\delta$ over $y\sim\mathcal I$, both
$\Pr[\mathsf{Succ}\mid\sigma,\,Y=y]\ge3/4$ and $\mathrm{Sm}(y)\le1/8$. Fix
such a target $y$. Conditioned on $\sigma$ and $Y=y$, with probability at
least $3/4-1/8-R^{-100}$ the output $o$ is valid for the final vector $y+Z$
and the smoothness event holds, where $R^{-100}$ accounts for the truncation
of $Z$. Then $d(o,f(y))\le2\epsilon$: this is automatic if $f(y)=*$, and
otherwise $f(y+Z)\in\mathcal M$ and the triangle inequality through
$f(y+Z)$ applies. Since $5/8-R^{-100}>(1+\tau)/2$, the rest is as in the
proof of \cref{thm:randomized-exact-transfer} with $2\epsilon$ in place of
$\epsilon$: part~\textup{(1)} of \cref{lem:stable-final-block-decoding}
gives a decoder that is $6\epsilon$-accurate at every such target $y$. This
gives the first conclusion.

For a relation, let $\mathrm{Sm}_{\mathcal R\to\mathcal R'}(y)$ be the
probability, over $Z\sim\gamma_R$, that some $o$ satisfies
$(y+Z,o)\in\mathcal R$ but $(y,o)\notin\mathcal R'$. If the output is valid
for $y+Z$ under $\mathcal R$ and this bad event does not occur, then the
output is valid for $y$ under $\mathcal R'$. Hence the conditional failure
probability for the target $y$ under $\mathcal R'$ is at most
$\Pr[\mathsf{Succ}^c\mid\sigma,\,Y=y]+\mathrm{Sm}_{\mathcal R\to\mathcal R'}(y)+R^{-100}$.
Part \textup{(2)} of \cref{lem:stable-final-block-decoding} therefore bounds
the total failure probability by
\[
2\delta+\delta+R^{-100}+\tau\le6\delta,
\]
where the first term comes from part~(2) of
\cref{lem:heavy-state-path-transfer} and the second from the smoothness
assumption, using $R\ge\delta^{-10}$.
\end{proof}

\begin{remark}[Strict turnstile streams and ROBPs]\label{rmk:strict-robp-extension}
For strict turnstile streams, with $M$ and $T_R$ as in \cref{def:exact-stream-distribution}, add $b_0\coloneqq\lceil MT_R\rceil$ initial $+1$ updates to each coordinate and subtract $b_0\mathbf1$ from the final increment. All prefix frequencies are then nonnegative, and the final canonical block moves each coordinate monotonically to its target (this requires $\operatorname{supp}(\mathcal I)\subseteq\Z_{\ge0}^n$ in the exact case and $\operatorname{supp}(\mathcal I*\gamma_R^{\mathrm{tr}})\subseteq\Z_{\ge0}^n$ in the mollified case). The length remains $O(nMT_R+\sqrt n\,D)$, at most $R^4$ under the assumptions of \cref{cor:yao-minmax-exact-transfer}. The initial $+1$ updates only change the state from which the Gaussian blocks start, and they shift the law $P_y$ of every final increment by the same vector $-b_0\mathbf1$. Hence the bounds on $\dTV(P_y,P_{y'})$ are unchanged, and the proofs of \cref{thm:randomized-exact-transfer,thm:randomized-mollified-transfer} and \cref{cor:yao-minmax-exact-transfer} go through for nonnegative targets.

The same argument extends to non-uniform read-once branching programs (ROBPs), provided that every block has a length independent of its increment, so that it starts and ends at fixed layers of the ROBP. Write $p(v)\coloneqq\sum_jv_j\bmod2$ for the parity of a vector $v$. A block of unit updates that adds $v$ has length $\equiv p(v)\pmod 2$. We condition each Gaussian prefix increment $X_i\sim\gamma_R^{\rm tr}$ on even parity and pad its block with pairs $+e_1,-e_1$ to a common even length, which preserves strictness and conditional independence. The parity of the length of the last block is then determined by the parity of the final frequency vector.

In the exact reduction, we pad the last block to a common length within each parity class of targets, build one sketch $\mathcal A_0,\mathcal A_1$ for each class, and use the linear sketch $x\mapsto(p(x),\mathcal A_0(x),\mathcal A_1(x))$, whose decoder uses $\mathcal A_{p(x)}$. The two classes use streams of lengths of different parities, and an ROBP of a fixed length $L$ only reaches final vectors $x$ with $p(x)\equiv L\pmod 2$, so for such an ROBP the exact reduction gives a sketch only for those targets.
In the mollified reduction, we instead fix the parity $c$ of the final frequency vector and condition $Z$ on $p(Z)=p(Y)+c$. For every target, this conditions the law of the final increment on the same parity event (no parity row is needed here). In both reductions, each conditioning is on an event of constant probability, so the density, total-variation, and smoothness estimates change only by constant factors, and the asymptotic sketch bounds and the polynomial bound on the stream length are preserved.
\end{remark}

\newcommand{\DG}{\mathcal{D}}

\newcommand{\bA}{\mathbf{A}}
\newcommand{\bG}{\mathbf{G}}
\newcommand{\bH}{\mathbf{H}}
\newcommand{\bJ}{\mathbf{J}}
\newcommand{\bM}{\mathbf{M}}
\newcommand{\bS}{\mathbf{S}}
\newcommand{\bX}{\mathbf{X}}
\newcommand{\bZ}{\mathbf{Z}}
\newcommand{\bx}{\mathbf{x}}
\newcommand{\bSigma}{\boldsymbol{\Sigma}}
\newcommand{\Law}{\mathrm{Law}}

\section{Applications to turnstile streaming lower bounds}
\label{sec:application}
In this section, we present applications of our reduction theorems to
turnstile streaming lower bounds. Each application uses one of the two reduction theorems.

In \cref{subsec:applications-lifting-real-sketches} we combine
\cref{thm:randomized-mollified-transfer} with the lifting theorem for
linear sketches in~\cite{gribelyuk2025lifting}. Starting from a hypothetical
turnstile algorithm, we first obtain a bounded-entry integer sketch for an
explicit discrete hard law, and then use the lifting reduction
of~\cite{gribelyuk2025lifting} to pass to a real sketch for the corresponding
continuous Gaussian hard instance. The contradiction then comes from the
continuous lower bounds in~\cite{PriceW11,LiW16,GangulyW18,NeedellSW22,SwartworthW23,gribelyuk2025lifting}.

In \cref{subsec:applications-smp}, which covers discontinuous problems
such as $L_0$ estimation, matching-size estimation, and subgraph
counting~\cite{AssadiKL17,DuMWY19,KallaugherKP18}, we instead use the
exact reduction in the form of \cref{cor:yao-minmax-exact-transfer}, which
converts the streaming algorithm into a public-coin linear sketch with image
size $2^{O(S)}$ that is correct on every input of bounded norm. For $L_0$
estimation, we prove a lower bound for such sketches directly
(\cref{thm:modular-l0}). For matching-size estimation and subgraph counting,
we turn the sketch into a public-coin simultaneous protocol with $O(S)$-bit
messages and invoke the relevant SMP lower bound.
In the same subsection, we apply the exact reduction to constraint
satisfaction problems. It converts every turnstile streaming algorithm for a
CSP into a linear sketch of the instance. Hence the sketching dichotomy for
Max-CSPs of~\cite{CGSV24} extends to turnstile streams of polynomial length.

We will always use $W$ to denote the stream-length bound. The detailed statements and
proofs of the lists of lower bounds in
\cref{subsec:applications-lifting-real-sketches,subsec:applications-smp} are
collected in \cref{app:application-proofs}. The results on CSPs are
proved in \cref{subsec:applications-smp} itself.

\subsection{Applications via lifting real sketches}\label{subsec:applications-lifting-real-sketches}

We collect the full list of turnstile streaming lower bounds obtained via the mollified reduction.
\begin{theorem}[Lifting-based applications]
\label{thm:lifting-applications-full}
Any randomized turnstile streaming algorithm that solves the following problems on unit-update streams of length at most $W\geq \poly(n / \delta) $ with success probability at least $1-\delta$
must use the stated amount of space. Assume $0<\delta\le1/3$ and sufficiently small $\epsilon>0$.
\begin{enumerate}
\item ($L_p$ estimation, $1\le p\le 2$) Algorithms that output a $(1+\epsilon)$-approximation to the $L_p$ norm of the final frequency vector must use
$
\Omega\!\left(\min\{\epsilon^{-2}\log(1/\delta),n\}\cdot \log W\right)
$
bits. See \cref{cor:turnstile-lp-small}.
\item ($L_p$ estimation, $p>2$) Assume $\epsilon<c_p$ and $\delta > 2^{-n^{c_p}}$ for a sufficiently small constant $c_p>0$ depending on $p$. Algorithms that output a $(1+\epsilon)$-approximation to the $L_p$ norm of the final frequency vector must use
\[
\begin{aligned}
\Omega\Bigl(\min\bigl\{&n^{1-2/p}\epsilon^{-2/p}\log n\,\log^{2/p}(1/\delta)+n^{1-2/p}\epsilon^{-2}\log(1/\delta),\,n\bigr\}\cdot\log W\Bigr)
\end{aligned}
\]
bits. See \cref{cor:turnstile-lp-large}.
\item (Accurate approximation to operator norm) Set $\delta = 1 /3$ and assume $n=\Omega(\epsilon^{-2})$. Algorithms that output a $(1+\epsilon)$-approximation to the operator norm of the final $n \times \lfloor n\epsilon^2\rfloor$ matrix must use $\Omega((n^2\epsilon^2)\log W)$ bits. See \cref{cor:turnstile-op-kyfan}.
\item (Large approximation to operator norm) Set $\delta = 1 /3$ and let $1+c_0\le\alpha\le c\sqrt n$, where $c,c_0>0$ are sufficiently small absolute constants. Algorithms that output an $\alpha$-approximation in $[\|\bM\|_{op},\alpha\|\bM\|_{op}]$ to the operator norm of the final $n\times n$ matrix $\bM$ must use $\Omega((n^2/\alpha^4)\log W)$ bits. See \cref{cor:turnstile-op-kyfan}.
\item (Ky Fan norm) Set $\delta = 1 /3$ and let $s$ be an integer with $1\le s\le c\sqrt n$, where $c>0$ is a sufficiently small absolute constant. Algorithms that output a fixed small constant approximation to
the Ky Fan $s$-norm of the final $n\times n$ matrix must use $\Omega((n^2/s^2)\log W)$ bits. See \cref{cor:turnstile-op-kyfan}.
\item (Eigenvalue approximation) Set $\delta = 1 /3$ and assume $n=\Omega(\epsilon^{-2})$. Algorithms that approximate the largest eigenvalue of the final symmetric $n\times n$ matrix $\bM$ to additive error $\epsilon\|\bM\|_F$ must use $\Omega(\epsilon^{-4}\log W)$ bits. See \cref{cor:turnstile-eigen-psd}.
\item (PSD testing) Set $\delta = 1 /3$. Algorithms that distinguish whether the final symmetric $n\times n$ matrix $\bM$ is PSD or satisfies $\lambda_{\min}(\bM)<0$ and $\lambda_{\min}(\bM)\le-\epsilon\|\bM\|_p$, where $\|\cdot\|_p$ is the Schatten-$p$ norm, must use $\Omega(\min\{\epsilon^{-2p},n^2\}\log W)$ bits for $p\in[1,2]$, $\Omega(\min\{\epsilon^{-4}n^{2-4/p},n^2\}\log W)$ bits for $p\in(2,\infty)$, and $\Omega(n^2\log W)$ bits for $p=\infty$. See \cref{cor:turnstile-eigen-psd}.
\item (Compressed sensing) Set $\delta = 1/3$ and $\epsilon \geq \sqrt{k\log n/n}$. Algorithms that output, for the final frequency vector $x\in \mathbb Z^n$, an $O(\frac{\epsilon n}{\log n})$-sparse $x'$ satisfying
$
\|x-x'\|_2\le (1+\epsilon)\min_{\text{$k$-sparse }\tilde{x}}\|x-\tilde{x}\|_2
$
must use
$
\Omega\bigl((k/\epsilon)\log(n/k)\cdot \log W\bigr)
$
bits. See \cref{cor:turnstile-compressed-sensing}.
\end{enumerate}
\end{theorem}

\subsection{Applications via modular sketches and the SMP model}\label{subsec:applications-smp}

We collect the full list of turnstile streaming lower bounds obtained via the exact reduction and the SMP model. We then apply the same exact reduction to constraint satisfaction problems in turnstile streams and combine the resulting linear sketches with the sketching lower bounds of~\cite{CGSV24}.
\begin{theorem}[SMP-based applications]
\label{thm:smp-applications-full}
Any randomized turnstile streaming algorithm that solves the following problems on unit-update streams of length at most $W$ with success probability at least $2/3$ must use the stated amount of space.
\begin{enumerate}
\item ($L_0$ estimation) for $0<\epsilon\le c_0$ and $W\ge n^{C_0}$, where $c_0>0$ is a sufficiently small and $C_0>0$ a sufficiently large absolute constant, algorithms that output a $(1+\epsilon)$-approximation to the $L_0$ norm of the final frequency vector must use
$
\Omega\!\left(\epsilon^{-2}\log(1+n\epsilon^2)\,\log\log W\right)
$
bits. See \cref{cor:turnstile-l0}.
\item (Estimating maximum matching size, large approximation) for $C_0\le\alpha\le\sqrt n$ and $W\ge n^{C_0}$, where $C_0>0$ is a sufficiently large absolute constant, in strict turnstile streams whose final frequency vector encodes a multigraph on $n$ vertices, algorithms that output an estimate $z$ satisfying $\mathsf{OPT}/\alpha\le z\le\mathsf{OPT}$ must use $\Omega(n/\alpha^2)$ bits. See \cref{cor:turnstile-matching-size}.
\item (Estimating maximum matching size, $(1+\epsilon)$ approximation) for every fixed sufficiently small $\epsilon>0$ and $W\ge n^{C_0}$, where $C_0>0$ is a sufficiently large absolute constant, in strict turnstile streams whose final frequency vector encodes a multigraph on $n$ vertices, algorithms that output a $(1+\epsilon)$-approximation to the maximum matching size must use $\Omega(n^{2-O(\epsilon)})$ bits. See \cref{cor:turnstile-matching-size}.
\item (Subgraph counting) fix a constant-sized connected hypergraph $H$ with $|E(H)|>1$, $T\in \mathbb{N}$, $\epsilon\in(1/\sqrt T,1]$, $m\ge C_HT$, and $W\ge m^{C_H}$, where $C_H>0$ is a sufficiently large constant depending only on $H$. In strict turnstile streams whose final frequency vector encodes a hypergraph $G$ with at most $m$ edges, algorithms distinguishing at least $T$ copies of $H$ in $G$ from at most $(1-\epsilon)T$ must use
\[
\Omega_H\!\left(
\max\left\{
\frac{m}{(\epsilon T)^{1/\mu_2}},
\frac{m}{(\epsilon^2 T)^{1/\mu_1}}
\right\}
\right)
\]
bits, where $\mu_2\coloneqq\mathsf{MVC}_{1/2}(H)$, $\mu_1\coloneqq\max_{e\in E(H)}\mathsf{MVC}_1(H\setminus e)$, and $\mathsf{MVC}_\lambda$ is a fractional vertex-cover variant in which putting one unit of mass on a hyperedge costs $\lambda$. See \cref{cor:turnstile-subgraph-counting}.
\end{enumerate}
\end{theorem}

\paragraph{Constraint satisfaction problems in turnstile streams.}
Applied to CSPs, the exact reduction turns every turnstile streaming
algorithm that is correct on streams of polynomial length into a linear
sketch of the instance. We apply this to the Max-CSP dichotomy
of~\cite{CGSV24}.

As in~\cite{CGSV24}, fix $q,k\in\mathbb N$ and a family
$\mathcal F\subseteq\{f:[q]^k\to\{0,1\}\}$. A \emph{constraint} on the
variables $[m]$ is a pair $C=(f,(v_1,\ldots,v_k))$ with $f\in\mathcal F$ and
distinct $v_1,\ldots,v_k\in[m]$, and an assignment $a\in[q]^m$
\emph{satisfies} $C$ if $f(a_{v_1},\ldots,a_{v_k})=1$. Let $\mathcal C_m$
be the set of constraints on $[m]$, so that
$n\coloneqq|\mathcal C_m|\le|\mathcal F|\,m^k$. An \emph{instance} is a
multiset of constraints, which we identify with its frequency vector
$x\in\mathbb Z_{\ge0}^{\mathcal C_m}$, where $x_C$ is the number of copies
of $C$; thus $\|x\|_1$ is its number of constraints, and
\[
\mathcal X_B\coloneqq\{x\in\mathbb Z_{\ge0}^{\mathcal C_m}:\|x\|_1\le B\}
\]
is the set of instances with at most $B$ constraints. A \emph{turnstile
stream of constraints} starts from the empty instance and then inserts and
deletes single copies of constraints. Since these updates are $\pm e_C$, it is a
unit-update turnstile stream on $\mathbb Z^{\mathcal C_m}$ whose final
frequency vector is the final instance, and it is strict if and only if it
never deletes a copy that is not present. A \emph{task} on $\mathcal X_B$ is
a streaming problem $(\mathcal O,\mathcal R)$ in the sense of
\cref{def:streaming-problems} with
$\mathcal R\subseteq\mathcal X_B\times\mathcal O$.

\begin{corollary}[Linear sketches for CSPs in turnstile streams]
\label{cor:turnstile-csp-transfer}
For every $q,k\in\mathbb N$ and every family
$\mathcal F\subseteq\{f:[q]^k\to\{0,1\}\}$, there is a constant
$C_{\mathcal F}\ge1$ such that the following holds. Let $m,B\ge1$, $0<\delta\le1/10$, and
$W\ge(m+B+\delta^{-1})^{C_{\mathcal F}}$. Suppose a randomized streaming
algorithm using $S\ge1$ bits of space solves a task on $\mathcal X_B$ with
success probability at least $1-\delta$ on every strict turnstile stream of
constraints of length at most $W$ whose final instance lies in
$\mathcal X_B$. Then some randomized linear sketch of dimension $O(S)$ and
image size $2^{O(S)}$ solves the task with success probability at least
$1-3\delta$ on every instance in $\mathcal X_B$.
\end{corollary}

\begin{proof}
Put $D\coloneqq B$ and $R\coloneqq\lfloor W^{1/4}\rfloor$, so that
$R^4\le W$. Every $x\in\mathcal X_B$ satisfies $\|x\|_2\le\|x\|_1\le B$, so
$\mathcal X_B\subseteq B_D\cap\mathbb Z_{\ge0}^n$. Since
$n\le|\mathcal F|m^k$, we may choose $C_{\mathcal F}$ so large that
$R\ge\max\{n^{8},\delta^{-3}\}$, $D\le R^{1/3}$, and $R$ is sufficiently
large for \cref{cor:yao-minmax-exact-transfer}. The claim is then
case~\textup{(2)} of \cref{cor:yao-minmax-exact-transfer} with
$\mathcal X=\mathcal X_B$, for strict streams as in the first paragraph of
\cref{rmk:strict-robp-extension}.
\end{proof}

\paragraph{Max-CSP approximation.}
For a nonempty instance $x$ and an assignment $a\in[q]^m$, let
$\mathsf{val}_x(a)$ be the fraction of the constraints of $x$, counted
with multiplicity, that $a$ satisfies, and let
$\mathsf{val}(x)\coloneqq\max_a\mathsf{val}_x(a)$ be the optimum
value. For $0\le\beta<\gamma\le1$, the promise problem
$(\gamma,\beta)$-$\mathsf{Max\text{-}CSP}(\mathcal F)$ asks to output $1$ if
$\mathsf{val}(x)\ge\gamma$ and $0$ if $\mathsf{val}(x)\le\beta$. The dichotomy
of~\cite{CGSV24} concerns sketching algorithms. In their definition, a deterministic streaming
algorithm with $s$ bits of state is a \emph{sketching algorithm} if there are
a compression function $\mathsf{SKETCH}$, which maps every stream of
constraints to $\{0,1\}^s$, and a combination function
$\mathsf{COMB}:\{0,1\}^s\times\{0,1\}^s\to\{0,1\}^s$ such that the state
after reading a constraint $C$ in state $z$ is
$\mathsf{COMB}(z,\mathsf{SKETCH}(C))$, and
\[
\mathsf{COMB}\bigl(\mathsf{SKETCH}(\sigma),\mathsf{SKETCH}(\sigma')\bigr)
=\mathsf{SKETCH}(\sigma\sigma')
\]
for all streams $\sigma,\sigma'$, where $\sigma\sigma'$ is their
concatenation. A non-uniform randomized sketching algorithm is a
distribution over deterministic ones, and their lower bound holds against
such algorithms. A public-coin linear sketch $\mathcal A$ with image size
$2^{O(S)}$ is a non-uniform randomized sketching algorithm with $s=O(S)$. Indeed, for each fixed seed, take $\mathsf{SKETCH}(\sigma)$ to be an encoding of
$\mathcal A(x_\sigma)$, where $x_\sigma$ is the frequency vector of
$\sigma$, and let $\mathsf{COMB}$ add the two sketch values, since
$\mathcal A(x+y)=\mathcal A(x)+\mathcal A(y)$. With
\cref{cor:turnstile-csp-transfer}, their classification therefore carries over
to turnstile streams of polynomial length.

\begin{corollary}[Max-CSP dichotomy for turnstile streams]
\label{cor:turnstile-max-csp}
For every $q,k\in\mathbb N$, every family
$\mathcal F\subseteq\{f:[q]^k\to\{0,1\}\}$, and every $0\le\beta<\gamma\le1$,
one of the following two conditions holds.
\begin{enumerate}
\item $(\gamma,\beta)$-$\mathsf{Max\text{-}CSP}(\mathcal F)$ on turnstile
streams of any length polynomial in $m$ can be solved with success
probability $2/3$ by a randomized linear sketch using $O(\log^3 m)$ space.
\item For every $0<\epsilon<(\gamma-\beta)/2$, there is a constant $C$ such
that every randomized algorithm that solves
$(\gamma-\epsilon,\beta+\epsilon)$-$\mathsf{Max\text{-}CSP}(\mathcal F)$
with success probability $2/3$ on all strict turnstile streams of length at
most $m^C$ requires $\Omega(\sqrt m)$ space.
\end{enumerate}
Furthermore, if $\gamma$ and $\beta$ are rational, there is a polynomial-space
algorithm that decides which of the two conditions holds, given $\gamma$ and
$\beta$ as ratios of two integers in binary and $\mathcal F$ as the truth
tables of its functions.
\end{corollary}

\begin{proof}
By the dichotomy of~\cite{CGSV24}, either
$(\gamma,\beta)$-$\mathsf{Max\text{-}CSP}(\mathcal F)$ has an
$O(\log^3m)$-bit sketching algorithm, or for every fixed
$0<\epsilon<(\gamma-\beta)/2$, every sketching algorithm for
$(\gamma-\epsilon,\beta+\epsilon)$-$\mathsf{Max\text{-}CSP}(\mathcal F)$ uses
$\Omega(\sqrt m)$ bits. In the first case, the algorithm of~\cite{CGSV24} is
a linear sketch, so it can be maintained under insertions and deletions,
which gives the first item.

In the second case, fix $\epsilon$ and suppose that an algorithm using $S$
bits of space solves $(\gamma-\epsilon,\beta+\epsilon)$-$\mathsf{Max\text{-}CSP}(\mathcal F)$
with success probability $2/3$ on all strict turnstile streams of length at most
$W\coloneqq m^C$. Running $O(1)$ copies with independent seeds and taking the majority
answer gives an algorithm with $O(S)$ bits of space and success probability
at least $0.99$. Let $B\coloneqq Tm$, where $T$ is the number of players in
the hard distribution of~\cite[Section~5.3]{CGSV24} and $T$ depends only on $\mathcal F$, $\gamma$, $\beta$, and $\epsilon$. Every
instance in the support of this distribution has $\ell_1$ norm at most $B$,
since each player contributes a partial hypermatching on $[m]$, which has $\ell_1$ norm at most $m$. For a suitable constant $C$, the bound
$W=m^C$ satisfies the stream-length condition of
\cref{cor:turnstile-csp-transfer} with $\delta=0.01$. Applied to the promise
problem on $\mathcal X_B$, the corollary gives a public-coin linear sketch
with image size $2^{O(S)}$ that succeeds with probability at least $0.97$ on
every instance in $\mathcal X_B$. As observed above, this linear sketch is a
non-uniform randomized sketching algorithm with $O(S)$ bits in the sense
of~\cite{CGSV24}, and it succeeds on their hard instances. Their lower bound
therefore gives $S=\Omega(\sqrt m)$.

Finally, \cite[Theorem~3.4]{CGSV24} gives a polynomial-space algorithm that decides which
of the two cases of their dichotomy holds, given $\gamma$, $\beta$, and
$\mathcal F$. By the two paragraphs above, the first case gives the first item
and the second case gives the second item.
\end{proof}

%!TEX root = main.tex
\section*{Acknowledgments}
We thank Sepehr Assadi and the anonymous reviewers for drawing our attention to the best known results in the literature. We thank Noah G. Singer, Santhoshini Velusamy, and Madhur Tulsiani for pointing us to the literature on streaming algorithms for Max-CSPs. Huacheng Yu is supported in part by an NSF CAREER award CCF-2339942.

\section*{AI Methodology}
We report AI assistance using the categories recommended by the Simons Institute AI + TCS working group~\cite{SimonsAITCS26}.

\begin{center}
\renewcommand{\arraystretch}{1.15}
\begin{tabular}{@{}l@{\qquad}c@{}}
\multicolumn{1}{c}{\textbf{Activity}} & \textbf{AI used} \\
Asking the research question & No \\
Coming up with the approach & No \\
Proof development & No \\
Writing and exposition & Yes \\
Checking for bugs & Yes \\
Other supporting tasks & Yes \\
\end{tabular}
\end{center}

\textbf{Details.}
We used OpenAI Codex and Anthropic Claude to help draft and revise technical exposition, check technical arguments for potential errors, and assist with language polishing, copy-editing, and \LaTeX\ formatting. The authors supplied the research questions, proof strategy, and mathematical arguments, and reviewed and revised the AI suggestions. The authors take responsibility for the correctness and presentation of the paper.
\newpage
\appendix
\section{\texorpdfstring{Proofs for \cref{sec:coarse-large-spectrum-general}}{Proofs of the Gaussian and Fourier estimates}}
\label{app:large-spectrum-proofs}

This appendix proves the analytic facts about discrete Gaussians and their
Fourier transforms that the main text uses without proof.
\Cref{app:gaussian-sums} collects estimates for Gaussian sums on lattices.
\Cref{app:dg-fourier-decay-general,app:gamma-convolution-vs-gamma-sqrt2R-general}
use these estimates to prove the Fourier decay of $\gamma_R$ and the
comparison between $\gamma_R*\gamma_R$ and $\gamma_{\sqrt2R}$.
\Cref{app:coarse-rudin-general} proves the coarse Rudin inequality. For convenience, we restate each
result from the main text before its proof.

\subsection{Gaussian sums on lattices}
\label{app:gaussian-sums}
A \emph{full-rank lattice} $\mathcal L\subseteq \mathbb R^k$ is a discrete
additive subgroup that spans $\mathbb R^k$. Its determinant
$\det\mathcal L$ is the volume of the quotient $\mathbb R^k/\mathcal L$, and
its dual lattice is
\[
\mathcal L^*\coloneqq\{y\in\mathbb R^k:\langle y,x\rangle\in\mathbb Z
\text{ for all }x\in\mathcal L\}.
\]
For the integer lattice, $\det\mathbb Z^n=1$ and $(\mathbb Z^n)^*=\mathbb Z^n$.

\begin{fact}[Poisson summation for Gaussians {\cite[Lemmas~2.8--2.9]{mr2004worstcase}}]
\label{lem:poisson-gaussian-general}
Let $\mathcal L\subset\mathbb R^k$ be a full-rank lattice with dual
$\mathcal L^*$. For $s>0$ and $c\in\mathbb R^k$,
\[
\rho_{s,c}(\mathcal L)
=\frac{s^k}{\det\mathcal L}
  \sum_{y\in\mathcal L^*}\rho_{1/s}(y)e(-\langle y,c\rangle)
\le\rho_s(\mathcal L).
\]
In particular,
\[
\rho_s(\mathbb Z^n)=s^n\rho_{1/s}(\mathbb Z^n)\ge s^n.
\]
\end{fact}

The next lemma shows that, once $s$ exceeds a logarithmic threshold, the
normalizing constant $\rho_s(\mathbb Z^n)$ is within a constant factor of the
volume $s^n$ of the corresponding continuous Gaussian.

\begin{lemma}[Normalization of the discrete Gaussian]
\label{lem:gaussian-normalization-general}
If $s>0$ satisfies $s^2\ge\ln(8n)/\pi$, then
\[
s^n\le\rho_s(\mathbb Z^n)\le\tfrac43\,s^n.
\]
\end{lemma}
\begin{proof}
The lower bound is part of \cref{lem:poisson-gaussian-general}. For the upper
bound, \cref{lem:poisson-gaussian-general} and the product structure of
$\rho_{1/s}$ give
\[
\frac{\rho_s(\mathbb Z^n)}{s^n}
=\rho_{1/s}(\mathbb Z^n)
=\left(1+2\sum_{k\ge1}e^{-\pi s^2k^2}\right)^n.
\]
The assumption on $s$ gives $e^{-\pi s^2}\le1/(8n)$. Since $k^2\ge k$, we get
$\sum_{k\ge1}e^{-\pi s^2k^2}\le\sum_{k\ge1}(8n)^{-k}=1/(8n-1)$, and therefore
\[
\frac{\rho_s(\mathbb Z^n)}{s^n}
\le\left(1+\frac{2}{8n-1}\right)^n
\le\exp\!\left(\frac{2n}{8n-1}\right)
\le e^{2/7}<\frac43.
\]
\end{proof}

\subsection{Fourier decay of the discrete Gaussian}
\label{app:dg-fourier-decay-general}

\begin{restatement}{\Cref{lem:dg-fourier-decay-general}}
For all $R\ge 2$ and $\zeta\in \mathbb T^n$, one has
\begin{equation*}
|\widehat{\gamma_R}(\zeta)|\le \exp\left(-R^2\|\zeta\|_{\mathbb T^n}^2 / 5\right).
\end{equation*}
\end{restatement}

\begin{proof}
Both $\gamma_R$ and its Fourier transform factor over the coordinates, and
$\|\zeta\|_{\mathbb T^n}^2$ is the sum of the squared distances of the
coordinates of $\zeta$ to $\mathbb Z$. It therefore suffices to prove, in
dimension one, that
\[
0\le\widehat\gamma_R(z)\le e^{-R^2z^2/5}
\qquad\text{for } z\in[-1/2,1/2].
\]
Applying \cref{lem:poisson-gaussian-general} with $\mathcal L=\mathbb Z$ and
$s=1/R$, and using the symmetry $k\mapsto-k$, we obtain
\[
\widehat\gamma_R(z)
=\frac{\sum_{k\in\mathbb Z}e^{-\pi R^2(k+z)^2}}
       {\sum_{k\in\mathbb Z}e^{-\pi R^2k^2}}.
\]
This is positive, and the inequality in \cref{lem:poisson-gaussian-general}
shows that it is at most $1$. Put $t\coloneqq R^2z^2$. We consider two cases.

\emph{Small frequencies: $|z|\le1/(4R)$.} Since $\gamma_R$ is symmetric,
\[
1-\widehat\gamma_R(z)
=\sum_{x\in\mathbb Z}\gamma_R(x)\bigl(1-\cos(2\pi xz)\bigr)
=\frac{2}{\rho_R(\mathbb Z)}\sum_{x\in\mathbb Z}e^{-\pi x^2/R^2}\sin^2(\pi xz).
\]
If $|x|\le R/\sqrt2$, then $|\pi xz|\le\pi/(4\sqrt2)$, and on this range
$\sin u\ge u(1-u^2/6)$ gives $\sin^2(\pi xz)\ge8.8x^2z^2$. We keep only these
terms. Comparing the sum with an integral gives $\rho_R(\mathbb Z)\le1+R\le3R/2$
for $R\ge2$, and an elementary computation gives
\[
\sum_{|x|\le R/\sqrt2}x^2e^{-\pi x^2/R^2}
\ge e^{-\pi/2}\sum_{|x|\le R/\sqrt2}x^2
\ge e^{-\pi/2}R^3/12.
\]
Combining these bounds,
\[
1-\widehat\gamma_R(z)
\ge\frac{4}{3R}\cdot8.8z^2\cdot\frac{e^{-\pi/2}R^3}{12}
=\frac{35.2e^{-\pi/2}}{36}\,t\ge\frac t5.
\]
Hence $\widehat\gamma_R(z)\le1-t/5\le e^{-t/5}$.

\emph{Large frequencies: $|z|>1/(4R)$, that is, $t>1/16$.} The denominator in
the Poisson formula is at least $1$. In the numerator, the term $k=0$ equals
$e^{-\pi t}$. Since $|z|\le1/2$, every $k\ne0$ has $|k+z|\ge|k|-1/2$, and so
\[
\sum_{k\ne0}e^{-\pi R^2(k+z)^2}
\le2\sum_{j\ge1}e^{-\pi R^2(j-1/2)^2}
\le3e^{-\pi R^2/4}.
\]
Thus $\widehat\gamma_R(z)\le e^{-\pi t}+3e^{-\pi R^2/4}$.
If $1/16\le t\le1$, use $R^2/4\ge1$ and consider the ratio
\[
g(t)\coloneqq\frac{e^{-\pi t}+3e^{-\pi}}{e^{-t/5}}
=e^{(1/5-\pi)t}+3e^{-\pi}e^{t/5}.
\]
Its derivative
$g'(t)=e^{t/5}\bigl[(1/5-\pi)e^{-\pi t}+(3/5)e^{-\pi}\bigr]$ is negative,
because $e^{-\pi t}\ge e^{-\pi}$ for $t\le1$. Since $g(1/16)<1$, we get
$g(t)<1$ on this interval, which is the desired bound.
If $t\ge1$, then $t\le R^2/4$ because $|z|\le1/2$, so
$3e^{-\pi R^2/4}\le3e^{-\pi t}$ and
$\widehat\gamma_R(z)\le4e^{-\pi t}\le e^{-t/5}$.
\end{proof}

\subsection{Convolution of two discrete Gaussians}
\label{app:gamma-convolution-vs-gamma-sqrt2R-general}

\begin{restatement}{\Cref{fact:gamma-convolution-vs-gamma-sqrt2R-general}}
If $R\ge\sqrt{(2/\pi)\ln(8n)}$, then
\[
\Gamma_R(x)\coloneqq(\gamma_R*\gamma_R)(x)\le4\gamma_{\sqrt2R}(x)
\qquad(x\in\mathbb Z^n).
\]
\end{restatement}

\begin{proof}
Fix $x\in\mathbb Z^n$. For every $y\in\mathbb Z^n$,
\[
\|y\|_2^2+\|x-y\|_2^2=2\|y-x/2\|_2^2+\tfrac12\|x\|_2^2.
\]
Summing over $y$ and applying the inequality in
\cref{lem:poisson-gaussian-general} with $s=R/\sqrt2$ and $c=x/2$, we obtain
\[
\begin{aligned}
\Gamma_R(x)
&=\frac{e^{-\pi\|x\|_2^2/(2R^2)}}{\rho_R(\mathbb Z^n)^2}
  \sum_{y\in\mathbb Z^n}e^{-2\pi\|y-x/2\|_2^2/R^2}
 =\frac{e^{-\pi\|x\|_2^2/(2R^2)}}{\rho_R(\mathbb Z^n)^2}\,
  \rho_{R/\sqrt2,\,x/2}(\mathbb Z^n)\\
&\le\frac{\rho_{R/\sqrt2}(\mathbb Z^n)\rho_{\sqrt2R}(\mathbb Z^n)}
           {\rho_R(\mathbb Z^n)^2}\,\gamma_{\sqrt2R}(x),
\end{aligned}
\]
where the last step uses
$e^{-\pi\|x\|_2^2/(2R^2)}=\rho_{\sqrt2R}(\mathbb Z^n)\gamma_{\sqrt2R}(x)$.
The hypothesis on $R$ gives $s^2\ge\ln(8n)/\pi$ for each
$s\in\{R/\sqrt2,R,\sqrt2R\}$, so \cref{lem:gaussian-normalization-general}
applies to all three normalizing constants. It bounds the prefactor by
\[
\frac{(4/3)^2(R/\sqrt2)^n(\sqrt2R)^n}{R^{2n}}=\frac{16}{9}<4,
\]
as required.
\end{proof}

\subsection{Coarse Rudin inequality}
\label{app:coarse-rudin-general}

\begin{restatement}{\Cref{lem:coarse-rudin-general}}
Let $\nu$ be a probability measure on $\mathbb Z^n$, let $\kappa>0$, and
assume that $|\widehat{\nu}(\zeta)|\le \eta$ whenever
$\|\zeta\|_{\mathbb T^n}\ge \kappa$, for some $\eta\in[0,1]$. Let
$T \subseteq \mathbb T^n$ be finite and $\kappa$-dissociated, let
$c:T\to\mathbb C$, and define
$F(x)\coloneqq\Re\sum_{\xi\in T} c(\xi)e(\langle \xi,x\rangle)$ for
$x\in \mathbb Z^n$. Then for every $\sigma\ge 0$,
\begin{equation*}
\mathbb E_{x\sim\nu} e^{\sigma F(x)}
\le \exp\!\left(\frac{\sigma^2}{2}\sum_{\xi\in T}|c(\xi)|^2\right)
+ \eta\exp\!\left(\sigma\sum_{\xi\in T}|c(\xi)|\right).
\end{equation*}
\end{restatement}

\begin{proof}
Write each coefficient as $c(\xi)=|c(\xi)|e(\theta_\xi)$, and put
$a_\xi\coloneqq\sigma|c(\xi)|$ and
$u_\xi(x)\coloneqq\cos\bigl(2\pi(\langle\xi,x\rangle+\theta_\xi)\bigr)$, so that
$F=\sum_{\xi\in T}|c(\xi)|u_\xi$ and $|u_\xi|\le1$.
For $a\ge0$ and $|t|\le1$, convexity of $s\mapsto e^{as}$ on $[-1,1]$ gives
\[
e^{at}\le\frac{1+t}{2}e^{a}+\frac{1-t}{2}e^{-a}=\cosh a+t\sinh a.
\]
Each right-hand side is nonnegative, so we may multiply these bounds over
$\xi\in T$:
\[
e^{\sigma F(x)}
\le\prod_{\xi\in T}\bigl(\cosh a_\xi+u_\xi(x)\sinh a_\xi\bigr).
\]
Write $u_\xi(x)=\frac12 e(\langle\xi,x\rangle+\theta_\xi)
+\frac12 e(-\langle\xi,x\rangle-\theta_\xi)$ and expand the product. This
gives a sum over sign patterns $\varepsilon\in\{-1,0,1\}^T$. The term indexed
by $\varepsilon$ is a constant of modulus
\[
w_\varepsilon\coloneqq\prod_{\xi:\varepsilon_\xi=0}\cosh a_\xi
              \prod_{\xi:\varepsilon_\xi\ne0}\tfrac12\sinh a_\xi
\]
times the character $x\mapsto e(\langle\zeta_\varepsilon,x\rangle)$, where
$\zeta_\varepsilon\coloneqq\sum_{\xi\in T}\varepsilon_\xi\xi$.

For $\varepsilon=0$, the term is the constant
$\prod_\xi\cosh a_\xi\le\exp(\frac{\sigma^2}{2}\sum_\xi|c(\xi)|^2)$, by
$\cosh a\le e^{a^2/2}$. For $\varepsilon\ne0$, the frequency
$\zeta_\varepsilon$ is a nonempty signed sum of elements of $T$, so
$\|\zeta_\varepsilon\|_{\mathbb T^n}\ge\kappa$ by the definition of dissociation. Hence
$\bigl|\mathbb E_{x\sim\nu}e(\langle\zeta_\varepsilon,x\rangle)\bigr|
=|\widehat\nu(-\zeta_\varepsilon)|\le\eta$, and the expectation of this term
has modulus at most $\eta w_\varepsilon$. Finally, summing over all sign
patterns,
\[
\sum_{\varepsilon}w_\varepsilon
=\prod_{\xi\in T}\bigl(\cosh a_\xi+\sinh a_\xi\bigr)
=\exp\Bigl(\sigma\sum_{\xi\in T}|c(\xi)|\Bigr).
\]
Taking expectations in the expanded product proves the lemma.
\end{proof}

% !TEX root = main.tex

\section{\texorpdfstring{Proofs for \cref{sec:application}}{Proofs for applications}}
\label{app:application-proofs}

All streaming lower bounds in this appendix concern randomized algorithms
on unit-update turnstile streams of length at most $W$, with space $S$ measured in bits.
The success probability is $2/3$ unless a failure parameter $\delta$ is stated.
We write $u\gg v$ or $v\ll u$ if $u\ge Cv$ for a sufficiently large constant
$C$, which may depend on fixed parameters of the problem such as $p$.

\subsection{Lattice preliminaries}
We use the lattice notation of \cref{app:gaussian-sums}. A lattice of rank
$k<n$ in $\mathbb R^n$ is regarded as a full-rank lattice in its real span,
so its determinant, its dual lattice, and Poisson summation
(\cref{lem:poisson-gaussian-general}) are taken within that span. The zero
lattice has determinant $1$.
The $i$-th successive minimum $\lambda_i(\mathcal L)$ of a rank-$k$ lattice
$\mathcal L$ is the smallest radius of a ball about the origin containing $i$
linearly independent lattice vectors, and
$\lambda_{\max}(\mathcal L)\coloneqq\lambda_k(\mathcal L)$, with
$\lambda_{\max}\coloneqq0$ for the zero lattice. For a matrix $\bA$ with rational
row space, $\calL^\perp(\bA)\coloneqq\ker(\bA)\cap\mathbb Z^n$ is a lattice of rank
$n-\operatorname{rank}(\bA)$.

For positive definite $\bSigma$ and $c\in\mathbb R^n$, put
\[
\rho_{\bSigma,c}(x)\coloneqq e^{-\pi(x-c)^\top\bSigma^{-1}(x-c)},\qquad
\calD_{\mathcal L,\bSigma}(x)
\coloneqq\frac{\rho_{\bSigma,0}(x)}{\sum_{z\in\mathcal L}\rho_{\bSigma,0}(z)}
\qquad(x\in\mathcal L),
\]
so that $\calD_{\mathcal L,R^2\Id_n}=\calD_{\mathcal L,R}$ as in
\cref{def:discrete-gaussian-notation-general}. We write
$\DG(0,\bSigma)\coloneqq\calD_{\mathbb Z^n,\bSigma}$. In particular,
$\DG(0,R^2\Id_n)=\gamma_R$.
For $\varepsilon>0$, the smoothing parameter $\eta_\varepsilon(\mathcal L)$
is the smallest $s>0$ with $\rho_{1/s}(\mathcal L^*\setminus\{0\})\le\varepsilon$.
We use the following standard bound on it.
\begin{fact}[Smoothing parameter bound {\cite[Lemma~3.3]{mr2004worstcase}}]
\label{lem:smoothing-parameter-bound-general}
For every lattice $\mathcal L$ of rank $k\ge1$ and every $\varepsilon>0$,
\[
\eta_\varepsilon(\mathcal L)
\le\sqrt{\frac{\ln(2k(1+1/\varepsilon))}{\pi}}\,
       \lambda_k(\mathcal L).
\]
\end{fact}

\subsection{Applications via lifting real sketches}
\label{app:applications-lifting}

We prove the bounds summarized in \cref{thm:lifting-applications-full}.
We first state and prove the lifting theorem, then record the Gaussian
estimates used by all applications, and finally give the hard inputs and
reductions for each problem.

We write $\calN(c,\bSigma)$ for the Gaussian with mean $c$ and covariance
$\bSigma$. The continuous law matching the integer Gaussian $\DG(0,\bSigma)$ is
$\calN(0,\bSigma/(2\pi))$.

We will use the lifting theorem~\cite[Theorem~3.4]{gribelyuk2025lifting} with explicit
entry-bound and error parameters. Such lifting theorems show that a discrete sketch of bounded entries can be lifted to a real sketch.

\begin{theorem}[Lifting integer sketches]
\label{thm:linear-sketches}
Let $\mathcal M$ be a finite set, and let $B$ be an $\mathcal M$-valued random variable. Conditioned on $B=b$, let $Y\sim \mathcal I_b$ on $\mathbb Z^n$, let $X\sim \DG(0,\bSigma_b)$, and let $G\sim \calN(0,\bSigma_b/(2\pi))$, where $\bSigma_b=\bS_b^\top\bS_b\in\mathbb R^{n\times n}$ is full rank. Assume that, conditioned on $B=b$, the random variables $X,G,Y$ are mutually independent.

Let $Q\ge n^5$ be sufficiently large, and let $\bA_0\in\mathbb Z^{r\times n}$ have $r\le n/4$ and entries bounded in absolute value by $Q$. Preprocessing as in~\cite[Lemma~3.1]{gribelyuk2025lifting} gives a full-row-rank $\widetilde{\bA}\in\mathbb Z^{\ell\times n}$ with $\ell\le 4r$ whose row space contains that of $\bA_0$. Assume
\[
Q^{6/5}\ge \sigma_1(\bS_b)\ge \sigma_n(\bS_b)\ge nQ\log Q
\qquad\text{for every }b\in\mathcal M.
\]
Let $\bA$ have orthonormal rows spanning the same row space as $\widetilde{\bA}$. Then, for every post-processing function $g$, there is a function $h$ such that
\[
\dTV\!\left(
\Law\bigl(B,g(\widetilde{\bA}(X+Y))\bigr),
\Law\bigl(B,h(\bA(G+Y))\bigr)
\right)\le Q^{-1/5}.
\]
The same bound holds conditioned on $B=b$, for every $b\in\mathcal M$ with $\Pr[B=b]>0$.
\end{theorem}

The proof follows that of~\cite[Theorem~3.4]{gribelyuk2025lifting}. Its main
ingredient is the following version of~\cite[Lemma~3.2]{gribelyuk2025lifting}
with explicit parameters: after adding a uniform point of the Voronoi cell of the
lattice $\bA\mathbb Z^n$, the projection $\bA X$ of the discrete Gaussian is
close in total variation to the projection $\bA G$ of the continuous Gaussian.

\begin{lemma}[Discrete-to-continuous comparison]
\label{lem:disc-to-cont-r-parameterized}
Let $Q\ge n^5$ be sufficiently large, let $\bA\in\mathbb R^{m\times n}$ have orthonormal rows, and assume that there exists a matrix with rational entries and the same rowspan as $\bA$. Let $\bSigma\coloneqq\bS^\top\bS$ be full rank, let $X\sim \DG(0,\bSigma)$, let $G\sim\calN(0,\bSigma/(2\pi))$, and let $\eta$ be uniform in the Voronoi cell of zero of $\bA\mathbb Z^n$, independently of everything else. Assume that
\[
Q^{6/5}\ge\sigma_1(\bS)\ge\sigma_n(\bS)\ge nQ\log Q,
\qquad
\lambda_{\max}(\calL^\perp(\bA))\le Q\sqrt n.
\]
Then the law of $\bA X+\eta$ is within total variation distance at most $Q^{-1/5}$ of the law of $\bA G$.
\end{lemma}

\begin{proof}
Put $\calL_{\bA}=\bA\mathbb Z^n$, $C=\bA\bSigma\bA^\top$, and
$s=\sigma_n(\bS)$. Let $p$ be the density of $\bA X+\eta$,
$p'$ that of $Y_{\rm lat}+\eta$ for
$Y_{\rm lat}\sim\calD_{\calL_{\bA},C}$, and $q$ that of $\bA G$.
With $\varepsilon=Q^{-2/5}/100$, the hypotheses and
\cref{lem:smoothing-parameter-bound-general} give
$\eta_\varepsilon(K)<1$ for $K\coloneqq\bSigma^{-1/2}\calL^\perp(\bA)$. We compare
the image of a discrete Gaussian under $\bA$ with the discrete Gaussian on
$\bA\mathbb Z^n$ as in~\cite[Theorem~2.2.18]{gribelyuk2025lifting}. For
$z\in\calL_{\bA}$, the points $\bSigma^{-1/2}x$ with $x\in\mathbb Z^n$ and
$\bA x=z$ form a coset of $K$ in the affine space
$\{y:\bA\bSigma^{1/2}y=z\}$, whose closest point to the origin has squared
norm $z^\top C^{-1}z$. Hence
\[
\sum_{x\in\mathbb Z^n:\,\bA x=z}e^{-\pi x^\top\bSigma^{-1}x}
=e^{-\pi z^\top C^{-1}z}\rho_1(K+t_z)
\]
for some $t_z$ in the span of $K$, and Poisson summation
(\cref{lem:poisson-gaussian-general}) with $\eta_\varepsilon(K)<1$ gives
$\rho_1(K+t_z)\det K\in[1-\varepsilon,1+\varepsilon]$. This is trivial if
$m=n$, when $K=\{0\}$ and $t_z=0$. Normalizing over $z$ gives
\[
p/p'\in[1-3\varepsilon,1+3\varepsilon]
\]
on every cell.

Let $\calC$ be the Voronoi cell of $\calL_{\bA}$ centered at $0$, and let $\phi$
map each point to its nearest point in $\calL_{\bA}$, so that every $u$ can be
written as $u=\phi(u)+w$ with $w\in\calC$. The cell $\calC$ has radius at
most $\sqrt n/2$. Indeed, for every $u$, let $x\in\mathbb Z^n$ be the coordinatewise
rounding of $\bA^\top u$, then $\bA x\in\calL_{\bA}$ and
$\|u-\bA x\|_2=\|\bA(\bA^\top u-x)\|_2\le\|\bA^\top u-x\|_2\le\sqrt n/2$.
Set
\[
L=10Q^{6/5}\sqrt{n\log Q},\qquad
U=\{u:\|\phi(u)\|_2\le L\}.
\]

We first bound the mass of $p$ outside $U$. Since $\phi(\bA X+\eta)=\bA X$
almost surely and $\|\bA X\|_2\le\|X\|_2$, we have
$p(U^c)\le\Pr[\|X\|_2>L]$.
Since $\bSigma^{-1/2}X\sim\calD_{\bSigma^{-1/2}\mathbb Z^n,1}$
and $\|X\|_2\le\sigma_1(\bS)\|\bSigma^{-1/2}X\|_2$,
the tail bound in \cref{lem:gaussian-moment-general} with
$L/\sigma_1(\bS)\ge10\sqrt{n\log Q}$ gives
\[
p(U^c)\le\Pr[\|X\|_2>L]\le2^{n/2}e^{-50\pi n\log Q}\le Q^{-10}.
\]

We next compare $p$ and $q$ on $U$. Fix $u\in U$ and write $u=z+w$ with
$z=\phi(u)$ and $w\in\calC$. Since $q(u)$ is proportional to
$e^{-\pi u^\top C^{-1}u}$ and $C\succeq s^2\Id$,
\[
\left|\ln\frac{q(u)}{q(z)}\right|
=\pi\bigl|2z^\top C^{-1}w+w^\top C^{-1}w\bigr|
\le\frac{\pi(2L\|w\|_2+\|w\|_2^2)}{s^2}
\le40Q^{-4/5},
\]
where the last step uses $\|w\|_2\le\sqrt n/2$ and $s\ge nQ\log Q$.
The law $\calD_{\calL_{\bA},C}$ gives each $z\in\calL_{\bA}$ mass
$q(z)/\sum_{z'\in\calL_{\bA}}q(z')$, and adding $\eta$ spreads this mass
uniformly over $z+\calC$, which has volume $\det\calL_{\bA}$. Hence
$p'(u)=q(z)/Z$, where Poisson summation (\cref{lem:poisson-gaussian-general}) gives the exact normalization factor
\[
Z\coloneqq\det\calL_{\bA}\sum_{z'\in\calL_{\bA}}q(z')
=\sum_{v\in\calL_{\bA}^*}e^{-\pi v^\top Cv}.
\]
Every $v\in\calL_{\bA}^*$ satisfies $\bA^\top v\in\mathbb Z^n$,
$\|\bA^\top v\|_2=\|v\|_2$, and $v^\top Cv\ge s^2\|v\|_2^2$, so
$1\le Z\le\sum_{y\in\mathbb Z^n}e^{-\pi s^2\|y\|_2^2}\le1+Q^{-10}$.
Combining these bounds with $p/p'\in[1-3\varepsilon,1+3\varepsilon]$ gives,
for sufficiently large $Q$,
\[
\left|\frac{p(u)}{q(u)}-1\right|
=\left|\frac{p(u)}{p'(u)}\cdot\frac{q(z)}{q(u)}\cdot\frac1Z-1\right|
\le a\coloneqq3Q^{-2/5}.
\]

Finally, write $b=Q^{-10}$. The last bound gives
$\int_U|p-q|\le a\,q(U)\le a$, and therefore
$q(U^c)=1-q(U)\le1-p(U)+a\le a+b$. Hence
\[
\dTV(p,q)
=\frac12\Bigl(\int_U|p-q|+\int_{U^c}|p-q|\Bigr)
\le\frac12\bigl(a+p(U^c)+q(U^c)\bigr)
\le a+b\le Q^{-1/5}.
\]
\end{proof}

\begin{proof}[Proof of \cref{thm:linear-sketches}]
We first justify preprocessing for arbitrary $Q\ge3n$.
For $r\ge1$, every coordinate of $\bA_0x$ with
$x\in\{0,\ldots,\lfloor Q\rfloor\}^n$ lies in $[-nQ^2,nQ^2]$, so this cube
has at most $(3nQ^2)^r$ images under $\bA_0$. Hence some fiber $F$ has at
least $Q^n/(3nQ^2)^r\ge Q^{n-3r}$ points, using $Q\ge3n$. Differences of
points of $F$ are kernel vectors with coordinates bounded by $Q$. If such
vectors spanned a space $W$ of dimension at most $n-4r$, then $F$ would lie
in a translate of $W$, and projecting onto $\dim W$ coordinates on which the
projection is injective on $W$ would give $|F|\le(Q+1)^{n-4r}<Q^{n-3r}$, a
contradiction.
Choose $n-4r$ independent such vectors, let $V$ be their span, and take
an integer row basis of $V^\perp$ as $\widetilde{\bA}$.
It has $4r$ rows and contains the row space of $\bA_0$, since
$V\subseteq\ker\bA_0$. Its integer kernel $V\cap\mathbb Z^n$ contains the
$n-4r$ chosen vectors, each of norm at most $Q\sqrt n$, so
$\lambda_{\max}(\calL^\perp(\widetilde{\bA}))\le Q\sqrt n$.
For $4r=n$, use $\Id_n$. For $r=0$, $g$ is constant and we take $h=g$.

Let $\phi$ be nearest-point rounding to $\bA\mathbb Z^n$, choose invertible
$\mathbf P$ with $\widetilde{\bA}=\mathbf P\bA$, and put $h(u)=g(\mathbf P\phi(u))$.
Conditioned on $B=b$ for each $b$ with $\Pr[B=b]>0$, \cref{lem:disc-to-cont-r-parameterized} and
independence of $Y$ give
\[
\dTV\bigl(\Law(\bA(X+Y)+\eta\mid B),
          \Law(\bA(G+Y)\mid B)\bigr)\le Q^{-1/5}.
\]
Because $X+Y$ is integral,
$h(\bA(X+Y)+\eta)=g(\widetilde{\bA}(X+Y))$ almost surely.
The data-processing inequality for total variation therefore gives the
joint-law bound conditioned on each such $b$, and averaging over $B$ gives the
unconditional bound in the theorem.
\end{proof}

We next record the Gaussian estimates shared by the applications below.
Write $[v]$ for coordinatewise nearest-integer rounding. The continuous hard
distributions we consider are of the form $G+V$, where $G$ is a large Gaussian
background and $V$ is the offset vector that encodes the hardness. To prove
hardness for discrete sketches, we discretize such a continuous hard input
$G+V$ as $X+[V]$, replacing the Gaussian background $G$ by its matching
integer Gaussian $X$ and rounding the whole offset vector $V$.

\begin{lemma}[Rounding Gaussian hard inputs]
\label{lem:rounded-gaussian-hard-law}
Let $G\sim\calN(0,\sigma^2 \Id_N)$, let $V\in\mathbb R^N$ be independent of $G$, and let $X\sim\DG(0,2\pi\sigma^2 \Id_N)$. There is a coupling in which $X$ is a function of $G$ and $\Pr[\|X+[V]-(G+V)\|_2>\sqrt N]\le CN/\sigma$. For every row-orthonormal $\bA$, $\dTV(\Law(\bA(G+[V])),\Law(\bA(G+V)))\le C\sqrt N/\sigma$.

Furthermore, both bounds also hold conditioned on any event under which $G$ still has law $\calN(0,\sigma^2\Id_N)$ and is independent of $V$.
\end{lemma}
\begin{proof}
For the one-dimensional Gaussian density $\phi_\sigma$, the $L^1$ error of replacing each unit cell by its center value is at most $\int|\phi_\sigma'|=O(\sigma^{-1})$. The normalized masses at the cell centers form $\DG(0,2\pi\sigma^2)$. Since $G_i$ has no atoms, there is therefore a fixed map $T:\mathbb R\to\mathbb Z$ with $T(G_i)\sim\DG(0,2\pi\sigma^2)$ and $T(G_i)=[G_i]$ except with probability $O(\sigma^{-1})$, which moves only the excess mass of each cell, and we put $X_i=T(G_i)$. By a union bound, $X=[G]$ except with probability $CN/\sigma$, and on this event
\[
\|X+[V]-(G+V)\|_2\le\|[G]-G\|_2+\|[V]-V\|_2\le\sqrt N.
\]
For the second bound, condition on $V=v$. Then $\bA(G+[v])$ and $\bA(G+v)$ are Gaussians with the same covariance $\sigma^2\Id$, and their means differ by $\|\bA([v]-v)\|_2\le\|[v]-v\|_2\le\sqrt N/2$. Their relative entropy is therefore at most $N/(8\sigma^2)$, and Pinsker's inequality~\cite[Lemma~2.5]{tsybakov2009introduction} bounds their total variation distance by $\sqrt N/(4\sigma)$. Averaging over $V$ gives the second bound. Both arguments use only that $G$ has law $\calN(0,\sigma^2\Id_N)$ and is independent of $V$, and that $X$ is a function of $G$, so they also hold conditioned on any event preserving these two properties.
\end{proof}

\paragraph{Reduction to real sketches.}
All lower bounds in this subsection follow from the same reduction from a
streaming algorithm to a real sketch $y\mapsto h(\bA y)$, where $\bA$ has
orthonormal rows. We use standard amplification
(majority or median over a fixed number of independent copies reduces the
failure probability from $\delta\le1/3$ to $c\delta$ for any fixed $c>0$, at the cost of a
constant factor in space) and truncation (conditioning the hard law on an
event of probability $1-\xi$ changes success probabilities by $O(\xi)$).

Throughout this subsection, $W\ge(n/\delta)^{C_0}$ for a sufficiently large
constant $C_0$, where $\delta$ is the failure probability ($\delta=1/3$ unless
stated otherwise), and
\begin{equation}\label{eq:common-parameters}
R\coloneqq\lfloor W^{1/4}\rfloor,\qquad \Lambda\coloneqq\lfloor R^{9/8}\rfloor,\qquad
\sigma\coloneqq\Lambda/\sqrt{2\pi}.
\end{equation}
Then $\log R=\Theta(\log W)$, and taking $C_0$ large makes $R\ge(n/\delta)^{17}$
and $R \gg 1$.

Each application below specifies a label random variable $B$ as in \cref{thm:linear-sketches} and a continuous hard input $G+V$ in
\emph{ambient dimension} $N$, the number of coordinates of the input.
Note that the parameter $n$ here is the size parameter of the problem in \cref{thm:lifting-applications-full},
namely the length of the input vector or the number of rows of the input
matrix, so $N\le n^2$ in every application.
In most applications, $B$ is uniform in $\{0,1\}$ and indicates which of two
hard distributions the offset $V$ is drawn from. Here $G\sim\calN(0,\sigma^2\Id_N)$ is independent of $(B,V)$. The rounded
discrete input is $Y=X+[V]$, where $X\sim\DG(0,\Lambda^2\Id_N)$ is the
matching integer Gaussian, also independent of $(B,V)$. Each application
checks the following two items.
\begin{enumerate}[(a)]
\item The problem is smooth on average on the law of $Y$ with respect to
$\gamma_R$ (\cref{def:smoothness}), and a valid answer for $Y$ determines $B$,
both up to error $\ll\delta$.
\item Every real sketch that determines $B$ from $G+V$ has dimension at least $d$.
\end{enumerate}
As in~\cite{gribelyuk2025lifting}, each continuous hard input is the Gaussian
hard distribution of a known real-valued sketching lower
bound~\cite{GangulyW18,LiW16,NeedellSW22,SwartworthW23,PriceW11}, and item~(b)
is that lower bound. The reduction then has the
following four steps.
\begin{enumerate}
\item\label{red:transfer} Amplify the failure probability to $ \ll \delta$ and
condition the hard law on $\|Y\|_2\le R^{5/4}$. Since $R\ge n^{17}$ and $R\gg1$,
\begin{equation}\label{eq:lambda-ratios}
\Lambda/R\ge R^{1/8}/2\gg n^2\log R,
\qquad
R^{5/4}/\Lambda\ge R^{1/8}\gg n^2\log R.
\end{equation}
In every application, the offset $[V]$ has norm at most $n^2\Lambda\log R$
except with probability $1/R$. This follows from $\|[V]-V\|_2\le n$,
$\|uv^\top\|_F=\|u\|_2\|v\|_2$, and the bound
$\Pr[\|g\|_2>\sqrt k+\sqrt{2\ln(2nR)}]\le1/(2nR)$ for $g\sim\calN(0,\Id_k)$,
applied to each of the at most $2n$ standard Gaussian vectors from which the
offsets below are built. Hence the tail bound in
\cref{lem:gaussian-moment-general} for $X$ and the second bound
in~\eqref{eq:lambda-ratios} give $\|Y\|_2\le R^{5/4}$ except with
probability $2/R$. By
\cref{rmk:mollified-stream-distribution}, the mollified streams then have
length at most $R^{7/2}\le W$. Since $R\ge(n/\delta)^{17}$, $R\gg1$, and
$N\le n^2$, we
have $R\ge N^8$ and $R^{-1/10}\ll\delta$. Hence, using~(a), case~\textup{(1)} of
\cref{thm:randomized-mollified-transfer} (case~\textup{(2)} for relations),
in ambient dimension $N$, gives an integer sketch of dimension $r=O(S/\log R)$
with entries bounded by $R$ whose decoder determines $B$ except with
probability $\ll\delta$. All claimed dimension lower bounds are $O(N)$, so if
$r>N/4$, then $S=\Omega(N\log R)$ already gives the desired space bound. We
therefore assume $r\le N/4$.
\item\label{red:lift} The sketch and decoder from step~\ref{red:transfer} are
fixed, so they also determine $B$ under the untruncated law of $Y$, in which
$X\sim\DG(0,\Lambda^2\Id_N)$ is independent of $(B,V)$, and this increases the
failure probability by at most $\Pr[\|Y\|_2>R^{5/4}]\le2/R\ll\delta$.
By the first bound in~\eqref{eq:lambda-ratios} and
$N\le n^2$, we have $NR\log R\le\Lambda\le R^{6/5}$. Hence
\cref{thm:linear-sketches} with $Q=R$ and $\bSigma_b=\Lambda^2\Id_N$ applies
to this sketch and its deterministic decoder under the untruncated law. This gives a real sketch of
dimension $O(r)$ and loses at most $R^{-1/5}\ll\delta$ in success probability.
\item\label{red:round} \cref{lem:rounded-gaussian-hard-law} removes the
rounding of the offsets at an additional cost of $O(\sqrt N/\Lambda)=O(1/R)$
in success probability.
\item\label{red:space} By~(b), $r=\Omega(d)$, and hence
$S=\Omega(r\log R)=\Omega(d\log W)$.
\end{enumerate}
Therefore, it remains to check items~(a) and~(b) for each application.
We first give a common tool for verifying item~(a) when the problem is to
estimate a Lipschitz function of the input, such as a norm. The following
subsections then treat each application in turn. For each application, we
specify the label $B$ and the hard input $G+V$, verify item~(a), and derive
item~(b) from the corresponding real-valued lower bound.

\paragraph{Norm estimation.}
Item~(a) follows from the lemma below whenever the problem asks to estimate a
function $q:\mathbb R^N\to\mathbb R$ that is $L$-Lipschitz with respect to
$\|\cdot\|_2$. Examples include the $L_p$ norms, the Schatten-$p$ and Ky Fan
norms of a matrix,
and in fact every norm on $\mathbb R^N$, as well as the largest eigenvalue of a
symmetric matrix. For such a function $q$ and thresholds $a<b$, set $f(y)=0$ for $q(y)\le a$, $f(y)=1$ for $q(y)\ge b$, and $f(y)=*$ otherwise.

\begin{lemma}[Lipschitz functions under rounding and noise]
\label{lem:norm-rounding-noise}
Let $N\le n^2$, $\Lambda\gg n^2R\log R$, and $\sigma=\Lambda/\sqrt{2\pi}$. Let $G\sim\calN(0,\sigma^2\Id_N)$, let $V\in\mathbb R^N$ be independent of $G$, couple $X\sim\DG(0,\Lambda^2\Id_N)$ with $G$ as in \cref{lem:rounded-gaussian-hard-law}, and put $Y=X+[V]$. Let $Z\sim\gamma_R$ be independent of $(G,V,X)$. Then there is an event of probability at least $1-2/R$ on which, for every $L>0$, every function $q:\mathbb R^N\to\mathbb R$ that is $L$-Lipschitz with respect to $\|\cdot\|_2$ satisfies
\[
 |q(Y)-q(G+V)|\le L\Lambda/n
 \qquad\text{and}\qquad
 |q(Y+Z)-q(Y)|\le L\Lambda/n.
\]
\end{lemma}

\begin{proof}
By \cref{lem:rounded-gaussian-hard-law}, $\|Y-(G+V)\|_2\le\sqrt N$ outside an event of probability $O(N/\Lambda)$, and then $|q(Y)-q(G+V)|\le L\sqrt N$. The tail bound in \cref{lem:gaussian-moment-general} gives $\|Z\|_2\le R\sqrt{N+\ln R}$ outside an event of probability $1/R$, and then $|q(Y+Z)-q(Y)|\le L\|Z\|_2$. Since $N\le n^2$ and $\Lambda\gg n^2R\log R$, both $\sqrt N$ and $R\sqrt{N+\ln R}$ are at most $\Lambda/n$, and $O(N/\Lambda)\le1/R$. 
\end{proof}

We apply this with $B$ uniform in $\{0,1\}$ and thresholds chosen so that
$f(Y)=B$ with high probability. Precisely, suppose that
$q(G+V)\le a-2L\Lambda/n$ given $B=0$ and $q(G+V)\ge b+2L\Lambda/n$ given
$B=1$, each except with probability $\ll\delta$. By the first bound
in~\eqref{eq:lambda-ratios} the lemma applies, and since $R^{-1}\ll\delta$, we
get $f(Y)=f(Y+Z)=B$ except with probability $\ll\delta$. Since the event in
the lemma does not depend on $q$, the same holds when $f$ is defined by
several such threshold conditions on Lipschitz functions, each satisfied by
$G+V$ with margin $2L\Lambda/n$.

If moreover $b/a$ is
large enough that every valid estimate of $q(y)$ determines $f(y)$, then a
streaming algorithm that estimates $q$ correctly also computes $f$ correctly,
and the above gives item~(a) for $f$.

\subsubsection{\texorpdfstring{$L_p$}{Lp} estimation}

\paragraph{The range $1\le p\le2$.}

Following~\cite{gribelyuk2025lifting}, let $\tau>0$, let $B$ be uniform in
$\{0,1\}$, and put $V=0$ if $B=0$ and
\[
V\sim\calN\bigl(0,((1+10\tau)^2-1)\sigma^2\Id_n\bigr)
\]
if $B=1$, independently of $G$. Thus $G+V$ has law $\calN(0,\sigma^2\Id_n)$
given $B=0$ and $\calN(0,(1+10\tau)^2\sigma^2\Id_n)$ given $B=1$. We take
$q=\|\cdot\|_p$ with thresholds
\[
a=(1+2\tau)T_p,\qquad b=(1+8\tau)T_p,\qquad
\text{where }T_p\coloneqq\mathbb E_{g\sim\calN(0,\sigma^2\Id_n)}\|g\|_p.
\]

\begin{theorem}
\label{cor:turnstile-lp-small}
Assume $1\le p\le 2$, $0<\epsilon<c_0$, $0<\delta\le1/3$, and $W\ge (n/\delta)^{C_0}$, where $c_0>0$ is a sufficiently small and $C_0$ a sufficiently large absolute constant. Then $(1+\epsilon)$-approximate $L_p$ estimation with failure probability $\delta$ requires
$
\Omega\Bigl(\min\{\epsilon^{-2}\log(1/\delta),n\}\cdot \log W\Bigr)
$
bits of space.
\end{theorem}

\begin{proof}
We first assume that $n\ge c_0^{-3}$, and treat $n<c_0^{-3}$ at the end. By amplification, we may assume that $\delta$ is at most a sufficiently small absolute constant. We may also assume $\delta\ge e^{-c_0^2n}$, since for even smaller $\delta$, we can set $\delta = e^{-c_0^2n}$, and still prove the claimed bound $\Omega(n\log W)$. Put
\[
 \tau\coloneqq\max\left\{\epsilon,A\sqrt{\log(1/\delta)/n}\right\},
\]
where $A$ is a sufficiently large absolute constant and $c_0$ is chosen after $A$. Then $\tau\le Ac_0$ is small, since $\epsilon<c_0$ and $\delta\ge e^{-c_0^2n}$.

For item~(a), Gaussian concentration, as in~\cite[Lemma~5.1.2]{gribelyuk2025lifting}, gives $\|G+V\|_p\le(1+3\tau/2)T_p$ given $B=0$ and $\|G+V\|_p\ge(1+17\tau/2)T_p$ given $B=1$, each except with probability $e^{-\Omega(\tau^2n)}\ll\delta$, since $\tau^2n\ge A^2\log(1/\delta)$ and $\delta\le1/3$. The function $q=\|\cdot\|_p$ is $L$-Lipschitz with $L=n^{1/p-1/2}$, and $T_p\ge n^{1/p-1}\mathbb E_{g\sim\calN(0,\sigma^2\Id_n)}\|g\|_1=\Lambda n^{1/p}/\pi$ by H\"older's inequality. Since $\tau\ge An^{-1/2}$, the thresholds $a=(1+2\tau)T_p$ and $b=(1+8\tau)T_p$ therefore satisfy $a-(1+3\tau/2)T_p=(1+17\tau/2)T_p-b=\tau T_p/2\ge2L\Lambda/n$, as required in the paragraph after \cref{lem:norm-rounding-noise}. Since $\epsilon\le\tau$ and $\epsilon<1/2$, we have $b/a\ge(1+8\epsilon)/(1+2\epsilon)>(1+\epsilon)^2$, so every two-sided $(1+\epsilon)$-approximation of $\|y\|_p$ determines $f(y)$. This gives item~(a).

For item~(b), determining $B$ from $G+V$ means distinguishing $\calN(0,\sigma^2\Id_n)$ from $\calN(0,(1+10\tau)^2\sigma^2\Id_n)$. By the real-sketch lower bound of Ganguly and Woodruff~\cite[Section~2]{GangulyW18}, doing so with error at most $\delta$ requires dimension
\[
 \Omega\bigl(\tau^{-2}\log(1/\delta)\bigr)
   =\Omega\!\left(\min\{\epsilon^{-2}\log(1/\delta),n\}\right),
\]
Therefore, steps~\ref{red:transfer}--\ref{red:space} give the result.

For $n<c_0^{-3}$, the claimed bound is $O(\log W)$, so it suffices to show $S=\Omega(\log W)$, and we may take $\delta$ to be a small absolute constant by amplification. We use inputs only supported on the first coordinate (that is, reducing to $n=1$), and let $G+V$ have law $\calN(0,\sigma^2)$ given $B=0$ and $\calN(0,A^2\sigma^2)$ given $B=1$. Thresholding $|\cdot|$ at $A^{1/3}\sigma$ and $A^{2/3}\sigma$ distinguishes these two laws with the margins required in the paragraph after \cref{lem:norm-rounding-noise}, which gives item~(a), and item~(b) holds with $d=1$, since a real sketch of dimension $0$ is constant. Steps~\ref{red:transfer}--\ref{red:space} give $S=\Omega(\log W)$.
\end{proof}

\paragraph{The range $p>2$.}

We use the two hard pairs of~\cite[Appendix~D]{GangulyW18}, scaled by $\sigma$. In each pair, $B$ is uniform in $\{0,1\}$, and the input $G+V$ is the first member of the pair given $B=0$ and the second given $B=1$, where $G\sim\calN(0,\sigma^2\Id_n)$. For an approximation parameter $\tau$, the single-spike pair of~\cite[Appendix~D.1.2]{GangulyW18}, which gives the bound $\Omega(\min\{n^{1-2/p}\epsilon^{-2}\log(1/\delta),n\}\log W)$, is
\begin{equation}\label{eq:large-p-single-spike}
 G+\sigma n^{1/p}e_I,\qquad
 G+(1+C_p\tau)\sigma n^{1/p}e_I,
 \quad I\sim\mathrm{Unif}([n]).
\end{equation}
For the failure probability $\delta$, the sparse-spike pair of~\cite[Appendix~D.2.2]{GangulyW18}, which gives the bound $\Omega(n^{1-2/p}\epsilon^{-2/p}\log n\log^{2/p}(1/\delta)\log W)$, is
\begin{equation}\label{eq:large-p-many-spikes}
 G,\qquad G+\lambda\sum_{i\in T}e_i,
 \quad |T|=t\coloneqq\max\{1,\lceil\log_3(1/\sqrt\delta)\rceil\},\quad
 \lambda=C_p\tau^{1/p}\sigma E_{n-t}/t^{1/p},
\end{equation}
where $T$ is uniform among $t$-subsets of $[n]$, $E_j\coloneqq\mathbb E_{g\sim\calN(0,\Id_j)}\|g\|_p$, and $C_p$ is a sufficiently large constant depending only on $p$.

\begin{theorem}
\label{cor:turnstile-lp-large}
Assume $p>2$ is fixed, $0<\epsilon<c_p$, $2^{-n^{c_p}}<\delta\le1/3$, and $W\ge(n/\delta)^{C_0}$, where $c_p>0$ is a sufficiently small and $C_0$ a sufficiently large constant depending only on $p$. Then $(1+\epsilon)$-approximate $L_p$ estimation with failure probability $\delta$ requires
\[
\Omega\Bigl(\min\{n^{1-2/p}\epsilon^{-2/p}\log n\log^{2/p}(1/\delta) + n^{1- 2/p} \epsilon^{-2} \log(1 / \delta),n\}\cdot \log W\Bigr)
\]
bits of space.
\end{theorem}

\begin{proof}
We first assume that $n\ge n_p$, where $n_p$ is a sufficiently large constant depending only on $p$, and treat $n<n_p$ at the end. By amplification, we may assume that $\delta$ is at most a sufficiently small constant. For the single-spike pair \eqref{eq:large-p-single-spike}, the range of~\cite[Appendix~D.1]{GangulyW18} is $\log(1/\delta)\ll\tau^2n^{2/p}$, so we take
\[
 \tau=\max\{\epsilon,A\sqrt{\log(1/\delta)}/n^{1/p}\}
\]
for a sufficiently large constant $A$. Then $\tau\ll1$, since we have $\epsilon<c_p$, and $\delta>2^{-n^{c_p}}$ gives $A\sqrt{\log(1/\delta)}/n^{1/p}\le An^{c_p/2-1/p}$, which is small since $c_p<2/p$ and $n\ge n_p$. 

For the sparse-spike pair \eqref{eq:large-p-many-spikes}, take $\tau=\Theta_p(\epsilon)$ with $\tau\ge\epsilon$. For this pair, we may assume that the first term in the claimed bound dominates the second and that $n^{1-2/p}\epsilon^{-2}\log(1/\delta)\ll n$, since otherwise the single-spike pair alone gives the theorem. Together with $\delta>2^{-n^{c_p}}$, these assumptions place the sparse-spike pair in the parameter range of~\cite[Appendix~D.2]{GangulyW18}. Indeed, the first assumption gives $\epsilon^{2-2/p}\ge1/\log n$, so the sparse-spike bound $n^{1-2/p}\tau^{-2/p}\log n\log^{2/p}(1/\delta)$ is at most $n^{1-2/p+2c_p/p}\log^2n\ll n$, and the remaining lower bound on $\delta$ there follows from $\delta>2^{-n^{c_p}}$ for $c_p$ small enough.

For item~(a), in these parameter ranges, the computations in~\cite[Appendix~D.1.2 and~D.2.2]{GangulyW18} (see also~\cite[Lemma~5.2.3]{gribelyuk2025lifting}) show that both pairs admit thresholds $a<b$ with $b>(1+\tau)^2a$ such that, except with probability $\ll\delta$, $\|G+V\|_p\le a-\Omega_p(\tau\Lambda n^{1/p})$ given $B=0$ and $\|G+V\|_p\ge b+\Omega_p(\tau\Lambda n^{1/p})$ given $B=1$. The function $q=\|\cdot\|_p$ is $L$-Lipschitz with $L=1$, since $p>2$. For both pairs $\log(1/\delta)\ll\tau^2n^{2/p}$, by the choice of $\tau$ for the single-spike pair and by $n^{1-2/p}\epsilon^{-2}\log(1/\delta)\ll n$ and $\tau\ge\epsilon$ for the sparse-spike pair. This gives $\tau n^{1/p}\gg1$, so these margins are at least $2L\Lambda/n$, as required in the paragraph after \cref{lem:norm-rounding-noise}. Since $\tau\ge\epsilon$ and $b>(1+\tau)^2a$, every two-sided $(1+\epsilon)$-approximation of $\|y\|_p$ determines $f(y)$. This gives item~(a).

For item~(b), determining $B$ from $G+V$ means distinguishing the two members of the pair. By~\cite[Appendix~D.1 and~D.2]{GangulyW18}, doing so with error at most $\delta$ requires dimension
\[
 \Omega\bigl(n^{1-2/p}\tau^{-2}\log(1/\delta)\bigr)
   =\Omega\!\left(\min\{n^{1-2/p}\epsilon^{-2}\log(1/\delta),n\}\right)
\]
for the single-spike pair, and $\Omega(n^{1-2/p}\epsilon^{-2/p}\log n\log^{2/p}(1/\delta))$ for the sparse-spike pair. Therefore, steps~\ref{red:transfer}--\ref{red:space}, applied to each pair, give two space bounds, and the larger of them gives the theorem.

For $n<n_p$, the claimed bound is $O_p(\log W)$, and the argument for $n<c_0^{-3}$ at the end of the proof of \cref{cor:turnstile-lp-small} gives $S=\Omega(\log W)$, since $\|y\|_p=|y_1|$ for $y$ supported on the first coordinate.
\end{proof}

\subsubsection{Operator norm and Ky Fan norm}

In each family below, $B$ is uniform in $\{0,1\}$, the entries of $\bG$ are
independent $\calN(0,\sigma^2)$, and the offset is $V=0$ given $B=0$ and a
random low-rank matrix independent of $\bG$ given $B=1$, depending on the
family. These are the pairs of~\cite[Section~5.3]{gribelyuk2025lifting},
scaled by $\sigma$, and $\beta\gg1$ is a constant. In each family, we specify
$q$ and thresholds $a<b$, and $C>0$ denotes an absolute constant.

For accurate operator-norm approximation, take $n\times d$ matrices with
$d=\lfloor n\epsilon^2\rfloor$ and $n\epsilon^2\gg1$. Put
\[
V=\frac{\beta\sigma}{\sqrt{\epsilon n}}uv^\top\text{ given }B=1,
\qquad u\sim\calN(0,\Id_n),\quad v\sim\calN(0,\Id_d)\text{ independent},
\]
and take $q=\|\cdot\|_{op}$ with thresholds $a=(1+3\epsilon)C\sigma\sqrt n$ and
$b=(1+9\epsilon)C\sigma\sqrt n$.

For large-factor operator-norm approximation with $1+c_0\le\alpha\le c\sqrt n$,
take $n\times n$ matrices,
\[
V=\frac{\beta\alpha\sigma}{\sqrt n}uv^\top\text{ given }B=1,
\qquad u,v\sim\calN(0,\Id_n)\text{ independent},
\]
and take $q=\|\cdot\|_{op}$ with thresholds $a=7C\sigma\sqrt n/2$ and
$b=15\alpha C\sigma\sqrt n/4$.

For the Ky Fan norm with integer $1\le s\le c\sqrt n$, take $n\times n$
matrices,
\[
V=\frac{\beta\sigma}{\sqrt n}\sum_{j=1}^s u_jv_j^\top\text{ given }B=1,
\qquad u_j,v_j\sim\calN(0,\Id_n)\text{ independent},
\]
and take $q=\|\cdot\|_{F_s}$ with thresholds $a=2Cs\sigma\sqrt n$ and
$b=4Cs\sigma\sqrt n$.

\begin{theorem}
\label{cor:turnstile-op-kyfan}
Assume $W\ge n^{C_0}$ for a sufficiently large absolute constant $C_0$.
For $\epsilon\in(0,1/3)$ with $n\epsilon^2$ at least a sufficiently
large absolute constant, $(1+\epsilon)$-approximate operator-norm estimation
on $n\times\lfloor n\epsilon^2\rfloor$ matrices requires
$\Omega((n^2\epsilon^2)\log W)$ bits.
On $n\times n$ matrices, an $\alpha$-approximate estimate in
$[\|\bM\|_{op},\alpha\|\bM\|_{op}]$ requires
$\Omega((n^2/\alpha^4)\log W)$ bits for
$1+c_0\le\alpha\le c\sqrt n$, and a $(1+c_0)$-approximation to the Ky Fan
$s$-norm requires $\Omega((n^2/s^2)\log W)$ bits for integer
$1\le s\le c\sqrt n$. Here $c,c_0>0$ are sufficiently small absolute constants.
\end{theorem}

\begin{proof}
For item~(a), the proofs of~\cite[Lemma~5.3.8, Corollary~5.3.7, and
Corollary~5.3.9]{gribelyuk2025lifting} show that, in the three families
respectively, $q(\bG+V)$ is at most $(1+2\epsilon)C\sigma\sqrt n$,
$3C\sigma\sqrt n$, and $Cs\sigma\sqrt n$ given $B=0$, and at least
$(1+10\epsilon)C\sigma\sqrt n$, $4\alpha C\sigma\sqrt n$, and
$(0.9\beta-C)s\sigma\sqrt n$ given $B=1$, except with probability
$e^{-\Omega(d)}$ in the first family and $e^{-\Omega(n)}$ in the other two.
In the first family, the proof of~\cite[Lemma~5.3.8]{gribelyuk2025lifting}
gives the factor $1+4\epsilon$ for its value of $\beta$, and taking $\beta$
larger gives $1+10\epsilon$.
These probabilities are $\ll1$, since $d=\lfloor n\epsilon^2\rfloor\gg1$ and
$n\ge1/c^2\gg1$. The function $q$ is $L$-Lipschitz with $L=1$ for the
operator norm and $L=\sqrt s$ for the Ky Fan $s$-norm. Since $\beta\ge20C/3$,
the thresholds leave margins of at least $\epsilon C\sigma\sqrt n$,
$C\sigma\sqrt n/4$, and $Cs\sigma\sqrt n$, which are at least $2L\Lambda/n$
since $\sigma=\Lambda/\sqrt{2\pi}$ and $n\epsilon^2,n\gg1$, as required in the
paragraph after \cref{lem:norm-rounding-noise} in ambient dimension $N=nd$ or
$n^2$. Since $b>(1+\epsilon)^2a$ for $\epsilon<1/3$, $b=15\alpha a/14>\alpha a$,
and $b=2a>(1+c_0)^2a$ in the three families, every valid estimate of $q(y)$
determines $f(y)$. This gives item~(a).

For item~(b), determining $B$ from $\bG+V$ means distinguishing $\bG$ from
$\bG+V$ with $V$ as given $B=1$. By \cite[Theorem~4]{LiW16} (see
also~\cite[Lemma~5.3.6]{gribelyuk2025lifting}), a real sketch distinguishing
$\bG$ from $\bG+\sigma\sum_jw_ju_jv_j^\top$ with success probability
$1/2+\Omega(1)$ has dimension $\Omega(\|w\|_2^{-4})$. Substituting the three choices of $V$
gives the dimension bounds
\begin{equation}\label{eq:operator-dimension-bounds}
\Omega(n^2\epsilon^2),\qquad \Omega(n^2/\alpha^4),\qquad \Omega(n^2/s^2).
\end{equation}
Therefore, steps~\ref{red:transfer}--\ref{red:space} give the result.
\end{proof}

\subsubsection{Eigenvalue approximation and PSD testing}

Here the two tasks are as follows. In \emph{eigenvalue approximation}, the algorithm must output an additive $\epsilon\|\bM\|_F$ approximation to the largest eigenvalue of a symmetric input matrix $\bM$. In \emph{PSD testing}, fix $p\geq 1$. The algorithm must distinguish between the YES case $\bM\succeq 0$ and the NO case $\lambda_{\min}(\bM)<0$ and $\lambda_{\min}(\bM)\le -\epsilon\|\bM\|_p$, where $\|\cdot\|_p$ denotes the Schatten-$p$ norm.

We use the hard inputs of~\cite[Theorems~5.4.10 and~5.4.11]{gribelyuk2025lifting}, scaled by $\sigma$. For $d\le n/2$ and $s\in(0,1)$, both chosen below for each problem, let $B$ be uniform in $\{0,1\}$, let $\bG\in\mathbb R^{d\times d}$ have independent $\calN(0,\sigma^2)$ entries, and put $V=0$ given $B=0$ and $V=\sigma suv^\top$ given $B=1$, where $u,v\sim\calN(0,\Id_d)$ are independent of each other and of $\bG$. Write $\bH\coloneqq\bG+\sigma suv^\top$, so that the hard input $\bG+V$ is $\bG$ given $B=0$ and $\bH$ given $B=1$. For $d\gg1$, the proofs there show that, except with probability $\ll1$,
\begin{equation}
\|\bG\|_{op}\le C_1\sigma\sqrt d,
\quad
\|\bH\|_{op}\ge C_2\sigma sd-C_1\sigma\sqrt d,
\quad
\|\bG\|_F,\ \|\bH\|_F\le C_1\sigma d,
\quad
\|u\|_2\|v\|_2\le C_1d
\label{eq:eigen-hard-good}
\end{equation}
for absolute constants $C_1,C_2>0$. As in~\cite[Theorem~5.4.11]{gribelyuk2025lifting}, both problems use the symmetric matrix
\[
\mathcal B(Y)=\begin{bmatrix}0&Y\\Y^\top&0\end{bmatrix}.
\]
The eigenvalues of the $2d\times2d$ matrix $\mathcal B(Y)$ are $\pm\sigma_i(Y)$, so that
\[
\lambda_{\max}(\mathcal B(Y))=\|Y\|_{op},
\quad \|\mathcal B(Y)\|_F=\sqrt2\|Y\|_F.
\]
We henceforth pad $\mathcal B(Y)$ with zero rows and columns to dimension $n$ and still denote it by $\mathcal B(Y)$. This preserves its largest eigenvalue and its norms.

For eigenvalue approximation, choose $d=\lfloor n/2\rfloor$ and $s=A\epsilon$ for a constant $A\gg1$, as in~\cite[Theorem~5.4.10]{gribelyuk2025lifting}. Since $n=\Omega(\epsilon^{-2})$ and $\epsilon\le c_0$ with $c_0\ll1$, we have $d\gg1$ and $\epsilon\sqrt d\gg1$. We use the promise problem
\[
f_{\mathrm{eig}}(Y)=
\begin{cases}
0,&\|Y\|_{op}\le 2C_1\sigma\sqrt d\text{ and }\|Y\|_F\le 2C_1\sigma d,\\
1,&\|Y\|_{op}\ge C_2\sigma sd/2\text{ and }\|Y\|_F\le 2C_1\sigma d,\\
*,&\text{otherwise}.
\end{cases}
\]
The two operator-norm thresholds differ by more than $4\sqrt2\epsilon C_1\sigma d$, and the events in \eqref{eq:eigen-hard-good} give $f_{\mathrm{eig}}(\bG)=0$ and $f_{\mathrm{eig}}(\bH)=1$ with margin $\Omega(\sigma\sqrt d)$ in each inequality. An additive eigenvalue estimate of $\mathcal B(Y)$, compared with the midpoint of these thresholds, solves this promise problem.

For PSD testing, as in~\cite[Theorem~5.4.11]{gribelyuk2025lifting}, fix $C>2C_1$ and put
\[
\mathcal S(Y)=\mathcal B(Y)+h\Id_{2d},
\qquad h=\lceil C\sigma\sqrt d\rceil,
\]
where $\Id_{2d}$ is padded with zeros like $\mathcal B(Y)$. The top-left $2d\times2d$ block of $\mathcal S(Y)$ has minimum eigenvalue $h-\|Y\|_{op}$ and all other entries vanish, so $\mathcal S(Y)$ has the Schatten norms of this block and, when $n>2d$, minimum eigenvalue $\min\{0,h-\|Y\|_{op}\}$. We use the promise problem
\[
f_{\mathrm{psd}}(Y)=
\begin{cases}
0,&\mathcal S(Y)\succeq0,\\
1,&\lambda_{\min}(\mathcal S(Y))<0\text{ and }\lambda_{\min}(\mathcal S(Y))\le-\epsilon\|\mathcal S(Y)\|_p,\\
*,&\text{otherwise}.
\end{cases}
\]
On the same events, the top-left $2d\times2d$ block of $\mathcal S(\bG)$ has minimum eigenvalue $h-\|\bG\|_{op}=\Omega(\sigma\sqrt d)$. If $s\sqrt d\gg1$, $\bH$ satisfies
\[
-\lambda_{\min}(\mathcal S(\bH))\ge c_1\sigma sd,
\qquad
\|\mathcal S(\bH)\|_p\le C_3\sigma\bigl(d^{1/2+1/p}+sd\bigr).
\]
The second bound follows from $\|\mathcal S(\bH)\|_p\le(2d)^{1/p}(h+\|\bG\|_{op})+2^{1/p}\sigma s\|u\|_2\|v\|_2$, as $\mathcal B(uv^\top)$ has the two nonzero eigenvalues $\pm\|u\|_2\|v\|_2$. For $p=\infty$ interpret $1/p=0$. For a constant $A\gg1$ and $\epsilon\le c_0$ with $c_0\ll1$, the following choices of $(d,s)$, as in~\cite[Theorem~5.4.11]{gribelyuk2025lifting} but with $d\le n/2$, make $-\lambda_{\min}(\mathcal S(\bH))\ge2\epsilon\|\mathcal S(\bH)\|_p$, so that $\bH$ satisfies the NO condition of $f_{\mathrm{psd}}$ with margin $c_1\sigma sd/2$:
\[
\begin{array}{c|c|c}
 &d&s\\\hline
1\le p\le2&\left\lfloor\min\{n/2,\epsilon^{-p}\}\right\rfloor&A/\sqrt d\\
2<p<\infty&\lfloor n/2\rfloor&A\max\{d^{-1/2},\epsilon d^{1/p-1/2}\}\\
p=\infty&\lfloor n/2\rfloor&A/\sqrt d.
\end{array}
\]
Taking $n\gg1$ and $c_0\ll1$ makes $d\gg1$ in each case.

\begin{theorem}
\label{cor:turnstile-eigen-psd}
Assume $0<\epsilon\le c_0$ and $W\ge n^{C_0}$ for sufficiently small $c_0>0$
and sufficiently large $C_0>0$. On symmetric $n\times n$ matrices,
approximating the largest eigenvalue to additive error $\epsilon\|\bM\|_F$
requires $\Omega(\epsilon^{-4}\log W)$ bits when $n=\Omega(\epsilon^{-2})$.
PSD testing under the promise above requires
\[
\Omega\bigl(\min\{\epsilon^{-2p},n^2 \}\log W\bigr),
\qquad
\Omega\bigl(\min\{\epsilon^{-4}n^{2-4/p}, n^2\}\log W\bigr),
\qquad
\Omega\bigl(n^2\log W\bigr)
\]
bits for $p\in[1,2]$, $p\in(2,\infty)$, and $p=\infty$, respectively.
\end{theorem}

\begin{proof}
For PSD testing, we first assume that $n\ge c_0^{-3}$, and treat $n<c_0^{-3}$ at the end. For eigenvalue approximation, $n$ is large since $n=\Omega(\epsilon^{-2})$ and $\epsilon\le c_0$. A streaming algorithm for the symmetric problem gives one for the corresponding promise on $Y$. To obtain it, replace every update to $Y_{ij}$ by the two updates in $\mathcal B(Y)$, and initialize the diagonal to $h$ in the PSD case. Since the mollified streams have length at most $R^{7/2}$, the stream with the doubled updates and the $2dh$ diagonal updates still has length at most $W$.

For item~(a), let $\bZ$ have independent $\gamma_R$ entries, and use the paragraph after \cref{lem:norm-rounding-noise} in ambient dimension $N=d^2$. The conditions defining $f_{\mathrm{eig}}$ and $f_{\mathrm{psd}}$ are thresholds on $\|Y\|_{op}$, $\|Y\|_F$, and $-\lambda_{\min}(\mathcal S(Y))-\epsilon\|\mathcal S(Y)\|_p$, since $\mathcal S(Y)\succeq0$ if and only if $\|Y\|_{op}\le h$, and $\lambda_{\min}(\mathcal S(Y))$ equals $h-\|Y\|_{op}$ if $n=2d$ and $\min\{0,h-\|Y\|_{op}\}$ otherwise. As functions on $\mathbb R^{d\times d}$ with the Frobenius norm, the first two are $L$-Lipschitz with $L=1$, and the third with $L=1+\epsilon(2d)^{1/p}\le2n$, because
\[
\|\mathcal S(Y+\bZ)-\mathcal S(Y)\|_p=\|\mathcal B(\bZ)\|_p
\le(2d)^{1/p}\|\bZ\|_{op}\le(2d)^{1/p}\|\bZ\|_F.
\]
Thus $2L\Lambda/n\le4\Lambda=O(\sigma)$ for each of these functions. On the events in \eqref{eq:eigen-hard-good}, the margins established above are $\Omega(\sigma\sqrt d)$ for $f_{\mathrm{eig}}$ and for the YES case of $f_{\mathrm{psd}}$, and $c_1\sigma sd/2$ for the NO case of $f_{\mathrm{psd}}$, which exceed $4\Lambda$ since $\sqrt d,sd\gg1$. Hence $f(Y)=f(Y+\bZ)=B$ for $f\in\{f_{\mathrm{eig}},f_{\mathrm{psd}}\}$, except with probability $\ll1$. Moreover, a valid answer for $Y$ determines $f(Y)$: for $f_{\mathrm{eig}}$ by comparing the eigenvalue estimate with the midpoint of the two thresholds, and for $f_{\mathrm{psd}}$ directly. This gives item~(a).

For item~(b), the ambient dimension is $N=d^2$, and determining $B$ from $\bG+V$ means distinguishing $\bG$ from $\bH=\bG+\sigma suv^\top$. After scaling by $1/\sigma$, this distinguishes $G_0$ from $G_0+suv^\top$, where $G_0\coloneqq\bG/\sigma$ has independent standard Gaussian entries. By~\cite[Theorem~4]{LiW16} (see also~\cite[Lemma~5.3.6]{gribelyuk2025lifting}), a real sketch doing so with success probability $1/2+\Omega(1)$ has dimension
$
\Omega(s^{-4}).
$
For eigenvalues this is $\Omega(\epsilon^{-4})$. For PSD testing the three choices of $(d,s)$ above give, respectively,
\[
\Omega\bigl(\min\{\epsilon^{-2p},n^2\}\bigr),\quad
\Omega\bigl(\min\{\epsilon^{-4}n^{2-4/p},n^2\}\bigr),\quad
\Omega(n^2).
\]
Therefore, steps~\ref{red:transfer}--\ref{red:space} give the result.

For PSD testing with $n<c_0^{-3}$, the claimed bound is $O(\log W)$, so it suffices to show $S=\Omega(\log W)$, and we may take $\delta$ to be a small absolute constant by amplification. We use inputs supported on the entry $(1,1)$, on which $\|\bM\|_p=|\bM_{11}|$, so that, as $\epsilon\le1$, the promise reads $\bM_{11}\ge0$ versus $\bM_{11}<0$. Let $G+V$ have law $\calN(A\sigma,\sigma^2)$ given $B=0$ and $\calN(-A\sigma,\sigma^2)$ given $B=1$, with $A$ as above. Thresholding $q(\bM)=-\bM_{11}$ at $0$ and $1$ distinguishes these two laws with the margins required in the paragraph after \cref{lem:norm-rounding-noise}, which gives item~(a), and item~(b) holds with dimension bound $1$, since a real sketch of dimension $0$ is constant. Steps~\ref{red:transfer}--\ref{red:space} give $S=\Omega(\log W)$.
\end{proof}

\subsubsection{Compressed sensing}

In $\ell_2/\ell_2$ sparse recovery, the algorithm must output, for the final frequency vector $y$, an $s$-sparse $x'$ with
\[
\|y-x'\|_2\le(1+\epsilon)\mathsf{OPT}_k(y),
\qquad
\mathsf{OPT}_k(y)\coloneqq\min_{k\text{-sparse }\widetilde x}\|y-\widetilde x\|_2.
\]
As in~\cite[Section~5.5]{gribelyuk2025lifting}, the output sparsity is $s\le C\epsilon n/\log n$ for a constant $C$. For $t>0$, let $\mathcal O_y(t)$ be the set of $s$-sparse $x'$ with $\|y-x'\|_2\le(1+t)\mathsf{OPT}_k(y)$.

We use the hard family of Price and Woodruff~\cite[Section~4]{PriceW11}, scaled by $\mu$. Let $\mathcal F$ be their family of supports $J\subseteq[n]$ with $|J|=k$, pairwise symmetric difference at least $k$, $\log|\mathcal F|=\Omega(k\log(n/k))$, and $\Pr_{J\sim\mathrm{Unif}(\mathcal F)}[i\in J]\le2k/n$ for every $i\in[n]$. 

Put $\mu\coloneqq\lceil A\Lambda\sqrt{\epsilon n/k}\rceil$, where $A$ is a sufficiently large constant depending on $C$. Let $x_0\in\{0,\pm\mu\}^n$ have support $B$ uniform in $\mathcal F$ and independent uniform signs on $B$. The label is $B$, and the hard input is $G+V$ with $V=x_0$ and $G\sim\calN(0,\sigma^2\Id_n)$, which is the hard input of~\cite[Section~4]{PriceW11} with noise level $\sigma^2n/(\mu^2k)=\Theta(1/\epsilon)$. Here we do not amplify in step~\ref{red:transfer}, since a constant success probability suffices for the lower bound of~\cite{PriceW11}, and we track the success probabilities explicitly.

\begin{theorem}
\label{cor:turnstile-compressed-sensing}
Assume $\sqrt{k\log n/n} \leq \epsilon \leq c_0$ and $W\ge n^{C_0}$, where $c_0>0$ is a sufficiently small and $C_0$ a sufficiently large absolute constant. Then $(1+\epsilon)$-approximate $\ell_2/\ell_2$ sparse recovery with output sparsity $O(\epsilon n/\log n)$ and failure probability $1/3$ requires
$
\Omega\Bigl((k/\epsilon)\log(n/k)\cdot \log W\Bigr)
$
bits of space.
\end{theorem}

\begin{proof}
For item~(a), we first show smoothness. Since $x_0$ is integral, $Y=X+[V]=x_0+X$. We use that, except with probability $n^{-100}$,
\[
\Lambda^2n/200\le\|X\|_2^2=O(\Lambda^2n),
\qquad
\|X\|_\infty^2=O(\Lambda^2\log n).
\]
The two upper bounds follow from discrete-Gaussian concentration (\cref{lem:gaussian-moment-general}), except with probability $n^{-100}/2$. For the lower bound, each coordinate satisfies $\Pr[|X_i|\le\Lambda/10]\le(\Lambda/5+1)/\rho_\Lambda(\mathbb Z)\le1/4$, as $\rho_\Lambda(\mathbb Z)\ge\Lambda$ by Poisson summation. Hoeffding's inequality then gives $\|X\|_2^2\ge\Lambda^2n/200$ except with probability $e^{-n/8}\le n^{-100}/2$.

On this event, $\mathsf{OPT}_k(Y)^2\ge\|X\|_2^2-2k\|X\|_\infty^2=\Omega(\Lambda^2n)$, since $k\log n\le\epsilon^2n$. As $\epsilon\ge n^{-1/2}$, this gives $3\Lambda/n\le\epsilon\mathsf{OPT}_k(Y)$. Now let $Z\sim\gamma_R$. By \cref{lem:norm-rounding-noise} in ambient dimension $N=n$, outside an event of probability $2/R$, the $1$-Lipschitz functions $\mathsf{OPT}_k$ and $y\mapsto\|y-x'\|_2$, for all $x'$, change by at most $\Lambda/n$ from $Y$ to $Y+Z$. Hence every $x'\in \mathcal O_{Y+Z}(\epsilon)$ satisfies
\[
\|Y-x'\|_2\le(1+\epsilon)\mathsf{OPT}_k(Y+Z)+\Lambda/n\le(1+\epsilon)\mathsf{OPT}_k(Y)+3\Lambda/n\le(1+2\epsilon)\mathsf{OPT}_k(Y),
\]
that is, $\mathcal O_{Y+Z}(\epsilon)\subseteq \mathcal O_Y(2\epsilon)$. Thus the relation with output sets $\mathcal O_y(\epsilon)$ is $(n^{-100}+4/R)$-smooth on average into the relation with output sets $\mathcal O_y(2\epsilon)$, where $2/R$ accounts for the truncation in step~\ref{red:transfer}.

We next show that a valid answer for $Y$ determines $B$, as in~\cite[Lemma~4.3]{PriceW11} and~\cite[Section~5.5]{gribelyuk2025lifting}. On the event above, let $x'\in \mathcal O_Y(2\epsilon)$ and $T=B\cup\operatorname{supp}(x')$, so $|T|\le k+s$. Since $\mathsf{OPT}_k(Y)\le\|Y-x_0\|_2=\|X\|_2$ and $(1+2\epsilon)^2\le1+5\epsilon$, we have $\|Y-x'\|_2^2\le(1+5\epsilon)\|X\|_2^2$. Splitting both sides on $T$ and its complement gives $\|x_0-x'+X_T\|_2^2\le5\epsilon\|X\|_2^2+\|X_T\|_2^2$. Therefore
\[
\|x'-x_0\|_2^2\le10\epsilon\|X\|_2^2+4\|X_T\|_2^2
=O\bigl((1+C)\mu^2k/A^2\bigr)<\mu^2k/16,
\]
using $k\log n\le\epsilon n$, $s\log n\le C\epsilon n$, and the choice of $A$. By the separation of $\mathcal F$ (any $x\in\{0,\pm\mu\}^n$ with $\operatorname{supp}(x)\in\mathcal F\setminus\{B\}$ differs from $x_0$ by $\mu$ on the at least $k$ coordinates of $\operatorname{supp}(x)\triangle B$, so $\|x-x_0\|_2\ge\mu\sqrt k$ and hence $\|x'-x\|_2\ge3\mu\sqrt k/4>\|x'-x_0\|_2$), the nearest vector to $x'$ in $\{x\in\{0,\pm\mu\}^n:\operatorname{supp}(x)\in\mathcal F\}$ has support $B$. This gives item~(a).

For item~(b), determining $B$ from $G+V=x_0+G$ is the recovery problem of~\cite[Section~4]{PriceW11}. By~\cite[Lemma~4.1]{PriceW11} and Fano's inequality, as in~\cite[Section~4]{PriceW11}, doing so with probability at least $0.3$ requires dimension $\Omega((k/\epsilon)\log(n/k))$.

It remains to track the success probabilities in steps~\ref{red:transfer}--\ref{red:space}, in ambient dimension $N=n$. The claimed dimension bound $(k/\epsilon)\log(n/k)$ is at most $\epsilon n=O(N)$, as required in step~\ref{red:transfer}. In step~\ref{red:transfer}, the proof of case~\textup{(2)} of \cref{thm:randomized-mollified-transfer} bounds the failure probability of its decoder by $2/3+n^{-100}+O(R^{-1/10})<0.67$. This gives an integer sketch of dimension $r=O(S/\log R)$ whose decoder returns an element of $\mathcal O_Y(2\epsilon)$ with probability greater than $0.33$, and hence, by item~(a), recovers $B$ with probability greater than $0.32$. Step~\ref{red:lift} then gives a real sketch of dimension $O(r)$ that recovers $B$ from $\bA(x_0+G)$ with probability at least $0.3$, and step~\ref{red:round} is not needed since $x_0$ is integral. Therefore, steps~\ref{red:transfer}, \ref{red:lift}, and~\ref{red:space} give the result.
\end{proof}

\subsection{Applications via the SMP model}
\label{app:applications-smp}

We prove the bounds summarized in \cref{thm:smp-applications-full}.

Here we use the exact reduction of \cref{cor:yao-minmax-exact-transfer}
rather than the mollified one, because these problems are not smooth. Even
a small independent perturbation can change zero coordinates, edges, or
copies of a fixed subgraph.

\paragraph{Reduction to SMP protocols.}
Consider a randomized streaming algorithm with frequency vectors in $\mathbb Z^N$ that
uses $S$ bits of space on streams of length at most $W$ with success probability at least $2 / 3$. First, amplify its
success probability to $0.999$, at the cost of a constant factor in space.
Then case~\textup{(2)} of \cref{cor:yao-minmax-exact-transfer} with
$\delta=0.001$, applied with the parameters $R,D$ and the set $\mathcal X$ of
inputs specified below, gives a public-coin linear sketch
$\mathcal A$ with image size $2^{O(S)}$ that succeeds with probability at
least $0.997$ on every input in $\mathcal X$. For strict streams, we use the
first paragraph of \cref{rmk:strict-robp-extension}.

We then turn $\mathcal A$ into a protocol in the SMP model
of~\cite[Section~2.3]{AssadiKL17}. In this model, $k$ players hold inputs
$x_1,\ldots,x_k$ and each sends one message to a referee, who may also hold
an input $x_0$~\cite[Remark~2.5]{AssadiKL17}. The referee must output an
answer $z$ such that $(x_0,x_1,\ldots,x_k,z)$ lies in a given relation. We
use this model with $x_0,\ldots,x_k\in\mathbb Z^N$ and relations of the form
$\{(x_0,\ldots,x_k,z):(x_0+\cdots+x_k,z)\in\mathcal R\}$, where $\mathcal R$
is the relation of the streaming problem. We allow
public randomness shared by the players and the referee.

Given a randomized sketching algorithm $\mathcal A$ with image size $2^{O(S)}$, we let players publicly sample a seed $\rho$ and we let player $i$ send $\mathcal A_\rho(x_i)$, which takes $O(S)$ bits, and the referee
adds the messages and the sketch of its own input and decodes. Since
$\mathcal A_\rho$ is a group homomorphism, the referee obtains the sketch of the
sum of all inputs, so after applying the decoder $g_\rho$, the protocol succeeds with probability at least $0.997$
whenever this sum lies in $\mathcal X$.
Hence every lower bound on the message length of such protocols is, up to a constant
factor, a lower bound on $S$.

Throughout this subsection,
$W\ge(2N)^{C_*}$ for a sufficiently large absolute constant $C_*$, and
\[
R\coloneqq\lfloor W^{1/4}\rfloor,\qquad D\coloneqq R^{1/3}.
\]
Then $\log D=\Theta(\log W)$, and taking $C_*$ large makes
$R\ge\max\{N^{8},0.001^{-3}\}$, $R\gg1$, and $D\ge N^2$.
Each application below specifies a relation and a set
$\mathcal X\subseteq B_D\cap\mathbb Z^N$ of inputs, with
$\mathcal X\subseteq\mathbb Z_{\ge0}^N$ for strict streams. Since
$R^4\le W$, the amplified algorithm succeeds on every stream of length at most
$R^4$ whose final frequency vector lies in $\mathcal X$, so case~\textup{(2)}
of \cref{cor:yao-minmax-exact-transfer} with $\delta=0.001$ applies.

\subsubsection{\texorpdfstring{$L_0$}{L0} approximation}

In $L_0$ estimation, the algorithm must output, for the final frequency
vector $x$, an estimate $z$ of the number $\|x\|_0$ of nonzero coordinates
of $x$ with $(x,z)$ in the relation
\[
\mathcal R_{L_0,\epsilon}
=\{(x,z):(1-\epsilon)\|x\|_0\le z\le(1+\epsilon)\|x\|_0\}.
\]
Du, Mitzenmacher, Woodruff, and Yang~\cite[Theorem~7.3]{DuMWY19} proved the
lower bound $\Omega(\epsilon^{-2}\log n\,\log\log W)$ for
$n^{-0.49}\le\epsilon<1/3$, through the equivalence of~\cite{AHLW16} and the
simultaneous communication complexity of a sum-equality problem. We instead
prove the lower bound directly for \emph{modular sketches} (linear
sketches whose image is a finite Abelian group). This recovers their bound and extends it to all
$0<\epsilon\le c_0$, with $\log n$ replaced by $\log(1+n\epsilon^2)$.

\begin{theorem}[$L_0$ lower bound for modular sketches]
\label{thm:modular-l0}
Assume $0<\epsilon\le1/128$ and $T\ge Cn^4$, where $C$ is a sufficiently
large absolute constant. Then a randomized linear sketch with image size at
most $2^S$ that outputs a $(1+\epsilon)$-approximation to $\|x\|_0$ with
probability at least $2/3$ for every $x\in\{0,1,\ldots,T\}^n$ requires
\[
S=\Omega\!\left(\epsilon^{-2}\log(1+n\epsilon^2)\,\log\log T\right).
\]
\end{theorem}

The proof is a direct-sum argument over $t=\Theta(\log(1+n\epsilon^2))$
independent blocks of about $\epsilon^{-2}$ coordinates each. In a block,
each coordinate holds either \textbf{noise} uniform in
$\{1,\ldots,T\}$ or a \textbf{signal} equal to a large integer $M$ times a coin of bias $\frac12\pm8\epsilon$, and the task is to find the sign of the
bias. This requires $\Omega(\epsilon^{-2})$ signals. Let $\calA$ be a linear sketch with image size $2^S$, so
$\calA(x)=\sum_{j=1}^nx_j\calA(e_j)$. A signal at coordinate $j$ is invisible if
$\calA(e_j)$ has a short relation $a\calA(e_j)=\sum_dc_d\calA(e_d)$ with
earlier noise coordinates $d$, with integers $1\le a\le P$ and
$|c_d|\le P$. Indeed, as $a$ divides $M=P!$, the signal term
$M\calA(e_j)=\sum_d(Mc_d/a)\calA(e_d)$ then only slightly shifts the noise:
we choose $P$ with $PM=T^{o(1)}$, so the noise at each coordinate $d$ is shifted by
$|Mc_d/a|\le PM$, which is negligible compared with its range $T$.
Since noise and signals are assigned by fair coins, the visible signals are,
in expectation, as many as the noise coordinates without short relations to
earlier ones, and a group of size $2^S$ admits at most $S/\log(P+1)$ such
noise coordinates, where $\log(P+1)=\Theta(\log\log T)$.

For the direct sum, a block $I$ is chosen uniformly at random among $t$
blocks, and we let the earlier blocks hold only noise and the later blocks be empty.
Each coordinate of block $i$ is copied $4^i$ times, so that block $I$
outweighs all earlier blocks together. Since the earlier blocks contribute a
fixed amount to the $L_0$ norm, a $(1+\epsilon)$-approximation of the $L_0$
norm determines the sign of block $I$. The visible signals of all $t$
blocks are counted against the same $S/\log(P+1)$ noise coordinates, which do
not depend on $I$, so block $I$ has on average at most $S/(t\log(P+1))$
visible signals, while finding its sign requires $\Omega(\epsilon^{-2})$ of them
(\cref{fact:l0-coin}). This gives
$t\cdot\Omega(\epsilon^{-2})\le S/\log(P+1)$.

\begin{fact}[$P$-independent lists]
\label{fact:l0-independent}
Let $P\ge1$ be an integer and $M\coloneqq P!$. A list $(u_1,\ldots,u_r)$ in
an Abelian group $G$ is \emph{$P$-independent} if $\sum_kc_ku_k=0$ with
integers $|c_k|\le P$ forces all $c_k=0$. Let $(u_1,\ldots,u_r)$ be
$P$-independent. Then:
\begin{enumerate}
\item If $|G|\le2^S$, then
\begin{equation}\label{eq:l0-count}
r\log(P+1)\le S.
\end{equation}
\item If $(u_1,\ldots,u_r,u)$ is not $P$-independent, then
\begin{equation}\label{eq:l0-relation}
Mu=\sum_k\lambda_ku_k\qquad\text{with integers }|\lambda_k|\le PM.
\end{equation}
\end{enumerate}
\end{fact}

\begin{proof}
The $(P+1)^r$ sums $\sum_kc_ku_k$ with $c_k\in\{0,\ldots,P\}$ are distinct,
since the difference of two of them has coefficients in $[-P,P]$; this gives
\eqref{eq:l0-count}. If $(u_1,\ldots,u_r,u)$ is not $P$-independent, a
nonzero relation with coefficients in $[-P,P]$ has a nonzero coefficient on
$u$, which is some $1\le a\le P$ after a sign change, and multiplying the
relation by the integer $M/a$ gives \eqref{eq:l0-relation}.
\end{proof}

The hardness of one block comes from testing the bias of a
coin.

\begin{fact}[Testing the bias of a coin]
\label{fact:l0-coin}
Let $0<\theta\le1/8$ and let $Y\in\{-1,+1\}$ be a uniform sign. Consider an
algorithm with internal randomness independent of $Y$ that can query a biased coin, which returns a fresh independent bit equal to $1$ with probability
$\frac12+\theta Y$, and suppose that the number $K$ of queries is determined by
the internal randomness. If the algorithm outputs $Y$ with probability at
least $0.9$, then $\E K\ge\theta^{-2}/8$.
\end{fact}

\begin{proof}
Let $R$ be the internal randomness, $Z_1,\ldots,Z_K$ the answers, and
$\mathsf P_\pm$ the law of $(R,Z_1,\ldots,Z_K)$ given $Y=\pm1$. The law of
$R$ does not depend on $Y$ and determines $K$, so by the chain rule for the
Kullback--Leibler divergence and $\theta\le1/8$,
\[
D_{\mathrm{KL}}(\mathsf P_-\|\mathsf P_+)
=\E K\cdot D_{\mathrm{KL}}\bigl(\operatorname{Bern}(\tfrac12-\theta)\,\big\|\,\operatorname{Bern}(\tfrac12+\theta)\bigr)
=\E K\cdot2\theta\ln\frac{1+2\theta}{1-2\theta}
\le9\theta^2\,\E K.
\]
The output is a function of $(R,Z_1,\ldots,Z_K)$ and equals the uniform sign
$Y$ with probability at least $0.9$, so by Pinsker's
inequality~\cite[Lemma~2.6]{Weinstein15},
\[
0.8\le\dTV(\mathsf P_-,\mathsf P_+)
\le\sqrt{\tfrac12D_{\mathrm{KL}}(\mathsf P_-\|\mathsf P_+)}
\le\sqrt{4.5\,\theta^2\,\E K}.
\]
Hence $\E K\ge0.64/(4.5\,\theta^2)\ge\theta^{-2}/8$.
\end{proof}

\begin{proof}[Proof of \cref{thm:modular-l0}]
By taking the median of $O(1)$ copies, we may assume success probability
$0.99$. Let $P\coloneqq\lfloor\sqrt{\log T}\rfloor$ and $M\coloneqq P!$.
Then $\log(P+1)=\Theta(\log\log T)$ and $PM=T^{o(1)}$, so $100n^2PM\le T$.

For $n=O(1)$, it suffices to show $S=\Omega(\log\log T)$. Fix a seed for
which the sketch is correct on $0$ and $Me_1$. Then $Mv\ne0$ for
$v\coloneqq\calA(e_1)$, so $av\ne0$ for all $1\le a\le P$, since $a$ divides
$M$. Thus $(v)$ is $P$-independent, and \eqref{eq:l0-count} gives
$S\ge\log(P+1)$. From now on, $n$ is sufficiently large.

\noindent\textbf{The hard distribution.}
Let $\theta\coloneqq\max\{8\epsilon,20/\sqrt n\}\le1/8$,
$b\coloneqq\lceil100\theta^{-2}\rceil$, and let $t$ be the largest integer
with $b\sum_{i<t}4^i\le n$, so $t=\Theta(\log(1+n\theta^2))$. Split
$\{1,\ldots,bt\}$ into consecutive blocks $J_0,\ldots,J_{t-1}$ of size $b$.
Draw independently
\[
I\sim\mathrm{Unif}(\{0,\ldots,t-1\}),\quad
Y\sim\mathrm{Unif}(\{-1,+1\}),\quad
\xi_j\sim\mathrm{Unif}(\{0,1\}),\quad
U_j\sim\mathrm{Unif}(\{1,\ldots,T\})
\]
for $j\le bt$, and, conditionally on $Y$, independent
$B_j\sim\operatorname{Bern}(\frac12+\theta Y)$ for $j\le bt$. Set
\[
X_j\coloneqq
\begin{cases}
U_j, & j\in J_0\cup\cdots\cup J_{I-1},\\
\xi_jU_j+(1-\xi_j)MB_j, & j\in J_I,\\
0, & \text{otherwise}.
\end{cases}
\]
The hard input $x\in\{0,\ldots,T\}^n$ copies each $X_j$ with $j\in J_i$ to
$4^i$ coordinates and is $0$ elsewhere. Fix a seed that succeeds on $x$ with
probability $0.99$, with homomorphism $\calA$ and decoder $g$, and let $v_j$
be the sum of the columns $\calA(e_k)$ over the copies $k$ of $j$, so that
$\calA(x)=\sum_jX_jv_j$.

\noindent\textbf{The sketch determines the sign.}
Given $I$ and $Y$, $\|x\|_0=b\sum_{h<I}4^h+4^IC$ with
$C\sim\operatorname{Bin}(b,\frac34+\frac12\theta Y)$. Let
$m_I\coloneqq b\sum_{h<I}4^h+\frac34b4^I$. By Chebyshev's inequality,
$|\|x\|_0-m_I-\frac12\theta b4^IY|<\frac14\theta b4^I$ with probability at
least $0.96$, and a $(1+\epsilon)$-approximation of $\|x\|_0<\frac43b4^I$ is
within $\frac43\epsilon b4^I\le\frac16\theta b4^I$ of $\|x\|_0$. So, with
probability at least $0.95$, the estimate $g(\calA(x))$ is at least $m_I$
if $Y=+1$ and less than $m_I$ if $Y=-1$. Thus, given $I$, the sketch $\calA(x)$ determines $Y$ with probability at least $0.95$.

\noindent\textbf{The direct sum.}
Call $j$ \emph{visible} if the columns $v_d$ of the visible $d<j$ with
$\xi_d=1$, followed by $v_j$, form a $P$-independent list, and let $K$ be
the number of visible $j\in J_I$ with $\xi_j=0$. Visibility of $j$ depends
only on $\xi_1,\ldots,\xi_{j-1}$, so in expectation the visible $j\le bt$
with $\xi_j=0$ are as many as those with $\xi_j=1$. By the definition of
visibility, the columns $v_j$ of the visible $j\le bt$ with $\xi_j=1$, in
increasing order of $j$, form a $P$-independent list, so by
\eqref{eq:l0-count} there are at most $S/\log(P+1)$ such $j$. Since $I$ is
uniform and independent of $\xi$,
\[
\E K\le\frac{S}{t\log(P+1)}.
\]
We now build from the sketch an algorithm that finds $Y$ by querying the coin
only for the visible signals of block $I$, that is, with $K$ queries.

Let
\[
\widetilde\calA(x)\coloneqq\calA(x)
-\sum_{\substack{j\in J_I:\,\xi_j=0\\ j\text{ invisible}}}MB_jv_j.
\]
For each $j$ in this sum, by \eqref{eq:l0-relation},
$Mv_j$ is an integer combination, with coefficients of absolute value at most
$PM$, of the columns $v_d$ of the visible $d<j$ with $\xi_d=1$, and
$X_d=U_d$ for these $d$. So given $I$, $Y$, $\xi$, and $B$, the value
$\calA(x)$ is $\widetilde\calA(x)$ with at most $n$ independent uniform $U_d$
shifted by fixed amounts of at most $bPM$, and the two laws are at total
variation distance at most $n\cdot bPM/T\le0.01$. Hence the algorithm that
samples $I$, $\xi$, and $U$, queries the coin for the bits $B_j$ of the
visible $j\in J_I$ with $\xi_j=0$, and compares $g(\widetilde\calA(x))$ with $m_I$
finds $Y$ with probability at least $0.94$. It makes $K$ queries, so
\cref{fact:l0-coin} gives $\E K\ge\theta^{-2}/8$, and hence
$S\ge\frac18\,t\,\theta^{-2}\log(P+1)$. Finally,
$t=\Theta(\log(1+n\theta^2))$ and
$\theta^{-2}\log(1+n\theta^2)=\Theta(\epsilon^{-2}\log(1+n\epsilon^2))$,
since $\theta=8\epsilon$ if $\epsilon\ge5/(2\sqrt n)$, and otherwise both
sides are $\Theta(n)$. Therefore
\[
S=\Omega\bigl(\epsilon^{-2}\log(1+n\epsilon^2)\log\log T\bigr).
\qedhere
\]
\end{proof}

We now apply \cref{thm:modular-l0} to the sketch given by the reduction at the beginning of
\cref{app:applications-smp}.

\begin{theorem}
\label{cor:turnstile-l0}
Assume $0<\epsilon\le c_0$ and $W\ge n^{C_0}$, where $c_0>0$ is a
sufficiently small and $C_0$ a sufficiently large absolute constant. Then
$(1+\epsilon)$-approximate $L_0$ estimation requires
\[
\Omega\!\left(
\epsilon^{-2}\log(1+n\epsilon^2)\,\log\log W
\right)
\]
bits of space.
\end{theorem}

\begin{proof}
We may assume $W\ge2^{C_0}$, since otherwise $n=1$ and the claimed bound is
$O(1)$. Then $W\ge\max\{n,2\}^{C_0}\ge(2n)^{C_*}$ once $C_0\ge2C_*$.
Use the reduction at the beginning of \cref{app:applications-smp} with $N=n$,
the relation
$\mathcal R_{L_0,\epsilon}$, and
\[
\mathcal X=\{0,1,\ldots,T\}^n,\qquad T\coloneqq\lfloor D/\sqrt n\rfloor.
\]
Then $\mathcal X\subseteq B_D\cap\mathbb Z_{\ge0}^n$, and the reduction gives
a public-coin linear sketch with image size $2^{O(S)}$ that succeeds with
probability at least $0.997$ on every $x\in\mathcal X$. Since $n^2\le D$,
we have $D^{3/4}/2\le T\le D$, hence $\log T=\Theta(\log W)$, and
$T\ge Cn^4$ for the constant $C$ of \cref{thm:modular-l0} once $C_*$ is
sufficiently large. With $c_0=1/128$, \cref{thm:modular-l0} gives
\[
S=\Omega\bigl(\epsilon^{-2}\log(1+n\epsilon^2)\,\log\log T\bigr)
 =\Omega\bigl(\epsilon^{-2}\log(1+n\epsilon^2)\,\log\log W\bigr).
\qedhere
\]
\end{proof}

\subsubsection{Estimating maximum matching size}

A frequency vector $x\in\mathbb Z_{\ge0}^N$ with $N=\binom n2$ encodes the
multigraph $G_x$ on $n$ vertices in which each pair $\{u,v\}$ of vertices
has $x_{\{u,v\}}$ parallel edges, and $\mathsf{OPT}(G_x)$ denotes the maximum
matching size of $G_x$. For $\beta\ge1$, estimating $\mathsf{OPT}$ within a
factor $\beta$ is the relation
\[
\mathcal R_{\mathsf{OPT},\beta}
\coloneqq\{(x,z):\mathsf{OPT}(G_x)/\beta\le z\le\mathsf{OPT}(G_x)\}.
\]
We use the public-coin lower bounds of Assadi,
Khanna, and Li~\cite[Section~2.3 and Theorems~8,~10]{AssadiKL17} for
protocols in the SMP model, which bound the total communication.

\begin{theorem}
\label{cor:turnstile-matching-size}
Assume $W\ge n^{C_0}$, where $C_0$ is a sufficiently large absolute constant.
On strict turnstile streams whose final frequency vector encodes a multigraph
on $n$ vertices, estimating the maximum matching size $\mathsf{OPT}$ within
$\mathsf{OPT}/\alpha\le z\le\mathsf{OPT}$ requires $\Omega(n/\alpha^2)$ bits
of space for $C_0\le\alpha\le\sqrt n$. For every fixed sufficiently small
$\epsilon>0$, a $(1+\epsilon)$-approximation requires
$\Omega(n^{2-O(\epsilon)})$ bits of space.
\end{theorem}

\begin{proof}
Use the reduction at the beginning of \cref{app:applications-smp} for strict
streams with $N=\binom n2$ and
$\mathcal X=B_D\cap\mathbb Z_{\ge0}^N$. Since $2N\le n^2$, the assumption
$W\ge n^{C_0}$ gives $W\ge(2N)^{C_*}$ once $C_0\ge2C_*$. Use the relation
$\mathcal R_{\mathsf{OPT},\beta}$ with $\beta=\alpha$ for
$\alpha$-approximation and $\beta=(1+\epsilon)^2$ for
$(1+\epsilon)$-approximation.

The hard instances of~\cite{AssadiKL17} have at most $n$ player
graphs and one referee graph, each simple, so the sum $x$ of all inputs
satisfies $\|x\|_2\le\|x\|_1\le n^3\le D$, and hence $x\in\mathcal X$.
In~\cite{AssadiKL17}, the referee must estimate the maximum matching size of
the union of these graphs. Since the inputs are nonnegative, this union has
the same edges as $G_x$, so its maximum matching size is $\mathsf{OPT}(G_x)$,
and the problem has the form required above.

The lower bounds of~\cite{AssadiKL17} hold for public-coin protocols in the
SMP model, such as the protocol at the beginning of
\cref{app:applications-smp}. Let $k$ be the number of players in
their hard instances, each sending $b=O(S)$ bits. For $\alpha$-approximation,
\cite[Theorem~8]{AssadiKL17} gives $kb=\Omega(nk/\alpha^2)$, hence
$S=\Omega(n/\alpha^2)$. For the $(1+\epsilon)$ bound,
\cite[Theorem~10]{AssadiKL17} gives $kb=\Omega(n^{2-O(\epsilon)})$ with
$k=n^{o(1)}$. Thus $S=\Omega(n^{2-O(\epsilon)-o(1)})
=\Omega(n^{2-O(\epsilon)})$ for each fixed sufficiently small $\epsilon>0$.
\end{proof}

\subsubsection{Subgraph counting}

We use the composable-state lower bound of Kallaugher, Kapralov, and
Price~\cite[Theorem~19]{KallaugherKP18}. For a hypergraph
$J=(V,E)$ and $\lambda\ge 0$, the fractional vertex-cover variant $\mathsf{MVC}_\lambda$ is defined by
\[
\mathsf{MVC}_\lambda(J)
\coloneqq
\min\left\{
\sum_{v\in V} f(v)+\lambda\sum_{e\in E} f(e):
f:V\cup E\to[0,\infty),\ \sum_{v\in e} f(v)+f(e)\ge 1\ \forall e\in E
\right\}.
\]
When $\lambda\ge 1$ and $J$ has no empty hyperedges, this is the usual
fractional vertex cover number. A streaming algorithm has \emph{composable
state} if, for every fixed seed, its state after the concatenation of two
streams is a fixed function of its states after each of
them~\cite[Definition~18]{KallaugherKP18}.

For $n_0,q\in\mathbb N$, a frequency vector $x\in\mathbb Z_{\ge0}^N$ with
$N=\sum_{j=1}^q\binom{n_0}j$ encodes the hypergraph $G_x$ on $n_0$ vertices
in which each nonempty set $e$ of at most $q$ vertices is a hyperedge of
multiplicity $x_e$. For a hypergraph $H$, $T>0$, and $\epsilon\in(0,1]$,
distinguishing hypergraphs with at least $T$ copies of $H$ from those with
at most $(1-\epsilon)T$ copies is the relation
\[
\mathcal R_{H,T,\epsilon}\coloneqq\{(x,z):z=1\text{ if }G_x\text{ has at least
}T\text{ copies of }H,\ z=0\text{ if it has at most }(1-\epsilon)T\}.
\]

\begin{theorem}
\label{cor:turnstile-subgraph-counting}
Fix a connected hypergraph $H$ with $|E(H)|>1$, and set
$\mu_2\coloneqq\mathsf{MVC}_{1/2}(H)$ and $\mu_1\coloneqq\max_{e\in E(H)} \mathsf{MVC}_1(H\setminus e)$.
Assume $m,T\in\mathbb N$, $\epsilon\in(1/\sqrt{T},1]$, $m\ge C_H T$, and
$W\ge m^{C_H}$, where $C_H>0$ is a sufficiently large constant depending only
on $H$. Then, on strict turnstile streams, distinguishing hypergraphs with at
most $m$ edges and at least $T$ copies of $H$ from those with at most
$(1-\epsilon)T$ copies requires
\[
\Omega_H\!\left(
\max\left\{
\frac{m}{(\epsilon T)^{1/\mu_2}},
\frac{m}{(\epsilon^2 T)^{1/\mu_1}}
\right\}
\right)
\]
bits of space.
\end{theorem}

\begin{proof}
Use the reduction at the beginning of \cref{app:applications-smp} for strict
streams with $n_0=\lceil C_Hm\rceil$ and $q=\max_{e\in E(H)}|e|$, the
relation $\mathcal R_{H,T,\epsilon}$, and $\mathcal X$ the set
of $x\in\mathbb Z_{\ge0}^N$ with $\|x\|_1\le m$, which contains the frequency
vectors of all promised hypergraphs. Since $H$ is fixed, $N=m^{O_H(1)}$, so a
sufficiently large $C_H$ gives $W\ge(2N)^{C_*}$, and
$\|x\|_2\le\|x\|_1\le m\le D$ gives $\mathcal X\subseteq B_D$.

For every seed, the sketch $\mathcal A$ has composable state, since
$\mathcal A(x_1+x_2)=\mathcal A(x_1)+\mathcal A(x_2)$, and its state takes
$b=O(S)$ bits. Thus \cite[Definition~18 and Theorem~19]{KallaugherKP18}
apply to give
\[
b=\Omega_H\!\left(
\max\left\{\frac{m}{(\epsilon T)^{1/\mu_2}},
             \frac{m}{(\epsilon^2T)^{1/\mu_1}}\right\}\right).
\]
Since $b=O(S)$, the same bound holds for $S$.
\end{proof}

\bibliographystyle{alpha}
\bibliography{main}

\end{document}